\documentclass[12pt,a4paper]{article}

\usepackage[english]{babel}

\usepackage[a4paper,top=2cm,bottom=2cm,left=2.5cm,right=2.5cm]{geometry} 

\usepackage[style=authoryear-comp, backend=biber, sortcites=false]{biblatex}
\usepackage{amsmath}
\allowdisplaybreaks
\usepackage{amssymb}
\usepackage{amsthm}
\usepackage{bbm}
\usepackage{graphicx}
\usepackage[colorlinks=true, allcolors=blue]{hyperref}
\usepackage[title]{appendix}
\usepackage{mathrsfs}
\usepackage{amsfonts}
\usepackage{caption}
\usepackage{authblk}
\usepackage{csquotes}
\usepackage[T1]{fontenc} 
\usepackage{lipsum}
\usepackage{float}

\DeclareMathOperator*{\argmax}{argmax} 
\DeclareMathOperator*{\argmin}{argmin} 
\usepackage{setspace}
\usepackage{titlesec}
\titleformat{\section} 
  {\normalfont\Large\bfseries}{\thesection.}{1em}{}

\theoremstyle{plain}
\newtheorem{theorem}{Theorem}
\newtheorem{lemma}[theorem]{Lemma}
\newtheorem{proposition}[theorem]{Proposition}

\theoremstyle{definition}
\newtheorem{definition}{Definition}
\newtheorem{example}{Example}

\title{Stochastic Non-Tâtonnement Processes and the Attraction Principle}

\author[1,2]{Leandro Lyra Braga Dognini\thanks{I am grateful to Aloisio Araujo and the participants of the Mathematical Economics Seminar at the Institute for Pure and Applied Mathematics (IMPA) and the 25th Annual SAET Conference for all insightful comments on this manuscript. The paper has also benefited greatly from the comments of anonymous referees. All figures were drawn and the revision was done with the assistance of Claude, Anthropic's large language model.}}
\affil[1]{\small Department of Economics, Rio de Janeiro State University}
\affil[2]{\small Legislative Advisory, Federal Senate of Brazil}
\date{\today} 

\begin{document}
\maketitle
\begin{abstract}
\noindent I introduce stochastic non-tâtonnement processes and argue that they are a natural outcome of general equilibrium theory. To this end, I revisit classical demand theory to define two diffeomorphisms: the normalized Walrasian demand and a map that flattens indifference hypersurfaces. These diffeomorphisms are applied to the three canonical manifolds in the consumption domain (i.e., the indifference and the offer hypersurfaces, and the trade hyperplane) to analyze their images in the normalized and flat domains, and their relations to the set of Pareto optimal allocations are discussed. Then, I derive, for arbitrary non-tâtonnement processes, an attraction principle based on the dynamics of marginal rates of substitution seen on the floor of the flat domain. This principle motivates the definition of stochastic non-tâtonnement processes and leads to a stochastic version of the First Welfare Theorem. When utility functions are attractive and sharp, a class of stochastic non-tâtonnement processes, called common density processes, is particularly well suited for Monte Carlo simulations, which can be used to compute the distribution of trade outcomes on the contract curve and to model price stickiness.

\end{abstract}



\section{Introduction}\label{sec1}

\textsc{general equilibrium theory} is one of the cornerstones of contemporary economics. When it comes to the dynamic aspect of this theory\footnote{\textit{Dynamic} in this context meaning ``a theory which determines the behavior through time of all variables from arbitrary initial conditions'' \parencite[p. 100]{Samuelson_1941}, in contrast to \textit{static} aspects dealing with ``existence, uniqueness, and optimality [of a competitive equilibrium]'' \parencite[p. 522]{ArrowHurwicz_1958}.} and to the classical question of how an equilibrium is attained, two streams have emerged in the literature.

The first stream developed from the tâtonnement processes conceived by Walras when defining the ``loi d'établissement des prix d'équilibre'' \parencite[p. 69]{Walras_1874}\footnote{See, for instance, this excerpt from \textcite[p. 69]{Walras_1874}: ``Ceci nous amène à formuler en ces termes la \textit{loi de l'offre et de la demande effectives} ou \textit{loi d'établissement des prix d'équilibre}, [...] il faut, pour arriver au prix d'équilibre, une hausse du prix de la marchandise dont la demande effective est supérieure à l'offre effective, et une baisse du prix de celle dont l'offre effective est supérieure à la demande effective.''}. These processes are defined by two basic features: (i) a function of aggregate excess demand (typically, though not necessarily\footnote{The dynamics of prices in tâtonnement processes is traditionally attributed to an auctioneer who follows the rule: ``never change any price if the given price vector is an equilibrium; never leave all prices unchanged if the given price vector is not an equilibrium'' \parencite[p. 266]{ArrowHahn_1971}. Clearly, this rule does not imply that prices are determined by a function of aggregate excess demand.}) drives the dynamics of prices; and (ii) no trade happens until an equilibrium is reached (i.e., a price at which all markets clear is attained)\footnote{Or, as stated by \textcite[p. 746]{Hahn_1982}, ``no actual trading or production at `false' prices was allowed in the process.''}.     

The second stream developed from the barter processes conceived by Edgeworth when describing that ``equilibrium is attained when the existing contracts can neither be varied without recontract with the consent of the existing parties, nor by recontract within the field of competition'' \parencite[p. 31]{Edgeworth_1881}\footnote{See, for instance, this excerpt from \textcite[pp. 21--22, 31]{Edgeworth_1881}: ``[...] which locus it is here proposed to call the \textit{contract-curve}. [...] it is evident that X will step only on one side of a certain line, the \textit{line of indifference}, as it might be called. [...] If then we enquire in what directions X and Y will consent to move \textit{together}, the answer is, in any direction between their respective lines of indifference, in a direction \textit{positive} as it may be called \textit{for both}. [...] Equilibrium is attained when the existing contracts can neither be varied without recontract with the consent of the existing parties, nor by recontract within the field of competition.''}. Therefore, the distinctive feature of these barter or non-tâtonnement processes is that trade may happen even though the economy is out of equilibrium. These Walrasian and Edgeworthian streams are discussed in more detail in Subsection \ref{subsec11}. 

This paper builds on the Edgeworthian stream to define, in pure exchange economies, stochastic non-tâtonnement processes motivated by an attraction principle that describes the dynamics of households' marginal rates of substitution when trade unfolds. The basic intuition can be straightforwardly obtained from Figure \ref{fig:edgeworth} (cf. \textcite[p. 316]{Edgeworth_1925}), where a non-tâtonnement process is depicted in the Edgeworth box with trade unfolding along the line joining $\omega_{0}\in\mathbb{R}^{4}_{++}$ to $\omega_{1}$, and then along the line joining $\omega_{1}$ to $\omega_{2}$. 

Along the first segment, trade unfolds according to prices $p\in\mathbb{R}^{2}_{++}$ and the marginal rates of substitution come closer to each other as one moves along the line, with $p_{1}/p_{2}>0$ always lying between them. Along the second segment, trade unfolds according to prices $(1,1)$ until the contract curve is reached and, therefore, no further trade remains possible. When moving along this line, the marginal rates of substitution come step by step closer to each other, thus being, once again, \textit{attracted} by the normalized ongoing trade price.  

\begin{figure}[h]
    \centering
    \includegraphics[width=\linewidth]{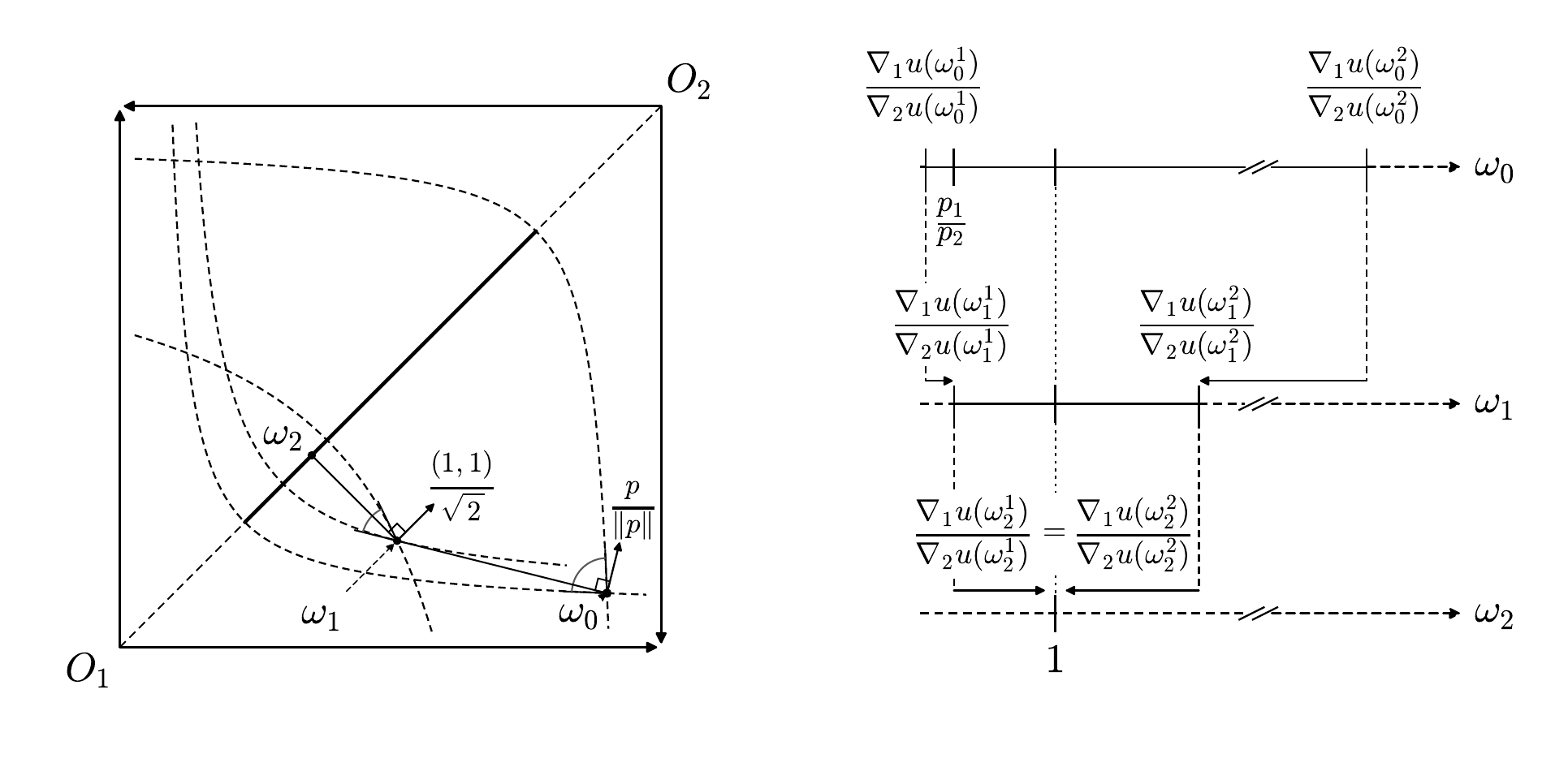}
    \caption{Non-tâtonnement process in the Edgeworth box (left) with trade unfolding according to two prices (first $p\in \mathbb{R}^{2}_{++}$ and then $(1,1)$), and the corresponding dynamics of the marginal rates of substitution (right). The bold line is the contract curve.}
    \label{fig:edgeworth}
\end{figure}

To provide a general statement for this attraction principle and to define stochastic non-tâtonnement processes, I first need to revisit classical demand theory to define two diffeomorphisms: the normalized Walrasian demand $x_{n}:\mathbb{R}^{L}_{++}\rightarrow\mathbb{R}^{L}_{++}$ and the flattening diffeomorphism $f:\mathbb{R}^{L}_{++}\rightarrow\mathbb{R}^{L}_{++}$. The normalized Walrasian demand is obtained when prices are stated as percentages of wealth (i.e., wealth is normalized to one), and the flattening diffeomorphism is the inverse of the Hicksian demand function after a suitable price normalization.

These diffeomorphisms are analytical tools that allow us to observe non-tâtonnement processes\footnote{Although conceived in this paper to deal with non-tâtonnement processes, these diffeomorphisms lead to other theoretical developments. For instance, they can be used to measure, through depth projections, how much more bent the offer hypersurface is relative to the indifference hypersurface and to reach a more general statement for the \textcite{Cass_1972} criterion \parencite{Dognini_2026d}. This general statement, in turn, is closely related to the existence of monetary equilibria — the \textcite{Hahn_1965} problem; see, e.g., \textcite[p. 1486]{Bewley_1983} — in consumption-loan overlapping generations economies \parencite{Dognini_2026a,Dognini_2026c}.} and, more generally, manifolds in the consumption domain through their counterparts (i.e., images) in the normalized and flat domains\footnote{These three \textit{domains} are copies of $\mathbb{R}^{L}_{++}$, each with a specific economic interpretation. They allow us to analyze non-tâtonnement processes from three different but equivalent perspectives.}. In particular, these diffeomorphisms transform the canonical manifolds in the consumption domain (i.e., the indifference and the offer hypersurfaces and the trade hyperplane) in a well-behaved manner (e.g., the flattening diffeomorphism makes indifference hypersurfaces flat).

The diffeomorphisms also lead to a contribution of this paper that is actually more related to general equilibrium theory than to non-tâtonnement processes. The contribution highlights that the definition of Pareto optimality can be made without any \textit{explicit} constraints on the aggregate resources of the economy (which is traditionally stated as a feasibility criterion), and this definition leads, as a direct consequence of the First and Second Welfare Theorems, to the characterization of the set of Pareto optimal allocations as a smooth manifold without boundary of dimension $H+L-1\geq3$.

Once these diffeomorphisms are defined and their basic properties established, I can state the attraction principle (see Proposition \ref{propAttractionPrinciple}) for arbitrary non-tâtonnement processes. This can be seen as a statement that bridges the accessibility of Pareto optima and the flattening diffeomorphism, since the existence of trade paths that asymptotically approach a given Pareto optimum implies that the households' marginal rates of substitution come closer to each other.

This attraction principle leads to the definition of stochastic non-tâtonnement processes (also called ``Edgeworth processes,'' after the ``Edgeworth's barter processes'' of \textcite{Uzawa_1962}; see Subsection \ref{subsec11}) and to the derivation of a stochastic version of the First Welfare Theorem. 

I also identify two properties of utility functions (namely, attractiveness and sharpness) that are particularly useful when applying Monte Carlo methods to calculate stochastic outcomes — and hence distributions and confidence regions — of common density processes, a well-behaved class of Edgeworth processes, on the contract curve. In particular, these common density processes can be used to model price stickiness in a different way from the New Keynesian literature.

All the definitions, results, and examples are then gathered and jointly presented as an argument for why stochastic non-tâtonnement processes are a natural outcome of general equilibrium theory. After this initial overview, the following subsection provides a review of the literature and positions this paper in relation to it.

\subsection{Related literature}\label{subsec11}
The conception of the tâtonnement process by \textcite{Walras_1874} was, from its beginning, tied to the fundamental question of the stability of equilibria\footnote{See, for instance, this excerpt from \textcite[pp. 73--74]{Walras_1874}: ``Dans les deux cas, \textit{au-delà} du point d'équilibre, \textit{l'offre de la marchandise est supérieure à sa demande}, ce qui doit amener une \textit{baisse} de prix, c'est-à-dire un retour vers le point d'équilibre. [...] C'est un équilibre \textit{stable}. [...] Ainsi, dans ce cas, \textit{au-delà} du point d'équilibre, \textit{la demande de la marchandise est supérieure à son offre}, ce qui doit amener une \textit{hausse} de prix, c'est-à-dire un éloignement du point d'équilibre. [...] C'est un équilibre \textit{instable}.''}: whether the process would actually reach an equilibrium from a given initial price. A thoughtful discussion of this stability matter can be found in \textcite[pp. 62--77]{Hicks_1939}\footnote{See, for instance, this excerpt from \textcite[p. 62]{Hicks_1939}: ``What is meant by stability in exchange? In order that equilibrium should be stable, it is necessary that a slight movement away from the equilibrium position should set up forces tending to restore equilibrium. This means that a rise in price above the equilibrium level must set up forces tending to produce a fall in price; which implies, under perfect competition, that a rise in price makes supply greater than demand. The condition of stability is that a rise in price makes supply greater than demand, a fall in price demand greater than supply.''}. However, a systematic treatment only arrived after \textcite[p. 102]{Samuelson_1941} introduced the well-known differential equation 
\begin{eqnarray}\label{eqSamuelsonAuctioneer}
    \dot{p}=z(p),
\end{eqnarray}
with $z(\cdot)$ the aggregate excess demand function, to explicitly model the price updating rule of the Walrasian auctioneer. 

Once the distinction between local and global stability was established, i.e., whether the tâtonnement process converges to a given equilibrium price for a local or for an arbitrary displacement (e.g., \textcite[pp. 523--524]{ArrowHurwicz_1958}), this Walrasian stream of the literature evolved\footnote{Comprehensive surveys on the local and global stability of tâtonnement processes can be found in \textcite{Negishi_1962}, \textcite[pp. 263--323]{ArrowHahn_1971}, and \textcite{Hahn_1982}, which also address the work done by \textcite[pp. 486--489]{Allais_1925}, \textcite{Nikaido_1959} and \textcite{McKenzie_1960}, among others.} by searching for results on fundamentals that would ensure either type of stability for tâtonnement processes driven by (\ref{eqSamuelsonAuctioneer}). For instance, in the gross substitutes case, \textcite{Hahn_1958,Negishi_1958, ArrowHurwicz_1958} proved local stability, and \textcite{ArrowHurwicz_1959} proved global stability.

\textcite{Scarf_1960} later provided examples of global instability, and after the Sonnenschein--Mantel--Debreu Theorem \parencite{Sonnenschein_1972,Mantel_1974,Debreu_1974}, it became clear that processes based on (\ref{eqSamuelsonAuctioneer}) required stronger assumptions (e.g., gross substitutes) than those traditionally adopted by general equilibrium theory to ensure stability. 

The literature then started looking for variants of the price updating rule (\ref{eqSamuelsonAuctioneer}). For instance, \textcite{Hahn_1961} derived a globally stable price adjustment process based on the ratio of excess demand to demand itself. By adopting a variant of (\ref{eqSamuelsonAuctioneer}) that takes into account the Jacobian of the excess demand function, \textcite{Smale_1976c} developed a price updating rule (in this case, as a specific instance of a global Newton method) that ensured global stability outside a measure zero set. More recently, \textcite{Werning_2026} derived local stability results for a tâtonnement process based on Frisch demands and forward-looking price setting.

In turn, the Edgeworthian stream of the literature (to which this paper belongs) first received a systematic treatment in the early 1960s through a sequence of papers by \textcite{Negishi_1961,Hahn_1962,Morishima_1962,Uzawa_1962,HahnNegishi_1962}. These papers established a common, though not indispensable, assumption\footnote{These papers also make the intuitive assumptions that the aggregate endowment of the economy is constant and that trade must be Pareto improving.} guiding non-tâtonnement processes: all trade occurs according to common current prices. Furthermore, a process in discrete time that satisfies this assumption was called \textit{Edgeworth's barter process} by \textcite{Uzawa_1962}, and I will simply refer to it here as an \textit{Edgeworth process}\footnote{There is a reason for this: the term ``barter'' is, I believe, more general than the setting in which there is a prevailing price at every moment, since barter need not take place under common terms of trade.}.    

Under this assumption, the non-tâtonnement processes in these papers were essentially described by two rules. First, a price updating rule for the Edgeworthian auctioneer (unlike the Walrasian auctioneer, this one does not prevent out-of-equilibrium trade from happening). Second, a rule for updating the allocation of the economy given the current allocation and price (i.e., representing the transactions that take place). In differential form\footnote{See, for instance, \textcite[p. 223]{Uzawa_1962} for a discrete-time formulation of these rules.}, an example of these rules, adapted from \textcite[pp. 326--333]{ArrowHahn_1971}, is given by       
\begin{eqnarray}\label{eqNonTatoPrice}
    \begin{bmatrix}
        \dot{p}\\
        \dot{\omega}
    \end{bmatrix}=\begin{bmatrix}
        z(p,\omega)\\
        T(p,\omega)
    \end{bmatrix},
\end{eqnarray}
with $z(\cdot)$ the aggregate excess demand, $\omega(\cdot)$ the allocation, and $T(\cdot)$ a function representing a \textit{transaction rule} \parencite[p. 220]{Uzawa_1962}. In particular, notice that the allocation $\omega(\cdot)$ enters as an argument of the aggregate excess demand function $z(\cdot)$ in (\ref{eqNonTatoPrice}) and that, if no transaction is allowed by the auctioneer (i.e., $T(\cdot)$ is identically zero), then (\ref{eqNonTatoPrice}) becomes (\ref{eqSamuelsonAuctioneer}).

The main concern of the aforementioned papers from the sixties was, as in the tâtonnement stream, the stability question, which can be posed in this context as follows: given a solution $(\tilde{p},\tilde{\omega})$ to (\ref{eqNonTatoPrice}), under which conditions does $\lim_{t\rightarrow\infty} (\tilde{p}(t),\tilde{\omega}(t))=(p^{*},\omega^{*})$ hold, with this limit being an equilibrium? For instance, \textcite{Negishi_1961} demonstrated stability in the case of gross substitutes and \textcite{HahnNegishi_1962} under a transaction rule that satisfies a given postulate\footnote{The postulate is that ``if a good is in excess demand at the going prices then trading ensures that no one will hold more of that good than he desires and vice versa when there is excess supply'' \parencite[p. 463]{HahnNegishi_1962}.}.

However, the definition of the non-tâtonnement process through the transaction rule $T(\cdot)$ and the differential equation (\ref{eqNonTatoPrice}) was by no means a straightforward or uncontroversial matter. As noted by \textcite[p. 333]{ArrowHahn_1971}, ``it is not easy to state a sensible [transaction] rule for the rate at which exchange is taking place between households at any one moment of time. There may be a considerable variety of exchanges possible at any [price] and [allocation], and it is not clear what simple rule of choice is appropriate.''

The literature then began to look for alternative formulations for non-tâtonnement processes and the work by \textcite{Smale_1976a} can be seen as a landmark of this turning point\footnote{See, for instance, the following excerpt from \textcite[p. 211]{Smale_1976a}: ``It is also the case that our treatment has a number of new and different features from the [Uzawa, Hahn, Hahn and Negishi, Morishima] 1962 work. For one thing, our economic process is not supposed to satisfy some ordinary differential equation. Such a hypothesis implies too deterministic a setting for a state of an economy moving in time. Economics seems to be different from theoretical physics where the state at one time determines the state at all future times. Economic forces don't seem to have such a Newtonian flavor and one has to be careful about carrying over methods from physics to economics. Another consideration is that there is no (apparent, at least) economic principle which gives us a differential equation.''}.

Instead of analyzing the solutions of (\ref{eqNonTatoPrice}), \textcite{Smale_1976a} looked directly at a whole set of non-tâtonnement processes that satisfy the \textit{exchange} and the \textit{price-adjustment axioms}. The exchange axiom consists simply of the intuitive restrictions on non-tâtonnement processes and can be stated as: (i) total resources are constant; (ii) exchange takes place at common prices; (iii) trade is utility-enhancing; (iv) if trade can happen, it happens.

The price-adjustment axiom can be stated as requiring
\begin{eqnarray}\label{eqSmale2Axiom}
    \dot{p}\in\biggl\{\sum_{1\leq h\leq H} \alpha_{h}d^{h}(x^{h},p)\mid \varepsilon<\alpha_{h}<\varepsilon^{-1},\forall 1\leq h\leq H\biggr\},
\end{eqnarray}
for some $\varepsilon>0$, with $d^{h}(\cdot)$ the \textit{short-run demand} of household $h\in\{1,\ldots,H\}$. The fundamental remark about (\ref{eqSmale2Axiom}) is that it encompasses several possible dynamics (i.e., a whole set of non-tâtonnement processes) and, therefore, circumvents the discussion on which specific system (\ref{eqNonTatoPrice}) is the appropriate one. In particular, \textcite[p. 211]{Smale_1976a} proved that, under a condition that can be seen as related to transaction costs, non-tâtonnement processes satisfying both axioms exist and converge to a Pareto optimal allocation. More recently, \textcite{YuHosoya_2022} were able to obtain an analogous result under different axioms, with the main change being the shift from the price-adjustment axiom (\ref{eqSmale2Axiom}) to the \textit{continuous adjustment} axiom \parencite[p. 4, Definition 3]{YuHosoya_2022}.

Another formulation of non-tâtonnement processes, named \textit{Walrasian exchange processes}, was conceived by \textcite{CornetChampsaur_1990}. These processes model the rule of price adjustment and trade through the Walrasian equilibria of a conveniently defined \textit{marginal economy}. In this case too, one allows for a whole set of exchange processes to be analyzed under a common framework, and the main result states that Walrasian exchange processes exist and converge to a Pareto optimal allocation.  

In addition to the question of stability, the matter of accessibility (i.e., whether a particular Pareto optimum can be arbitrarily approached by a non-tâtonnement process) also drove a fruitful research branch by \textcite{Schecter_1977}, \textcite{Cornet_1977,Cornet_2024} and \textcite{Bottazzi_1994}, among others. 

At this point, the contributions of this paper can be clearly related to the existing results in the Edgeworthian stream of the literature. First, this paper provides a statement for the attraction principle (depicted on the right of Figure \ref{fig:edgeworth}) that is valid for a general Edgeworth process defined in the deterministic setting of \textcite{Uzawa_1962}. Fundamentally, the principle states that, if utility functions are attractive, the ongoing price attracts the households' marginal rates of substitution as trade unfolds.

Second, this paper recasts the definition of these Edgeworth processes as stochastic processes in discrete time\footnote{Although \textcite{Smale_1976a} deals with non-tâtonnement processes in continuous time, Section 5 of that paper indicates ``briefly and tentatively how a discrete version of the main result could go'' \parencite[p. 224]{Smale_1976a}.} (i.e., a sort of random walk), thus circumventing the too deterministic setting \parencite[p. 211]{Smale_1976a} of (\ref{eqNonTatoPrice}). The economic rationale behind this stochasticity is discussed in Section \ref{sec6}.

Third, this paper identifies a class of these stochastic non-tâtonnement processes, namely the \textit{common density processes}, that are particularly suited for Monte Carlo simulations, which can be used, for instance, to model price stickiness in this non-tâtonnement framework. Furthermore, this form of modeling price stickiness is fundamentally different from those found in the literature on menu costs (e.g., \textcite{Mankiw_1985,GolosovLucas_2007}), search models (e.g., \textcite{BurdettJudd_1983,HeadLiuMenzioWright_2012}), disequilibrium models (e.g., \textcite{BarroGrossman_1971,Dreze_1975,Benassy_1975}), and the New Keynesian models (e.g., \textcite{Calvo_1983,Clarida_1999,Woodford_2003,Christiano_2005,Gali_2015}).

Finally, in order to obtain the aforementioned results, this paper first needs to define diffeomorphisms (i.e., the normalized Walrasian demand and the flattening diffeomorphism) and study manifolds that naturally emerge from classical demand theory. These foundational definitions and results can, therefore, be seen as lying in the intersection of differential topology and general equilibrium theory, which is a field of long-standing sound research, with works by \textcite{Debreu_1970}, \textcite{Mas-Colell_1985}, and \textcite{Balasko_1975,Balasko_1979,Balasko_1988,Balasko_2009}, among others.

\subsection{Outline of the paper}\label{subsec12}

The remainder of the paper is organized as follows. The next three sections are dedicated to the definitions and results regarding demand theory, differential topology, and general equilibrium. Section \ref{sec2} defines and analyzes the normalized Walrasian demand and the flattening diffeomorphism. Section \ref{sec3} defines the canonical manifolds in the consumption domain and characterizes their counterparts in the normalized and flat domains. Section \ref{sec4} then characterizes the set of Pareto optimal allocations as a smooth manifold and discusses its relation to general equilibrium theory. 

Then, the following three sections are devoted to non-tâtonnement processes. Section \ref{sec5} focuses on the derivation of the attraction principle for a general non-tâtonnement process. It also defines attractive utility functions, which lead to a more specific statement of this principle. 

Section \ref{sec6} turns to the definitions, results, and examples related to stochastic non-tâtonnement processes. It starts with the definition of these processes and proves a stochastic version of the First Welfare Theorem. Furthermore, it defines sharp utility functions and identifies a specific class of processes, called common density processes, that is particularly well suited for Monte Carlo simulations. It also discusses how these common density processes can be used to model price stickiness.

Finally, Section \ref{sec7} gathers all previous results and argues that stochastic non-tâtonnement processes are a natural outcome of general equilibrium theory, and Section \ref{sec8} concludes. All omitted proofs are given in the \hyperref[appx]{Appendix}.

\section{Normalized Walrasian demand and the flattening diffeomorphism}\label{sec2}

The consumption set\footnote{I follow \textcite{Mas-ColellWhinstonGreen_1995} for classical demand theory definitions and results.} is $X=\mathbb{R}^{L}_{+}$, $L\geq2$, and preferences are defined through a smooth\footnote{By \textit{smooth}, I mean continuous in $\mathbb{R}^{L}_{+}$ and with continuous derivatives of any order in $\mathbb{R}^{L}_{++}$.}, strictly increasing, and strictly quasi-concave\footnote{Utility functions for which the indifference hypersurface of the zero utility level is $\partial \mathbb{R}^{L}_{+}$ (e.g., Cobb--Douglas) require a caveat, since for two bundles on this boundary, any convex combination also has zero utility.} utility function $u:\mathbb{R}^{L}_{+}\rightarrow\mathbb{R}_{+}$, with $u(0)=0$, $\nabla u(c)\in\mathbb{R}^{L}_{++}$, $c\in\mathbb{R}^{L}_{++}$, and $u(\cdot)$ unbounded above. The budget set correspondence is $B(p,w)=\{c\in\mathbb{R}^{L}_{+}\mid p \cdot c\leq w\}$, $(p,w)\in\mathbb{R}^{L+1}_{++}$. The Walrasian demand function\footnote{It is a function, and not a correspondence, due to the strict quasi-concavity assumption on $u(\cdot)$ \parencite[p. 52, Proposition 3.D.2]{Mas-ColellWhinstonGreen_1995}.} $x:\mathbb{R}^{L+1}_{++}\rightarrow \mathbb{R}^{L}_{+}$ is defined through 
\begin{eqnarray}\label{UMP}
    x(p,w)=\argmax_{c\in\mathbb{R}^{L}_{+}} u(c) \textrm{ s.t. } c\in B(p,w),
\end{eqnarray}
and I assume $x(\cdot)$ is smooth\footnote{As a remark, it is well known (e.g., \textcite[pp. 94--95]{Mas-ColellWhinstonGreen_1995}) that the differentiability of $x(\cdot)$ is ensured by a nonzero determinant of the bordered Hessian of $u(\cdot)$.} and $x(p,w)\in\mathbb{R}^{L}_{++}$, $(p,w)\in\mathbb{R}^{L+1}_{++}$ (i.e., (\ref{UMP}) has no corner solutions\footnote{I prefer to state this assumption rather than impose conditions that ensure it. For instance, if the utility functions are strongly monotonic, satisfy a strong bordered Hessian condition, and the indifference hypersurfaces are closed in $\mathbb{R}^{L}$, then this assumption is satisfied (e.g., \textcite[p. 2]{YuHosoya_2022}).}).

\begin{definition}\label{defNormalizedWalrasian}
The \textit{normalized Walrasian demand} $x_{n}:\mathbb{R}^{L}_{++}\rightarrow \mathbb{R}^{L}_{++}$ is given by $x_{n}(p)=x(p,1)$.
\end{definition}
The normalized Walrasian demand represents the household demand when prices are stated as percentages of wealth. It gathers all information about household consumption behavior since $x_{n}(p/w)=x(p,w)$, $(p,w)\in\mathbb{R}^{L+1}_{++}$, and can be seen as representing the fact that, from a households' perspective, consumption decisions only depend on its purchasing power.

Let $\lambda(p,w)>0$ be the Lagrange multiplier of (\ref{UMP}) and the indirect utility function $v:\mathbb{R}^{L+1}_{++}\rightarrow \mathbb{R}_{++}$ be given by $v(p,w)=u(x(p,w))$. Following Definition \ref{defNormalizedWalrasian}, let $\lambda_{n}(p)=\lambda(p,1)$ and $v_{n}(p)=u(x_{n}(p))$, $p\in\mathbb{R}^{L}_{++}$. The Hicksian demand function $h:\mathbb{R}^{L+1}_{++}\rightarrow\mathbb{R}^{L}_{++}$ is
\begin{eqnarray*}
    h(p,u)=\argmin_{c\in\mathbb{R}^{L}_{+}} p\cdot c \textrm{\ s.t.\ } u(c)\geq u,
\end{eqnarray*}
and the expenditure function $e:\mathbb{R}^{L+1}_{++}\rightarrow\mathbb{R}_{++}$ is $e(p,u)=p \cdot h(p,u)$. The following lemma rewrites well-known identities regarding the functions defined above\footnote{Throughout the paper, vectors are row vectors, $y^{T}$ denotes the transpose of $y$, and $x\cdot y=xy^{T}$ is the usual inner product; for a smooth function $g(\cdot)$, $\nabla g(\cdot)$, $\mathbf{J}g(\cdot)$ and $\mathbf{H}g(\cdot)$ denote its gradient (a row vector), Jacobian, and Hessian matrices.}.
\begin{lemma}\label{lemmaBasicIdentities}
Let $x_{n}(\cdot)$, $v_{n}(\cdot)$, $\lambda_{n}(\cdot)$, $h(\cdot)$, and $e(\cdot)$ be as defined above. Then, for $(p,u)\in\mathbb{R}^{L+1}_{++}$:
\begin{itemize}
    \item[i.]$p \cdot x_{n}(p)=1$, $p\mathbf{J}x_{n}(p)=-x_{n}(p)$, and $\lambda_{n}(p)=x_{n}(p)\cdot \nabla u(x_{n}(p))$;
    \item[ii.] $\nabla v_{n}(p)=-\lambda_{n}(p)x_{n}(p)$ and $p\cdot \nabla v_{n}(p)=-\lambda_{n}(p)$;
    \item[iii.] $x_{n}(p)=h(p,v_{n}(p))$, $h(p,u)=x_{n}(p/e(p,u))$ and $e(p,v_{n}(p))=1$.
\end{itemize}
\end{lemma}

In particular, Lemma \ref{lemmaBasicIdentities} implies that
\begin{eqnarray*}
    x_{n}(p)=\frac{\nabla v_{n}(p)}{p\cdot \nabla v_{n}(p)},
\end{eqnarray*}
for $p\in\mathbb{R}^{L}_{++}$. The next result reveals that $x_{n}(\cdot)$ is a diffeomorphism and that its inverse is obtained through an analogous expression.
\begin{proposition}\label{propInverseDemand}
The normalized Walrasian demand $x_{n}(\cdot)$ is a diffeomorphism, with
\begin{eqnarray*}
    x_{n}^{-1}(c)=\frac{\nabla u(c)}{c\cdot\nabla u(c)}.
\end{eqnarray*}
\end{proposition}
I will call the space $\mathbb{R}^{L}_{++}$ that represents the consumption bundles the \textit{consumption domain}. Due to Proposition \ref{propInverseDemand}, every set $A\subseteq\mathbb{R}^{L}_{++}$ in the consumption domain can be smoothly deformed into another set $x^{-1}_{n}(A)\subseteq\mathbb{R}^{L}_{++}$ (see the shaded area in Figure \ref{FigNormalizedWalrasian}). Although both $A$ and $x^{-1}_{n}(A)$ rest in $\mathbb{R}^{L}_{++}$, they have entirely different economic interpretations. While $A$ represents a set of bundles in the consumption domain, $x^{-1}_{n}(A)$ represents the normalized prices that would lead the household to choose such bundles. This is why I shall henceforth say that the set $x^{-1}_{n}(A)$ lies in the \textit{normalized domain}\footnote{This terminology is built on an analogy with the time and frequency domain representations of a signal obtained through the Fourier transform.}.

Clearly, a set $B\subseteq\mathbb{R}^{L}_{++}$ in the normalized domain can also be deformed into its counterpart $x_{n}(B)\subseteq\mathbb{R}^{L}_{++}$ in the consumption domain. For example, open half-lines starting at the origin in the normalized domain are deformed to wealth expansion paths \parencite[p. 25]{Mas-ColellWhinstonGreen_1995} in the consumption domain. The next result clarifies how $x^{-1}_{n}(\cdot)$ deforms the consumption domain. 

\begin{proposition}\label{propFixedPoint}
For all $\alpha>0$, there is a unique $p(\alpha)\in\mathbb{R}^{L}_{++}$ such that $\Vert x_{n}(p(\alpha))\Vert=\alpha$ and
\begin{eqnarray*}
    \frac{p(\alpha)}{\Vert p(\alpha)\Vert}=\frac{x_{n}(p(\alpha))}{\Vert x_{n}(p(\alpha))\Vert}.
\end{eqnarray*}
Furthermore, $\Vert p(\alpha)\Vert=\alpha^{-1}$ and $p(1)\in\mathbb{R}^{L}_{++}$ is the unique fixed point of $x_{n}(\cdot)$.
\end{proposition}

Proposition \ref{propFixedPoint} reveals that there is a unique point that is kept fixed by $x^{-1}_{n}(\cdot)$ and this point lies in the strictly positive orthant of the unit sphere. Also, in the strictly positive orthant of each sphere centered at the origin, there is a single point that is only scaled by $x^{-1}_{n}(\cdot)$, with its length inverted.

After some careful thought, it is possible to say that $x^{-1}_{n}(\cdot)$ can be seen as a ``double twist'' (see the regions identified by roman numerals in Figure \ref{FigNormalizedWalrasian}). The first twist is a kind of rotation around the one-dimensional ``axis'' given by $\{x_{n}(p(\alpha))\in\mathbb{R}^{L}_{++}\mid \alpha>0\}$. The second twist takes large consumption bundles and brings them closer to the origin while also taking small consumption bundles and sending them far away (in doing so, it stretches the regions near the origin and compresses the ones far from it).  The following example illustrates this interpretation, and Section \ref{sec3} will further develop these intuitive notions.

\begin{example}\label{ex1}
Let $L=2$ and $u(c_{1},c_{2})=c_{1}c_{2}$. Then, for $c,p\in\mathbb{R}^{2}_{++}$ and $u>0$, $x_{n}(p)=(1/(2p_{1}),1/(2p_{2}))$, $v_{n}(p)=1/(4p_{1}p_{2})$, $\lambda_{n}(p)=1/(2p_{1}p_{2})$, $x^{-1}_{n}(c)=(1/(2c_{1}),1/(2c_{2}))$, $h(p,u)=(\sqrt{u p_{2}/p_{1}},\sqrt{u p_{1}/p_{2}})$ and $e(p,u)=2\sqrt{u p_{1}p_{2}}$. Then,
\begin{eqnarray*}
    p(\alpha)&=&\frac{1}{\alpha}\biggl(\frac{1}{\sqrt{2}},\frac{1}{\sqrt{2}}\biggr)\\
    x_{n}(p(\alpha))&=&\alpha \biggl(\frac{1}{\sqrt{2}},\frac{1}{\sqrt{2}}\biggr),
\end{eqnarray*}
for $\alpha>0$. Figure \ref{FigNormalizedWalrasian} illustrates how $x^{-1}_{n}(\cdot)$ deforms the consumption domain to the normalized domain.

\begin{figure}[h]
\centering
\includegraphics[width=\linewidth]{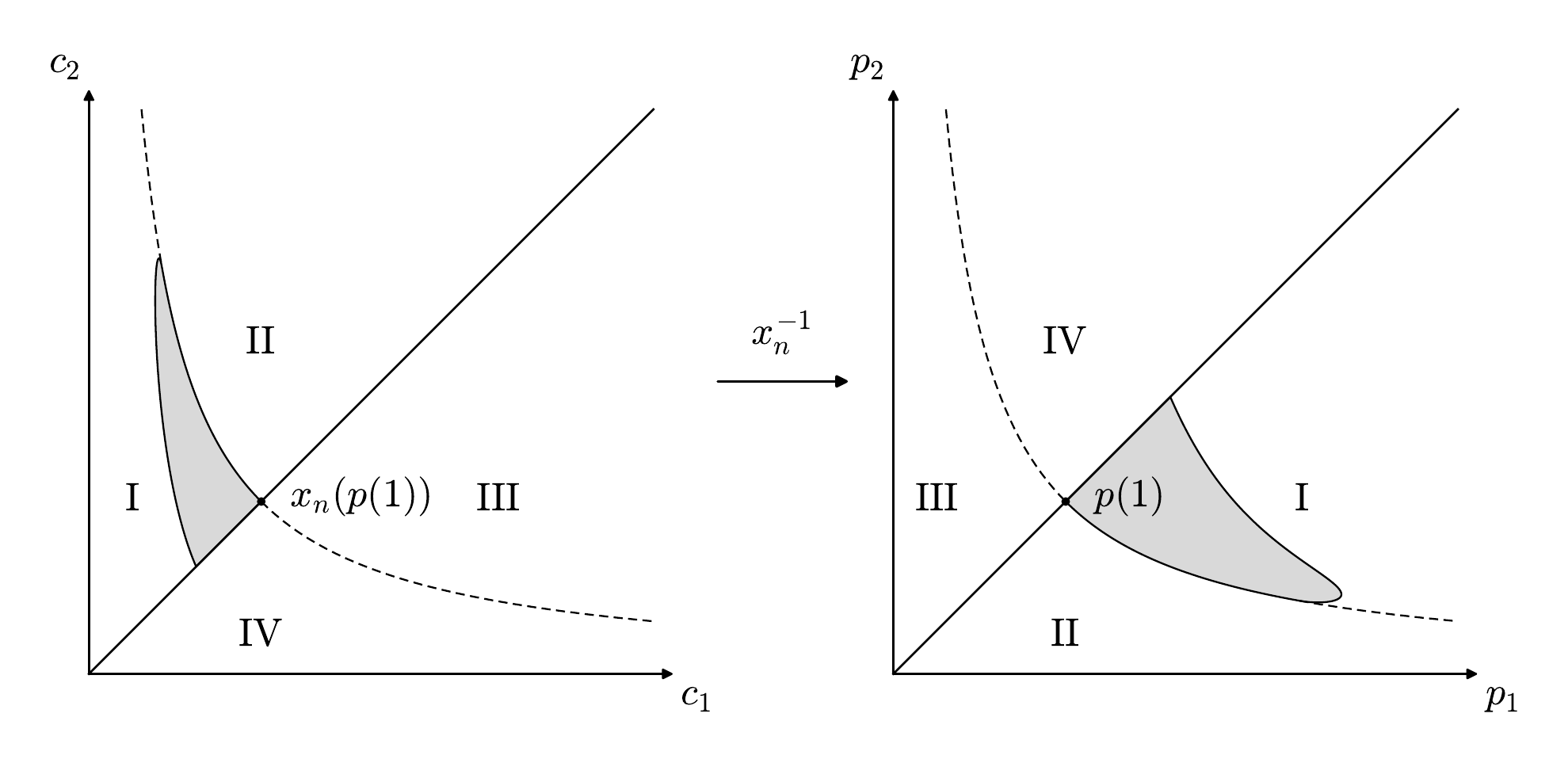}
\caption{The inverse of the normalized Walrasian demand $x^{-1}_{n}(\cdot)$ deforming the consumption domain (left) into the normalized domain (right).}
\label{FigNormalizedWalrasian}
\end{figure}
\end{example}

Proposition \ref{propInverseDemand} allows us to deal with consumption bundles and normalized prices in an equivalent way, and this could presumably be very useful when using demand theory in economic modeling (especially when dealing with non-tâtonnement models of trade). It would be even better if we could find a way in which consumption bundles were equivalent to vectors that contained, in one of their coordinates, the utility the household obtains from them, since demand theory basically relies on utility maximization and Pareto partial ordering. An attentive look at Lemma \ref{lemmaBasicIdentities} and Proposition \ref{propInverseDemand} reveals that
\begin{eqnarray*}
    e(x_{n}^{-1}(c),u(c))= x_{n}^{-1}(c) \cdot h(x_{n}^{-1}(c),v_{n}(x_{n}^{-1}(c)))= x_{n}^{-1}(c)\cdot x_{n}(x_{n}^{-1}(c))=1,
\end{eqnarray*}
for $c\in\mathbb{R}^{L}_{++}$, and, since $e(\cdot)$ is homogeneous of degree one in prices,
\begin{eqnarray*}
    e\biggl(\frac{x_{n}^{-1}(c)}{x_{nL}^{-1}(c)},u(c)\biggr)=\frac{1}{x_{nL}^{-1}(c)}.
\end{eqnarray*}
Therefore, the expenditure function is a natural way to bind a coordinate from a vector in the normalized domain to the utility level of the corresponding consumption bundle. This leads us to the following proposition.
\begin{proposition}\label{propDDiffeo}
Let $d:\mathbb{R}^{L}_{++}\rightarrow \mathbb{R}^{L}_{++}$ be given by $d(q,u)=e((q,1),u)^{-1}(q,1)$. Then, $d(\cdot)$ is a diffeomorphism, with
\begin{eqnarray*}
    d^{-1}(p)=\biggl(\frac{p_{1}}{p_{L}},\ldots,\frac{p_{L-1}}{p_{L}},v_{n}(p)\biggr).
\end{eqnarray*}
\end{proposition}

Proposition \ref{propDDiffeo} states that $d^{-1}(\cdot)$ also deforms the normalized domain $\mathbb{R}^{L}_{++}$ in a well-behaved manner. Moreover, the remarkable feature of this diffeomorphism is that it takes the sets $\{p\in\mathbb{R}^{L}_{++}\mid v_{n}(p)=u\}$, $u>0$, in the normalized domain and flattens them (actually, it also twists the two regions separated by these sets). For a set $B\subseteq\mathbb{R}^{L}_{++}$ in the normalized domain, I will say that $d^{-1}(B)\subseteq\mathbb{R}^{L}_{++}$ lies in the \textit{flat domain}. 

The next definition gathers Propositions \ref{propInverseDemand} and \ref{propDDiffeo} to define the natural diffeomorphism between the consumption and the flat domains. Also, to ease notation, throughout the paper I will write $L-1$ as a subscript to represent the first $L-1$ coordinates of a given vector (or the corresponding projection, when applied to a set).

\begin{definition}\label{defFlatteningDiffeo}
    The \textit{flattening diffeomorphism} $f:\mathbb{R}^{L}_{++}\rightarrow\mathbb{R}^{L}_{++}$ is given by
    \begin{eqnarray*}
        f(c)=d^{-1}(x^{-1}_{n}(c))=\biggl(\frac{\partial u(c)/\partial c_{1}}{\partial u(c)/\partial c_{L}},\ldots,\frac{\partial u(c)/\partial c_{L-1}}{\partial u(c)/\partial c_{L}},u(c)\biggr)=\biggl(\frac{\nabla_{L-1}u(c)}{\partial u(c)/\partial c_{L}},u(c)\biggr),
    \end{eqnarray*}
    with $f^{-1}(q,u)=x_{n}(d(q,u))=h((q,1),u)$, for $(q,u)\in\mathbb{R}^{L}_{++}$.
\end{definition}

The flattening diffeomorphism $f(\cdot)$, which is the inverse of the Hicksian demand function after a suitable normalization, represents the intuitive notion that instead of describing a consumption bundle in the consumption domain, one can equivalently pin down the indifference curve and state the marginal rates of substitution, leading to a point in the flat domain. 

The most important feature of this flattening diffeomorphism is that it eases the analysis of Pareto dominance and trade paths that are utility enhancing, since it all boils down to checking whether one is ``moving upward'' in the flat domain. This will become clearer in the next sections. 

\section{Canonical manifolds}\label{sec3}

This section presents a few results regarding the normalized Walrasian demand and the flattening diffeomorphism. I start with a definition based on classical demand theory.  

\begin{definition}\label{defCanonical}
     The \textit{canonical manifolds} through  $c\in\mathbb{R}^{L}_{++}$ are:
    \begin{itemize}
        \item[i.] The \textit{indifference hypersurface} given by $\mathcal{I}(c)=\{y\in\mathbb{R}^{L}_{++}\mid u(y)=u(c)\}$;
        \item[ii.] The \textit{offer hypersurface} given by $\mathcal{O}(c)=\{ y\in\mathbb{R}^{L}_{++}\mid (y-c)\cdot\nabla u(y)=0\}$;
        \item[iii.] The \textit{trade hyperplane} given by $\mathcal{H}(c)=\{y\in\mathbb{R}^{L}_{++}\mid (y-c)\cdot\nabla u(c)=0\}$.
    \end{itemize}
\end{definition}

Definition \ref{defCanonical} highlights sets in the consumption domain that are at the core of classical demand theory and possess a well-established economic interpretation\footnote{The trade hyperplane $\mathcal{H}(c)$ can be seen as the strictly positive boundary of the classical budget set $B(p,w)$ when $p=x^{-1}_{n}(c)$ and $w=1$.}. The next result reveals that they are, indeed, manifolds. These are, moreover, ideal candidates to deepen the study of the normalized Walrasian demand and the flattening diffeomorphism. For example, they have a single intersection $c\in\mathcal{I}(c)\cap\mathcal{O}(c)\cap \mathcal{H}(c)$, a property that is preserved under diffeomorphisms.

\begin{proposition}\label{propAreManifolds}
    For $c\in\mathbb{R}^{L}_{++}$, $\mathcal{I}(c)$, $\mathcal{O}(c)$ and $\mathcal{H}(c)\subset\mathbb{R}^{L}_{++}$ are connected manifolds without boundary of dimension $L-1$, with $\mathcal{I}(c)=\{ h((q,1),u(c))\mid q\in\mathbb{R}^{L-1}_{++}\}$, $\mathcal{O}(c)=\{y\in\mathbb{R}^{L}_{++}\mid x_{n}^{-1}(y)\cdot c=1\}=\{x_{n}((q,1)/(q,1)\cdot c)\mid q\in\mathbb{R}^{L-1}_{++}\}$, and $\mathcal{H}(c)=\{y\in\mathbb{R}^{L}_{++}\mid y\cdot x^{-1}_{n}(c)=1\}$.
\end{proposition}

Proposition \ref{propAreManifolds} shows, in particular, that the normalized Walrasian demand and the Hicksian demand furnish
global parametrizations of the offer and the indifference hypersurfaces, respectively. The next result characterizes the counterparts of the canonical manifolds in the normalized and flat domains. 

\begin{proposition}\label{propManifoldsDomains}
Let $x_{n}(\cdot)$ and $f(\cdot)$ be as defined above. Then, for $c\in\mathbb{R}^{L}_{++}$, $ x^{-1}_{n}(\mathcal{I}(c))=\{p\in\mathbb{R}^{L}_{++}\mid v_{n}(p)=u(c)\}$, $x^{-1}_{n}(\mathcal{O}(c))=\{p\in\mathbb{R}^{L}_{++}\mid p\cdot c=1\}$, $x^{-1}_{n}(\mathcal{H}(c))=\{p\in\mathbb{R}^{L}_{++}\mid x_{n}(p)\cdot x^{-1}_{n}(c)=1\}$, $f(\mathcal{I}(c))=\{(q,u)\in\mathbb{R}^{L}_{++}\mid u=u(c)\}$, $f(\mathcal{O}(c))=\{(q,u)\in\mathbb{R}^{L}_{++}\mid (q,1)\cdot c=e((q,1),u)\}$, and $f(\mathcal{H}(c))=\{(q,u)\in\mathbb{R}^{L}_{++}\mid x^{-1}_{n}(c)\cdot h((q,1),u)=1\}$.
\end{proposition}

Proposition \ref{propManifoldsDomains} reveals that the inverse of the normalized Walrasian demand takes the offer hypersurface in the consumption domain and deforms it into a hyperplane in the normalized domain. Similarly, the flattening diffeomorphism takes the indifference hypersurface and deforms it into a horizontal hyperplane in the flat domain. 

The following two results further characterize the counterparts of the canonical manifolds in the normalized domain. To ease notation, I write $\Omega(c)=\{p\in\mathbb{R}^{L}_{++}\mid v_{n}(p)\leq u(c)\}$, $c\in\mathbb{R}^{L}_{++}$.

\begin{lemma}\label{lemmaIndirectQuasiconvex}
Let  $v_{n}(\cdot)$ be as defined above. Then, $v_{n}(\cdot)$ is quasi-convex.
\end{lemma}

\begin{theorem}\label{theoIndifferenceNormDom}
Let $c\in\mathbb{R}^{L}_{++}$. Then, $\Omega(c)$ is a closed convex set in $\mathbb{R}^{L}_{++}$, $\partial \Omega(c)=x^{-1}_{n}(\mathcal{I}(c))$ and $x^{-1}_{n}(\mathcal{O}(c))$ is the supporting hyperplane at $x^{-1}_{n}(c)\in x^{-1}_{n}(\mathcal{I}(c))$. If $v_{n}(\cdot)$ is strictly quasi-convex, then $\Omega(c)$ is strictly convex.
\end{theorem}

Theorem \ref{theoIndifferenceNormDom} reveals that the counterpart of the indifference hypersurface in the normalized domain $x^{-1}_{n}(\mathcal{I}(c))$ can be seen as the boundary of a convex set $\Omega(c)$. Furthermore, the counterpart of the offer hypersurface $x^{-1}_{n}(\mathcal{O}(c))$ is the supporting hyperplane at $x^{-1}_{n}(c)\in \Omega(c)\cap x^{-1}_{n}(\mathcal{O}(c))$. This is illustrated in the following example.

\begin{example} \label{ex2}
I build on Example \ref{ex1}. The canonical manifolds are $\mathcal{I}(c)=\{y\in\mathbb{R}^{2}_{++}\mid y_{1}y_{2}=c_{1}c_{2}\}$, $\mathcal{O}(c)=\{y\in\mathbb{R}^{2}_{++}\mid c_{1}/(2y_{1})+c_{2}/(2y_{2})=1\}$ and $\mathcal{H}(c)=\{y\in\mathbb{R}^{2}_{++}\mid y_{1}/(2c_{1})+y_{2}/(2c_{2})=1\}$. The counterparts in the normalized domain are $x^{-1}_{n}(\mathcal{I}(c))=\{p\in\mathbb{R}^{2}_{++}\mid p_{1}p_{2}=1/(4c_{1}c_{2})\}$, $x^{-1}_{n}(\mathcal{O}(c))=\{p\in\mathbb{R}^{2}_{++}\mid p\cdot c=1\}$, and $x^{-1}_{n}(\mathcal{H}(c))=\{p\in\mathbb{R}^{2}_{++}\mid 1/(p_{1}c_{1})+1/(p_{2}c_{2})=4\}$, and they are depicted in Figure \ref{FigCanonicalNormalized}.

\begin{figure}[h]
\centering
\includegraphics[width=\linewidth]{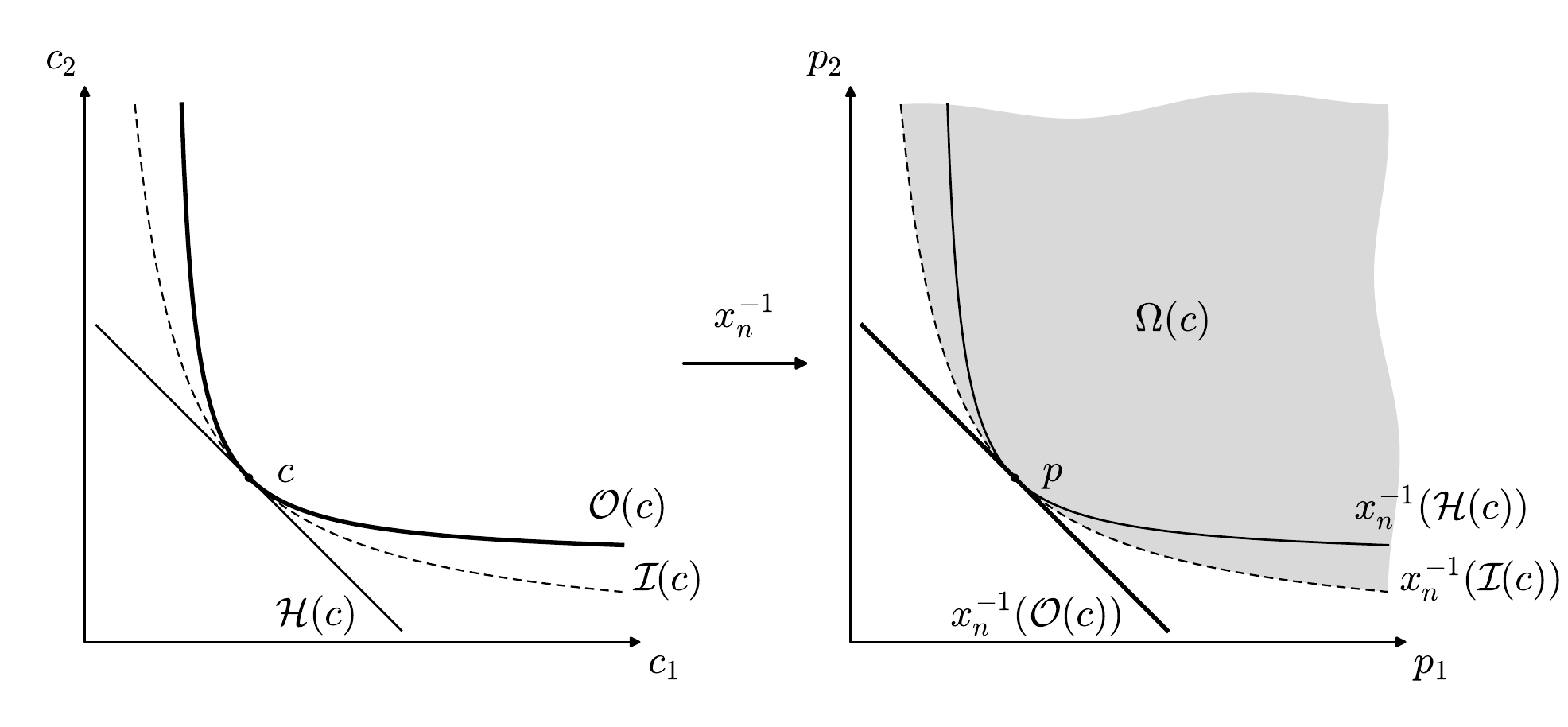}
\caption{The canonical manifolds through $c=(1,1)\in\mathbb{R}^{2}_{++}$ in the consumption domain (left) and their counterparts in the normalized domain (right).}
\label{FigCanonicalNormalized}
\end{figure}
\end{example}

The next result reveals that the patterns seen in Figure \ref{FigCanonicalNormalized} also appear when dealing with counterparts of the canonical manifolds in the flat domain (see Figure \ref{FigCanonicalManifolds2}). Before stating it, however, we need some extra notation. For $c\in\mathbb{R}^{L}_{++}$, let $\gamma_{c}:\mathbb{R}^{L-1}_{++}\rightarrow\mathbb{R}_{+}$ be given by $\gamma_{c}(q)=v_{n}((q,1)/(q,1)\cdot c)$, $\Gamma(c,s)=\{(q,s)\in\mathbb{R}^{L}_{++}\mid (q,1)\cdot c\leq e((q,1),s)\}$, $s\geq u(c)$, and $\Gamma(c)=\bigcup_{s\geq u(c)} \Gamma(c,s)$\footnote{Notice that the expenditure minimization problem implies $\Gamma(c,s)=\emptyset$, $0\leq s<u(c)$.}.

\begin{theorem}\label{theoOfferFlatDom}
Let $c\in\mathbb{R}^{L}_{++}$. Then, $\Gamma(c)$ is a closed set in $\mathbb{R}^{L}_{++}$, $\{\Gamma_{L-1}(c,s)\}_{s\geq u(c)}$ is a non-decreasing net of closed nonempty convex sets, and $\partial\Gamma(c)=f(\mathcal{O}(c))$. Furthermore, if $\gamma_{c}(\cdot)$ is (strictly) convex, then $\Gamma(c)$ is a (strictly) convex set and $f(\mathcal{I}(c))$ is the supporting hyperplane at $f(c)\in f(\mathcal{O}(c))$.
\end{theorem}

Theorem \ref{theoOfferFlatDom} reveals that the horizontal sections of $\Gamma(c)$ (i.e., $\Gamma(c,s)$, $s\geq u(c)$) are closed convex sets and their upward projections are contained in $\Gamma(c)$ (see Figure \ref{FigCanonicalManifolds2}), with the boundary of $\Gamma(c)$ furnishing the counterpart of the offer hypersurface in the flat domain $f(\mathcal{O}(c))$. 

Also, under a convexity assumption related to the normalized indirect utility function $v_{n}(\cdot)$, $\Gamma(c)$ becomes convex, and $f(\mathcal{O}(c))$ can be seen as the boundary of this convex set. This set stands on top of the horizontal hyperplane given by the counterpart of the indifference hypersurface $f(\mathcal{I}(c))$, with a single contact point at $f(c)$ (i.e., $f(\mathcal{I}(c))$ is the supporting hyperplane at $f(c)\in \Gamma(c)\cap f(\mathcal{I}(c))$). In addition, the counterpart of the trade hyperplane $f(\mathcal{H}(c))$ lies almost entirely below this horizontal hyperplane, except for the tangency point $f(c)$. 

The following example illustrates Theorem \ref{theoOfferFlatDom}.

\begin{example} \label{ex3}
I build on Example \ref{ex2}. Then, for $c\in\mathbb{R}^{2}_{++}$, $u,q>0$, $f(c)=(c_{2}/c_{1},c_{1}c_{2})$ and $f^{-1}(q,u)=h((q,1),u)=(\sqrt{u/q},\sqrt{uq})$. The counterparts in the flat domain are $f(\mathcal{I}(c))=\{(q,u)\in\mathbb{R}^{2}_{++}\mid u=c_{1}c_{2}\}$, $f(\mathcal{O}(c))=\{(q,u)\in\mathbb{R}^{2}_{++}\mid u\geq c_{1}c_{2}, q=(\sqrt{u}\pm \sqrt{u-c_{1}c_{2}})^{2}/c^{2}_{1}\}$, and $f(\mathcal{H}(c))=\{(q,u)\in\mathbb{R}^{2}_{++}\mid u\leq c_{1}c_{2}, q=(c_{1}c_{2}\pm \sqrt{c^{2}_{1}c^{2}_{2}-c_{1}c_{2}u})^{2}/c^{2}_{1}u\}$, and they are depicted in Figure \ref{FigCanonicalManifolds2}.

\begin{figure}[h]
\centering
\includegraphics[width=\linewidth]{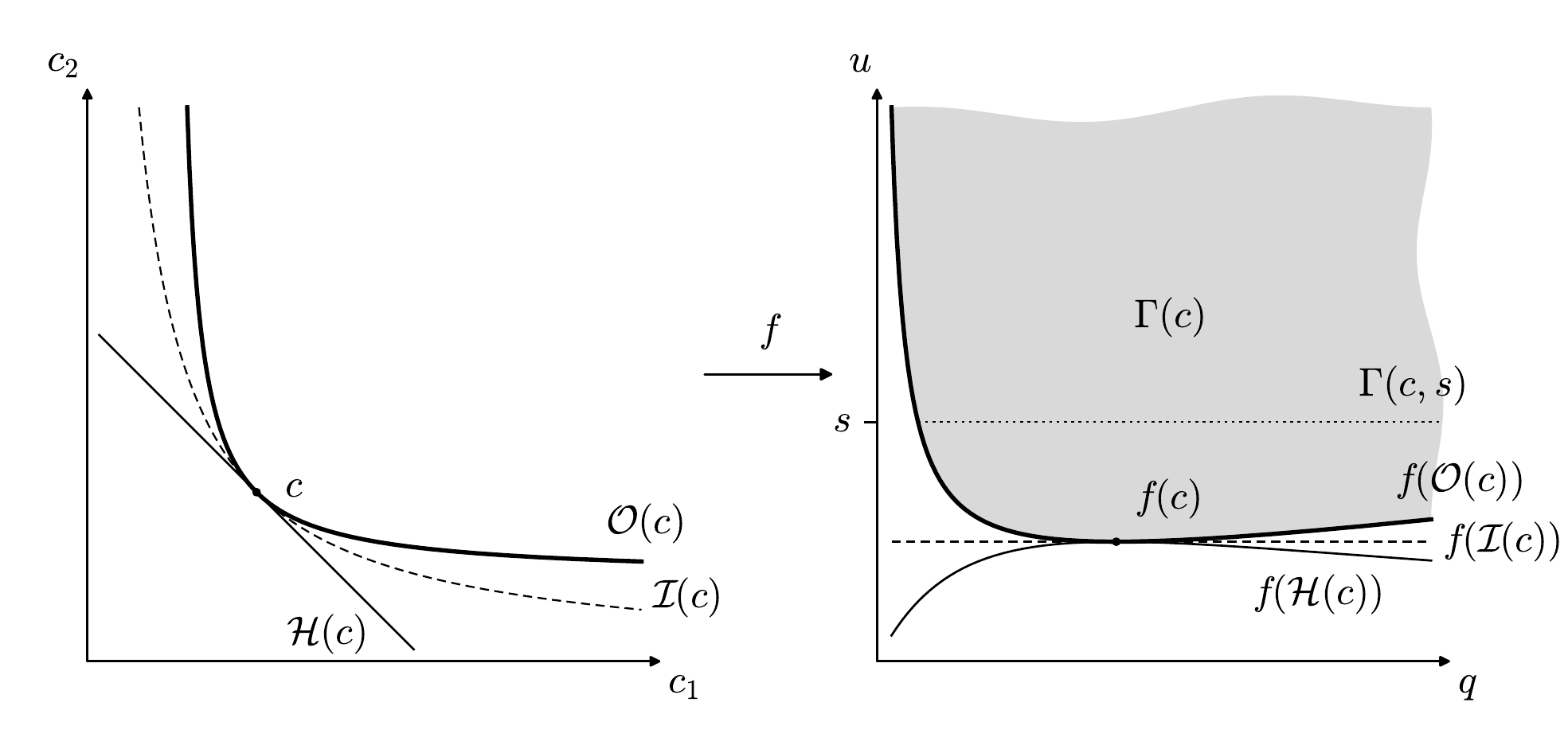}
\caption{The canonical manifolds through $c=(1,1)\in\mathbb{R}^{2}_{++}$ in the consumption domain (left) and their counterparts in the flat domain (right).}
\label{FigCanonicalManifolds2}
\end{figure}
\end{example}

With the intuition around the normalized Walrasian demand and the flattening diffeomorphism now better established, the last point of this section is a comment on the Gauss map of the indifference hypersurfaces, $N:\mathcal{I}(c)\rightarrow \mathbb{S}^{L-1}_{++}$, $c\in\mathbb{R}^{L}_{++}$, where $\mathbb{S}^{L-1}_{++}=\{p\in\mathbb{R}^{L}_{++}\mid\Vert p \Vert=1\}$, which is the mapping that associates each point of the indifference hypersurface with its oriented unit normal vector (e.g., \textcite[p. 39]{Mas-Colell_1985}). 

Let $c\in\mathbb{R}^{L}_{++}$ and notice that the Gauss map is given by 
\begin{eqnarray}\label{eqGaussMap}
    N(y)=\frac{\nabla u(y)}{\Vert \nabla u(y)\Vert}=\frac{x_{n}^{-1}(y)}{\Vert x_{n}^{-1}(y)\Vert},
\end{eqnarray}
for $y\in\mathcal{I}(c)$. The relation between the Gauss map and the Hicksian demand is also straightforward, since the map $p\mapsto h(p,u(c)), p\in\mathbb{R}^{L}_{++},\Vert p\Vert=1$, is its inverse. 

Finally, based on the Gauss map in (\ref{eqGaussMap}), we are led to define the following diffeomorphism $\tilde{f}:\mathbb{R}^{L}_{++}\rightarrow\mathbb{R}^{L}_{++}$, $\tilde{f}(y)=u(y)\nabla u(y)/\Vert \nabla u(y)\Vert$, whose inverse is $\tilde{f}^{-1}(z)=h(z/\Vert z\Vert,\Vert z\Vert)$, $z\in\mathbb{R}^{L}_{++}$. The notation $\tilde{f}(\cdot)$ is used to denote a parallel with the flattening diffeomorphism $f(\cdot)$. 

When looking for the flattening diffeomorphism, the economic intuition leads us to represent the normal direction to the indifference hypersurface as a sequence of $L-1$ marginal rates of substitution (i.e., $\nabla_{L-1}u(y)/(\partial u(y)/\partial c_{L})$) and then we just keep track of the utility level by adding a last coordinate, thus obtaining $f(\cdot)$. However, with the Gauss map in mind, one can choose to represent the normal direction by a unit vector and the utility level by a scaling factor, thus obtaining $\tilde{f}(\cdot)$. In the former case, the indifference hypersurfaces become flat; in the latter, they become the strictly positive orthant of a sphere of radius $u(y)$\footnote{Both diffeomorphisms map the foliation of $\mathbb{R}^{L}_{++}$ by indifference hypersurfaces onto a foliation of the target space indexed by the utility level: by horizontal hyperplanes in the case of $f(\cdot)$, and by spherical shells in the case of $\tilde{f}(\cdot)$.}.


\section{The manifold of Pareto optimal allocations}\label{sec4}

This section shows that the reasoning behind the flattening diffeomorphism leads to the identification of a manifold of Pareto optimal allocations in classical general equilibrium economies \parencite{ArrowDebreu_1954}. 

Let the economy be a pure exchange one with $H\geq2$ households. Utility functions $u^{h}(\cdot)$ satisfy all assumptions of Section \ref{sec2} and endowments are given by $\omega^{h}\in\mathbb{R}^{L}_{++}$, $1\leq h\leq H$, so that the budget set of household $h$ when facing prices $p\in\mathbb{R}^{L}_{++}$ is $B(p,p\cdot \omega^{h})$. The optimal consumption bundle is $x^{h}_{n}(p/p\cdot \omega^{h})$ and its trade hyperplane is given by $\mathcal{H}(x^{h}_{n}(p/p\cdot\omega^{h}))=\{y\in\mathbb{R}^{L}_{++}\mid p\cdot y=p\cdot \omega^{h}\}$, so that $\mathcal{H}(x^{h}_{n}(p/p\cdot\omega^{h}))=\partial B(p,p\cdot\omega^{h})\cap\mathbb{R}^{L}_{++}$. Also, Proposition \ref{propAreManifolds} implies that $x^{h}_{n}(p/p\cdot \omega^{h})\in \mathcal{O}(\omega^{h})$. 

Therefore, in the classical general equilibrium model, the strictly positive part of the boundary of the budget set can be seen as the trade hyperplane of the optimal consumption bundle, and this optimal bundle is an element of the offer hypersurface that passes through the household's endowment. 

An equilibrium for this economy is given by a price $p\in\mathbb{R}^{L}_{++}$ solving the market clearing equation $\sum^{H}_{h=1}x^{h}(p,p\cdot\omega^{h})=\sum^{H}_{h=1}\omega^{h}\in\mathbb{R}^{L}_{++}$, which can also be written as
\begin{eqnarray}\label{eqAlternativeMarketClearing}
\sum^{H}_{h=1}x^{h}_{n}\biggl(\frac{p}{p\cdot\omega^{h}}\biggr)=\sum_{h=1}^{H}\omega^{h}.
\end{eqnarray}
To ease notation, I write $\omega=(\omega^{1},\ldots,\omega^{H})\in\mathbb{R}^{HL}_{++}$ for the allocation. The \textit{set of equilibrium allocations} $\mathcal{E}(\omega)\subseteq\mathbb{R}^{HL}_{++}$ is given by $\mathcal{E}(\omega)=\{(x^{1}_{n}(p/p\cdot\omega^{1}),\ldots,x^{H}_{n}(p/p\cdot\omega^{H}))\in\mathbb{R}^{HL}_{++}\mid p\in\mathbb{R}^{L}_{++}\textrm{ satisfies (\ref{eqAlternativeMarketClearing})}\}$. Notice that, in equilibrium, all the counterparts of the households' bundles lie on the same vertical line in the flat domain. The following definition is also well-known (e.g., \textcite[p. 548]{Mas-ColellWhinstonGreen_1995}).

\begin{definition}\label{defSetEqAllocationsTransfers}
Let $\omega\in\mathbb{R}^{HL}_{++}$. The \textit{set of equilibrium allocations with transfers} $\mathcal{G}(\sum^{H}_{h=1}\omega^{h})\subseteq\mathbb{R}^{HL}_{++}$ is given by $\mathcal{G}(\sum^{H}_{h=1}\omega^{h})=\{(x^{1}_{n}(p/w_{1}),\ldots,x^{H}_{n}(p/w_{H}))\in\mathbb{R}^{HL}_{++} \mid \sum^{H}_{h=1}x^{h}_{n}(p/w_{h})=\sum_{h=1}^{H}\omega^{h}, p\in\mathbb{R}^{L}_{++}, w_{h}>0, 1\leq h\leq H\}$\footnote{The upper indices $h$ (e.g., $x^{h}_{n}(\cdot)$ and $\omega^{h}\in\mathbb{R}^{L}_{++}$) are used to reference functions and vectors, while the lower ones are used to reference scalars (e.g., wealth $w_{h}>0$ and utility level $u_{h}>0$).}.
\end{definition}

Definition \ref{defSetEqAllocationsTransfers} is based on the fact that, in the Second Welfare Theorem, an equilibrium with transfers admits an arbitrary wealth distribution among households and, equivalently, an arbitrary utility distribution. Furthermore, for a given allocation $\omega\in\mathbb{R}^{HL}_{++}$, one can identify the \textit{contract subset} $\mathcal{C}(\omega)\subseteq \mathcal{G}(\sum^{H}_{h=1}\omega^{h})$ \parencite[p. 523]{Mas-ColellWhinstonGreen_1995} given by $\mathcal{C}(\omega)=\{(h^{1}(p,u_{1}),\ldots,h^{H}(p,u_{H}))\in\mathbb{R}^{HL}_{++} \mid \sum^{H}_{h=1}h^{h}(p,u_{h})= \sum^{H}_{h=1}\omega^{h}, p\in\mathbb{R}^{L}_{++},u_{h}\geq u^{h}(\omega^{h}), 1\leq h\leq H\}$.

I proceed with the following definition of the set of Pareto optimal allocations.
\begin{definition}\label{defParetoOptimality}
    An allocation $z\in \mathbb{R}^{HL}_{++}$ \textit{Pareto dominates} $\tilde{z}\in\mathbb{R}^{HL}_{++}$ (denoted $z\succ \tilde{z}$), if $u^{h}(z^{h})\geq u^{h}(\tilde{z}^{h})$, $1\leq h\leq H$, with at least one strict inequality, and $\sum^{H}_{h=1}z^{h}=\sum^{H}_{h=1}\tilde{z}^{h}$. An allocation $z\in\mathbb{R}^{HL}_{++}$ is \textit{Pareto optimal} if there is no allocation that Pareto dominates it, and $\mathcal{P}\subseteq\mathbb{R}^{HL}_{++}$ is the \textit{set of Pareto optimal allocations}.
\end{definition}

The distinguishing feature of Definition \ref{defParetoOptimality} is that the aggregate endowment of the economy is only \textit{implicitly} defined by the allocation itself, which differentiates it from the classical one that \textit{explicitly} constrains the aggregate resources of the economy, usually by a \textit{feasibility} criterion (e.g., \textcite[pp. 312--313]{Mas-ColellWhinstonGreen_1995}, \textcite[p. 361]{Balasko_1979}), thus leading to the following set\footnote{Notice that the right-hand side of (\ref{eqManifoldClassicalPO}) is the intersection of two manifolds. If there is a production technology in the economy, say represented by a production set $Y\subset\mathbb{R}^{L}$ \parencite[p. 128]{Mas-ColellWhinstonGreen_1995}, then the set of Pareto optimal allocations and production plans can be written as
\begin{eqnarray*}
    \mathcal{P}=\{(z,y)\in\mathbb{R}^{HL}_{++}\times Y\mid \sum^{H}_{h=1}z^{h}-y\in\mathbb{R}^{L}_{++}\textrm{ and } \nexists\, (\tilde{z},\tilde{y})\in\mathbb{R}^{HL}_{++}\times Y \textrm{ s.t. } \sum_{h=1}^{H}z^{h}-y=\sum_{h=1}^{H}\tilde{z}^{h}-\tilde{y},\\
    u^{h}(\tilde{z}^{h})\geq u^{h}(z^{h}), 1\leq h\leq H, \textrm{ with at least one strict inequality}\}.
\end{eqnarray*}
Notice that, once more, the aggregate endowment remains implicitly defined. In this case, (\ref{eqManifoldClassicalPO}) becomes
\begin{eqnarray*}
    \mathcal{P}(\omega)=\mathcal{P}\cap\{(z,y)\in\mathbb{R}^{HL}_{++}\times Y\mid \sum^{H}_{h=1}z^{h}=\sum^{H}_{h=1}\omega^{h}+y,y\in\partial Y\},
\end{eqnarray*}
with $\partial Y\subset \mathbb{R}^{L}$ the \textit{transformation frontier} \parencite[p. 128]{Mas-Colell_1985}.} 
\begin{eqnarray}\label{eqManifoldClassicalPO}
\mathcal{P}(\omega)=\mathcal{P}\cap\biggl\{z\in\mathbb{R}^{HL}_{++}\mid \sum^{H}_{h=1}z^{h}=\sum^{H}_{h=1}\omega^{h}\biggr\}.    
\end{eqnarray}

Therefore, Definition \ref{defParetoOptimality} focuses strictly on the impossibility of utility-enhancing redistribution of resources. It leads, however, to some subtlety. For example, there are $y,z\in\mathcal{P}$ (i.e., two Pareto optimal allocations) with $u^{h}(y^{h})>u^{h}(z^{h})$, $1\leq h\leq H$, but this can only happen if $\sum^{H}_{h=1}y^{h}\neq \sum^{H}_{h=1}z^{h}$. 

Nonetheless, there are good reasons\footnote{\textcite{Dognini_2026a,Dognini_2026c} show that this definition of Pareto optimality is particularly useful when dealing with overlapping generations economies representing long-standing economic activity.} for adopting Definition \ref{defParetoOptimality}, and the next result, which was previously derived\footnote{\textcite{Balasko_1979} identified this manifold as a ``set $Q$,'' since his definition of Pareto optimality adopted a feasibility criterion \parencite[p. 361]{Balasko_1979} and, thus, differed from Definition \ref{defParetoOptimality}.} by \textcite[p. 374-375]{Balasko_1979} and is a direct consequence of the First and Second Welfare Theorems, is one of them.

\begin{theorem}\label{theoParetoOptimalSet}
    The set of Pareto optimal allocations $\mathcal{P}\subseteq\mathbb{R}^{HL}_{++}$ is a connected manifold without boundary of dimension $H+L-1\geq3$ and
    \begin{eqnarray*}
        \mathcal{P}&=&\{(x^{1}_{n}(p/w_{1}),\ldots, x^{H}_{n}(p/w_{H}))\in\mathbb{R}^{HL}_{++}\mid p\in \mathbb{R}^{L}_{++}, w_{h}>0,1\leq h\leq H\}\\
        &=&\{(h^{1}(p,u_{1}),\ldots,h^{H}(p,u_{H}))\in\mathbb{R}^{HL}_{++}\mid p\in \mathbb{R}^{L}_{++}, u_{h}>0,1\leq h\leq H\}.
    \end{eqnarray*}
\end{theorem}

The main message of Theorem \ref{theoParetoOptimalSet} is that to obtain a maximal and well-behaved structure for the set of Pareto optimal allocations, one must forgo the resource constraints. The total resources are already implicitly defined by the allocation itself, and by limiting the aggregate resources, one may unnecessarily ``chop'' the manifold of Theorem \ref{theoParetoOptimalSet}. The next example provides a depiction of this reasoning.

\begin{example}\label{ex4}
I build on Example \ref{ex3}. Let $H=2$, $u^{1}(c)=u^{2}(c)=u(c)=c_{1}c_{2}$ and $f(c)=(c_{2}/c_{1},c_{1}c_{2})$. Let $\omega=(\omega^{1},\omega^{2})=(\omega^{1}_{1},\omega^{1}_{2},\omega^{2}_{1},\omega^{2}_{2})\in\mathbb{R}^{4}_{++}$ be an allocation. Then, based on the flattening diffeomorphism (we are just reordering its coordinates), we can define the following map $\phi:\mathbb{R}^{4}_{++}\rightarrow\mathbb{R}^{4}_{++}$,
\begin{eqnarray*}
    \phi(\omega)=\biggl(\frac{\omega^{1}_{2}}{\omega^{1}_{1}},\frac{\omega^{2}_{2}}{\omega^{2}_{1}},u(\omega^{1}),u(\omega^2)\biggr).
\end{eqnarray*}
Next, in order to represent the image of an allocation by $\phi(\cdot)$ in $\mathbb{R}^{3}$, we are going to think of it as a horizontal disk. In this sense, the first three coordinates of $\phi(y)$ will simply be the coordinates of the center of the disk and the fourth coordinate will be its radius. 

Then, the image by $\phi(\cdot)$ of the manifold of Pareto optimal allocations $\mathcal{P}$ is simply all circles whose center lies on the vertical plane through the diagonal of the floor of $\mathbb{R}^{3}$. Furthermore, when we fix the aggregate endowment of the economy, say $\omega^{1}+\omega^{2}=(1,1)\in\mathbb{R}^{2}$, then the image of $\mathcal{G}(1,1)$  becomes a ``wizard hat'' with unit height and a base of unit radius. This is depicted in Figure \ref{FigWizardHat}. 

\begin{figure}[h]
\centering
\includegraphics[width=\linewidth]{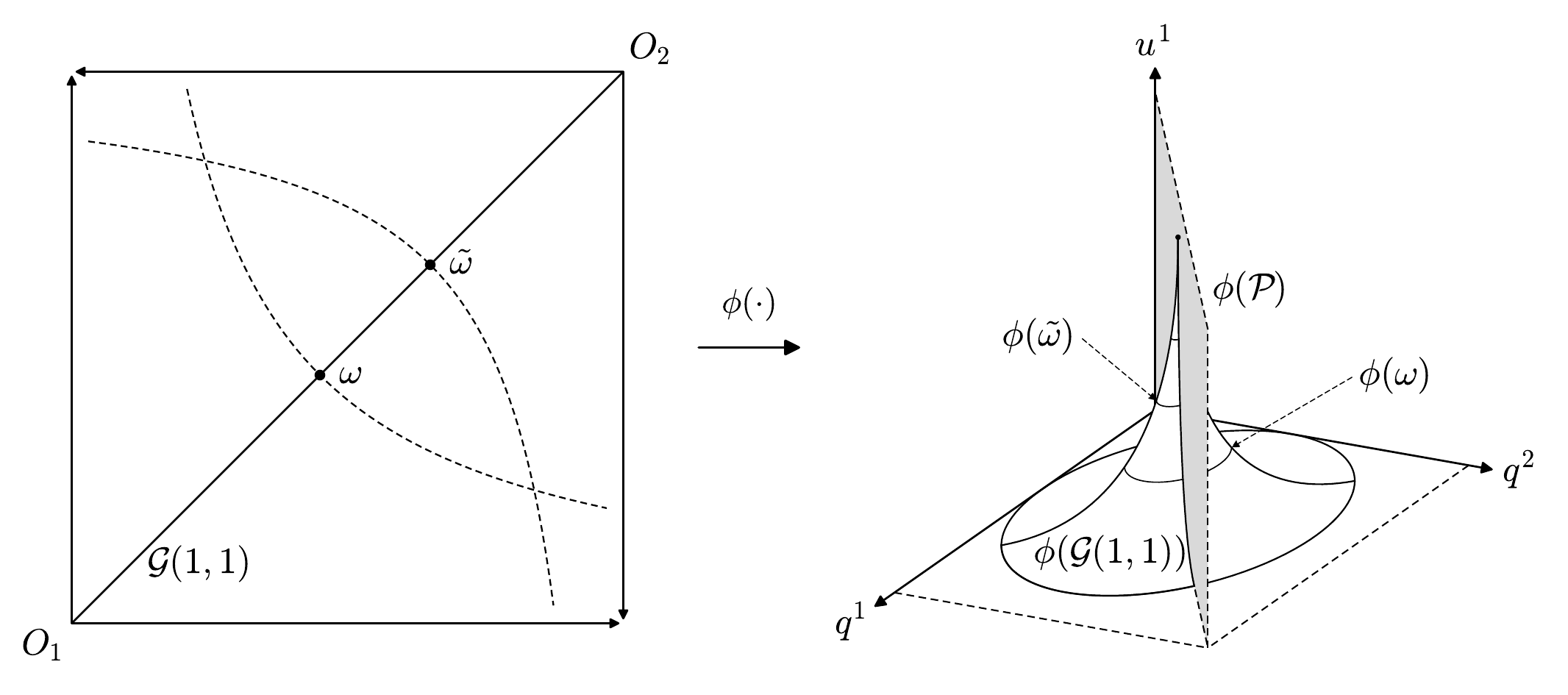}
\caption{The set of equilibrium allocations with transfers $\mathcal{G}(1,1)$ in the Edgeworth box (left) is mapped by $\phi(\cdot)$ into the ``wizard hat'' (right). All circles whose centers lie in the shaded plane form the image of the set of Pareto optimal allocations (i.e., $\phi(\mathcal{P})$).}
\label{FigWizardHat}
\end{figure}

\end{example}

A more subtle geometrical interpretation of Theorem \ref{theoParetoOptimalSet} is that the well-behaved set $\mathcal{P}$ is the locus from which all pure exchange economies originate according to a \textit{decentralizing movement} under the prevailing vector of marginal rates of substitution. 

Seen from the consumption domain, this decentralizing geometrical procedure is simply a movement of each households' bundle on the respective trade hyperplanes. Notice, however, that this ``centrifugal movement'' is like a time reversal when we compare it to an Edgeworth process, since in the latter, the allocations move toward an equilibrium point, not away from it. Stated differently, one can think of time running backward and households undoing trade according to the equilibrium market prices. I stress that this is a geometrical interpretation rather than a real economic mechanism.

Furthermore, let $\omega(y)\in\mathbb{R}^{HL}_{++}$ be a decentralized allocation obtained from $y\in\mathcal{P}$ and $\omega(z)\in\mathbb{R}^{HL}_{++}$ be obtained from $z\in\mathcal{P}$, $z\neq y$. If $\omega(y)=\omega(z)$, then $y,z\in\mathcal{E}(\omega(y))=\mathcal{E}(\omega(z))$, and we have multiple equilibria in the corresponding pure exchange economy.

Theorem \ref{theoParetoOptimalSet} also shows that the dimension of $\mathcal{P}$ is $L+H-1$ and that it is connected. Moreover, it reveals that one may ``build'' Pareto optimal allocations by simply picking an arbitrary supporting price vector $p\in\mathbb{R}^{L}_{++}$ (i.e., a vector of marginal rates of substitution, up to normalization) and, then, by distributing wealth (which leads us to points in a single ``ray'' in the normalized domain) or, more directly, utility (which leads us to points in a single ``vertical ray'' in the flat domain).

Finally, \textcite{Smale_1973}\footnote{At the end of Section 1, \textcite{Smale_1973} says that this result was ``communicated to me by Truman Bewley.''} and \textcite[p. 363]{Balasko_1979} have shown that the set $\mathcal{G}(\sum^{H}_{h=1}\omega^{h})\subseteq\mathcal{P}$ is a submanifold of dimension $H-1$. Also, a well-known result by \textcite{Debreu_1970} states that, on an open set of full measure over strictly positive endowment distributions $\omega\in\mathbb{R}^{HL}_{++}$, the set $\mathcal{E}(\omega)$ is finite.

Therefore, generically, we have $\mathcal{E}(\omega)\subsetneq \mathcal{G}(\sum^{H}_{h=1}\omega^{h})\subsetneq \mathcal{P}$ and, when going from $\mathcal{P}$ to $\mathcal{G}(\sum^{H}_{h=1}\omega^{h})$, we lose $L$ dimensions and, when going from $\mathcal{G}(\sum^{H}_{h=1}\omega^{h})$ to $\mathcal{E}(\omega)$, there is a complete degeneration of the manifold.


\section{The attraction principle}\label{sec5}

This section derives a principle governing non-tâtonnement trade based on the flattening diffeomorphism $f(\cdot)$. I start with the definition of joint trade paths, which is the main concept embodying non-tâtonnement trade.

\begin{definition}\label{defJointTradePath}
    Let $c:[0,1]\rightarrow\mathbb{R}^{HL}_{++}$ be a piecewise smooth function. Then $c(\cdot)$ is a \textit{joint trade path} if $u^{h}\circ c^{h}$ is non-decreasing, $1\leq h\leq H$, $c(1)\succ c(0)$, and $\sum^{H}_{h=1}c^{h}(t)=\sum^{H}_{h=1}c^{h}(0)$, $t\in [0,1]$.
\end{definition}

Definition \ref{defJointTradePath} states that a joint trade path is characterized by the aggregate resource constraint and the fact that utility never decreases when the trade unfolds during the normalized time period, with the final allocation Pareto dominating the initial one (so that at least one household benefits from the trade).

This is an intuitive definition that has a long tradition in the literature. For instance, the Pareto dominance between the initial and final allocations was already in \textcite[p. 221, Assumption III]{Uzawa_1962} and \textcite[p. 208, Assumption A.II]{Hahn_1962}, and the fact that utility is non-decreasing along the joint trade path is a weakened version of the exchange axiom in \textcite[p. 218]{Smale_1976a} or of the monotonicity definition in \textcite[p. 2, Definition 1]{Cornet_2024}.

A natural question that arises is whether, for a given allocation $\omega\in\mathbb{R}^{HL}_{++}$, there is some joint trade path departing from it. Notice that if $\omega\in\mathcal{P}$, Definition \ref{defJointTradePath} implies that there cannot be any such joint trade path. If $\omega\notin\mathcal{P}$, let $p\in\mathbb{R}^{L}_{++}$ be an equilibrium price \parencite{ArrowDebreu_1954,McKenzie_1954} and $c:[0,1]\rightarrow\mathbb{R}^{L}_{++}$ be the path given by 
\begin{eqnarray}\label{eqLinearPath}
    c^{h}(t)=\omega^{h}+t\biggl(x^{h}_{n}\biggl(\frac{p}{p\cdot \omega^{h}}\biggr)-\omega^{h}\biggr),
\end{eqnarray}
for $1\leq h\leq H$. Market clearing implies that
\begin{eqnarray}\label{eqConstantAgg}
    \sum^{H}_{h=1}c^{h}(t)=\sum^{H}_{h=1}\omega^{h}+t\sum^{H}_{h=1}\biggl(x^{h}_{n}\biggl(\frac{p}{p\cdot \omega^{h}}\biggr)-\omega^{h}\biggr)=\sum^{H}_{h=1}\omega^{h}=\sum^{H}_{h=1}c^{h}(0),
\end{eqnarray}
for $t\in[0,1]$. I proceed with the following lemma.

\begin{lemma}\label{lemmaLinearIncreasingUtility}
Let $c^{h}(\cdot)$, $1\leq h\leq H$, be given by (\ref{eqLinearPath}). If $p\parallel \nabla u^{h}(\omega^{h})$, then $u^{h}\circ c^{h}$ is constant. If $p\nparallel \nabla u^{h}(\omega^{h})$, then $u^{h}\circ c^{h}$ is strictly increasing.
\end{lemma}

Since $\omega\notin\mathcal{P}$, Theorem \ref{theoParetoOptimalSet} implies that $p\nparallel \nabla u^{h^{\prime}}(\omega^{h^{\prime}})$, for some $1\leq h^{\prime}\leq H$. Therefore, Lemma \ref{lemmaLinearIncreasingUtility} implies that $c(1)\succ c(0)$ and that $u^{h}\circ c^{h}$ is non-decreasing, $1\leq h\leq H$, and (\ref{eqConstantAgg}) allows us to conclude that there is a joint trade path departing from $\omega\in\mathbb{R}^{HL}_{++}$. 

This reasoning leads us to name the complement of the set of Pareto optimal allocations as the \textit{set of trade-compatible allocations}. Furthermore, I will call a joint trade path $c(\cdot)$ \textit{maximal} if $c(1)\in\mathcal{P}$, since this can be seen as a property of joint trade paths that exhaust all trade possibilities once they reach their end\footnote{The literature (e.g., \textcite{Smale_1976a}, \textcite{YuHosoya_2022}) also uses the term \textit{complete} for maximal joint trade paths. Since every joint trade path given by Definition \ref{defJointTradePath} that reaches its end at $t=1$ is, in this sense, complete, I have adopted the term \textit{maximal} to denote the impossibility of \textit{further} trade.}. 

The discussion up to this point can be summed up as follows. The set of Pareto optimal allocations is, precisely, the set of allocations that are not compatible with trade. Stated otherwise, there is a joint trade path departing from an allocation if, and only if, it is not Pareto optimal. Furthermore, every maximal joint trade path reaches the Pareto optimal set at its end. 

Next, given an allocation $\omega\in\mathbb{R}^{HL}_{++}$, let 
\begin{eqnarray*}
    \mathcal{D}(\omega)=\biggl\{y\in\mathbb{R}^{HL}_{++}\mid u^{h}(y^{h})\geq u^{h}(\omega^{h}),1\leq h\leq H, \sum^{H}_{h=1}y^{h}=\sum^{H}_{h=1}\omega^{h}\biggr\},
\end{eqnarray*}
and $\mathcal{D}^{h}(\omega)=\{z\in\mathbb{R}^{L}_{++}\mid \exists y\in \mathcal{D}(\omega),y^{h}=z\}$, $1\leq h\leq H$. Notice that all these sets are convex, and if all indifference hypersurfaces are closed in $\mathbb{R}^{L}$, then they are compact. These sets are depicted in Figure \ref{FigLens}. 

\begin{figure}[h]
\centering
\includegraphics[width=\linewidth]{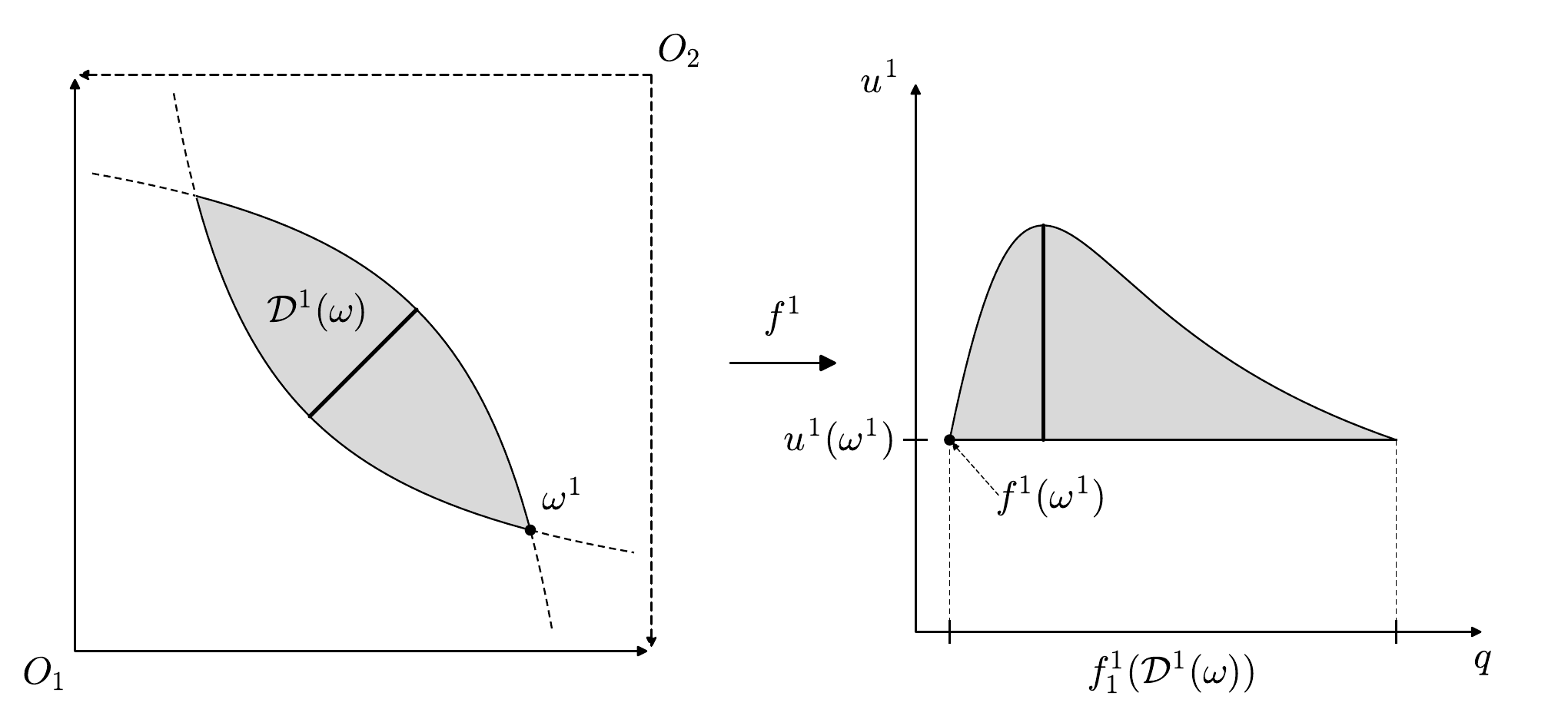}
\caption{The set $\mathcal{D}^{1}(\omega)$ of consumption bundles for household 1 compatible with a Pareto improvement over $\omega$ (left), and its image $f^{1}(\mathcal{D}^{1}(\omega))$ under the flattening diffeomorphism $f^{1}(\cdot)$ in the flat domain (right).}
\label{FigLens}
\end{figure}

Define also the \textit{set of dominant floor coordinates} $\mathcal{F}(\omega)=\bigcup_{h=1}^{H}f^{h}_{L-1}(\mathcal{D}^{h}(\omega))$, where $f^{h}(\cdot)$ denotes the flattening diffeomorphism associated with $u^{h}(\cdot)$, $1\leq h\leq H$. I claim that $\mathcal{F}(\omega)$ is arc-connected. To see this, let $p\in\mathbb{R}^{L}_{++}$ be an equilibrium price for the allocation $\omega\in\mathbb{R}^{HL}_{++}$ \parencite{ArrowDebreu_1954,McKenzie_1954} and notice that
\begin{eqnarray*}
    \biggl(x_{n}^{1}\biggl(\frac{p}{p\cdot\omega^{1}}\biggr),\ldots,x_{n}^{H}\biggl(\frac{p}{p\cdot\omega^{H}}\biggr)\biggr)\in \mathcal{D}(\omega).
\end{eqnarray*}
Therefore, for $q=(p_{1}/p_{L},\ldots,p_{L-1}/p_{L})\in\mathbb{R}^{L-1}_{++}$, we have $q\in f^{h}_{L-1}(\mathcal{D}^{h}(\omega))$, $1\leq h\leq H$. Since $\mathcal{D}^{h}(\omega)$ is convex and $f^{h}_{L-1}(\cdot)$ is continuous, $f^{h}_{L-1}(\mathcal{D}^{h}(\omega))$, $1\leq h\leq H$, is arc-connected. We conclude that $\mathcal{F}(\omega)\subseteq\mathbb{R}^{L-1}_{++}$, as the union of arc-connected sets with a common point, is itself arc-connected. 

In order to derive the attraction principle, let $c(\cdot)$ be a maximal joint trade path. Notice that $\mathcal{F}(c(t_{2}))\subseteq\mathcal{F}(c(t_{1}))$, $0\leq t_{1}\leq t_{2}\leq1$, and $f^{h}_{L-1}(c^{h}(t))\in \mathcal{F}(c(t))$, $0\leq t\leq 1$, $1\leq h\leq H$. Furthermore, since $c(1)\in\mathcal{P}$, we have $\mathcal{F}(c(1))=\{f^{1}_{L-1}(c^{1}(1))\}$. We are now able to state the principle.

\begin{proposition}[Attraction Principle]\label{propAttractionPrinciple}
    Let $c(\cdot)$ be a maximal joint trade path and $\mathcal{F}(\cdot)$ be as defined above. Then, $f^{h}_{L-1}(c^{h}(t))\in \mathcal{F}(c(t))$, $0\leq t\leq 1$, $1\leq h\leq H$, and $\{\mathcal{F}(c(t))\}_{0\leq t\leq 1}$ is a non-increasing net of arc-connected sets, with $\mathcal{F}(c(1))=\{f^{1}_{L-1}(c^{1}(1))\}$.
\end{proposition}

The Attraction Principle states that when non-tâtonnement trade unfolds according to a maximal joint trade path, one will always observe, in the flat domain, the same sort of dynamics: the counterparts of the trade paths $f^{h}\circ c^{h}$ will move upward (or at least not downward), and their projections on the ``floor'' of the flat domain will lie within a non-increasing net of arc-connected sets that shrink over time towards the point that represents the final common marginal rates of substitution. This is depicted in Figure \ref{FigAttractionPrinciple}.

\begin{figure}[h]
\centering
\includegraphics[width=\linewidth]{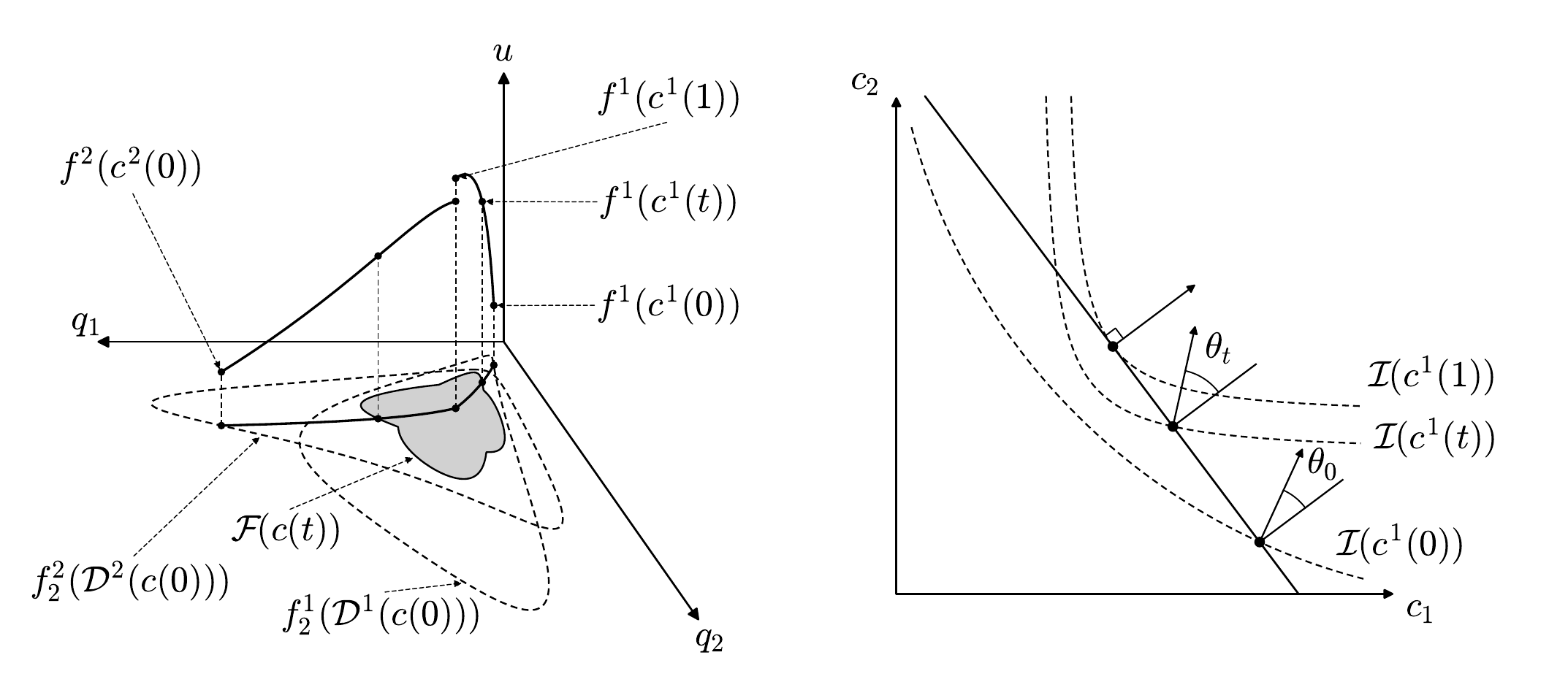}
\caption{Graphical depiction of the attraction principle in the flat domain (left) and an example of an unattractive utility function in the consumption domain (right).}
\label{FigAttractionPrinciple}
\end{figure}

Notice, however, that in Figure \ref{FigAttractionPrinciple} we have $\Vert f^{h}_{L-1}(c^{h}(t))-f^{1}_{L-1}(c^{1}(1))\Vert$, $1\leq h\leq H$, strictly decreasing, but this is not ensured by Proposition \ref{propAttractionPrinciple}. To further investigate this and motivated by (\ref{eqLinearPath}) and Lemma \ref{lemmaLinearIncreasingUtility}, I proceed with the following definition.

\begin{definition}\label{defJointLinearTradePath}
The \textit{linear joint trade path} $c(\omega,p,s,\cdot):[0,1]\rightarrow\mathbb{R}^{HL}_{++}$ starting at $\omega\in \mathbb{R}^{HL}_{++}$, under prices $p\in \mathbb{R}^{L}_{++}$, and \textit{relative trade speeds} $s\in[0,1]^{H}$ is given by 
\begin{eqnarray}\label{eqLinearPathSpeed}
    c^{h}(\omega^{h},p,s_{h},t)=\omega^{h}+ts_{h}\biggl(x^{h}_{n}\biggl(\frac{p}{p\cdot \omega^{h}}\biggr)-\omega^{h}\biggr),
\end{eqnarray}
for $1\leq h\leq H$, provided that $\sum^{H}_{h=1}c^{h}(\omega^{h},p,s_{h},t)=\sum_{h=1}^{H}\omega^{h}$, $0\leq t\leq 1$, and $c(\omega,p,s,1)\succ\omega$.
\end{definition}

Definition \ref{defJointLinearTradePath} requires a few comments. The first is that every linear joint trade path is, indeed, a joint trade path due to Lemma \ref{lemmaLinearIncreasingUtility}. Second, one must pay careful attention to the terminology, since it may be misleading. There are several joint trade paths that are \textit{linear} and do not follow (\ref{eqLinearPathSpeed}) (think of an economy with at least three households to see this). I use this terminology because the only \textit{linear joint trade paths} that are relevant for the theory of stochastic non-tâtonnement in this paper are the ones satisfying (\ref{eqLinearPathSpeed}).

The third comment is related to the premises that support Definition \ref{defJointLinearTradePath} as an intuitive way to model non-tâtonnement trade. Notice that one can reach such a definition from the following premises: (i) households trade to obtain their most desired bundle given current prices; (ii) they do so through the shortest possible trade path (i.e., linear); and (iii) trade must be utility-enhancing in every step of the way (i.e., locally). Clearly, (i) and (ii) are encompassed by (\ref{eqLinearPathSpeed}) and (iii) is satisfied due to Lemma \ref{lemmaLinearIncreasingUtility}.  

The fourth comment is that if $\omega\notin\mathcal{P}$ and $p\in\mathbb{R}^{L}_{++}$ is an equilibrium \parencite{ArrowDebreu_1954,McKenzie_1954}, then $c(\omega,p,s,\cdot)$, $s=(1,\ldots,1)\in[0,1]^{H}$, is a linear joint trade path and the First Welfare Theorem implies that it is maximal (i.e., $c(\omega,p,s,1)\in\mathcal{P}$).

The last comment is that the concept of \textit{relative trade speed} can be understood from the following perspective. The scalar $s_{h}\in[0,1]$, $1\leq h\leq H$, measures the fraction of the desired trade $x^{h}_{n}(p/p\cdot\omega^{h})-\omega^{h}\in\mathbb{R}^{L}$ that the household actually carried out during the unit time interval. Therefore, if households are in a situation in which prices shift in every unit interval, then these scalars measure the speed at which they trade \textit{relative} to price movements. This point will become clearer in Section \ref{sec6}. The next result is obtained from the same reasoning as the one that led to the attraction principle.

\begin{lemma}\label{lemmaLinearNet}
    Let $\omega\in\mathbb{R}^{HL}_{++}$, $p=(q,1)\in\mathbb{R}^{L}_{++}$, $s\in[0,1]^{H}$, and $c(\omega,p,s,\cdot)$ be a linear joint trade path. Then, $f^{h}_{L-1}(c^{h}(\omega^{h},p,s_{h},t))\in \mathcal{F}(c(\omega,p,s,t))$, $0\leq t\leq 1$, $1\leq h\leq H$, and $\{\mathcal{F}(c(\omega,p,s,t))\}_{0\leq t\leq 1}$ is a non-increasing net of arc-connected sets, with $q\in\mathcal{F}(c(\omega, p,s,1))$.
\end{lemma}

Now we are able to provide a partial answer to the question of which conditions ensure that $\Vert f^{h}_{L-1}(c^{h}(t))-q\Vert$, $1\leq h\leq H$, $t\in[0,1]$, is  non-increasing. Before that, however, I need the following definition.

\begin{definition}\label{defAttractiveUtility}
Let $u:\mathbb{R}^{L}_{+}\rightarrow\mathbb{R}_{+}$ be a utility function satisfying all assumptions from Section \ref{sec2}. Then, $u(\cdot)$ is \textit{attractive} if 
\begin{eqnarray}\label{eqAttractiveUtility}
    \biggl(\frac{\nabla_{i} u(c)}{\nabla_{j}u(c)}-\frac{p_{i}}{p_{j}}\biggr)\nabla u(c)(e_{j}^{T}e_{i}-e_{i}^{T}e_{j})\mathbf{H}u(c)\biggl(x_{n}\biggl(\frac{p}{p\cdot c}\biggr)-c\biggr)^{T}\leq0,
\end{eqnarray}
for $c,p\in\mathbb{R}^{L}_{++}$, $1\leq i,j\leq L$.
\end{definition}

Definition \ref{defAttractiveUtility} is straightforwardly obtained when one asks $\Vert f^{h}_{i}(c^{h}(\omega^{h}, p, s_{h}, t))-p_{i}/p_{L}\Vert$, $0\leq t\leq 1$, $1\leq i\leq L-1$, to be non-increasing, with $c(\omega,p,s,\cdot)$ a linear joint trade path. Therefore, if $c(\omega,p,s,\cdot)$ is a linear joint trade path with prices normalized as $p=(q,1)$, and $u^{h}(\cdot)$ is attractive, $1\leq h\leq H$, then $f^{h}_{L-1}(c^{h}(\omega^{h},p,s_{h},t))$ approaches $q\in\mathbb{R}^{L-1}_{++}$ coordinate by coordinate.

The next result makes this coordinate-by-coordinate convergence precise for an arbitrary pair of goods $i,j$ (not only $i$ against the numeraire $L$), and is the technical ingredient behind Lemma \ref{lemmaMonotoneMm} in Section \ref{sec6}.

\begin{proposition}\label{propAttractionOfMarginalRates}
Let $\omega^{h}\in\mathbb{R}^{L}_{++}$, $p\in\mathbb{R}^{L}_{++}$, $s_{h}\in[0,1]$, $c^{h}(\omega^{h},p,s_{h},\cdot):[0,1]\rightarrow\mathbb{R}^{L}_{++}$ be given by (\ref{eqLinearPathSpeed}), and
\begin{eqnarray*}
    \delta_{ij}(t)=\biggl(\frac{\nabla_{i}u^{h}(c^{h}(\omega^{h},p,s_{h},t))}{\nabla_{j}u^{h}(c^{h}(\omega^{h},p,s_{h},t))}-\frac{p_{i}}{p_{j}}\biggr)^{2},
\end{eqnarray*}
for $1\leq i,j\leq L$, $t\in[0,1]$. If $u^{h}(\cdot)$ is attractive, then $\delta_{ij}(\cdot)$ is non-increasing on $[0,1]$. Furthermore, if $s_{h}=1$, then $\delta_{ij}(1)=0$.
\end{proposition}

The last result of this section shows that CES utility functions with $\sigma<1$ (and Cobb--Douglas) are attractive.

\begin{proposition}\label{propIsAttractive}
Let $u:\mathbb{R}^{L}_{+}\rightarrow \mathbb{R}_{+}$ be given by $u(c)=(\sum^{L}_{i=1}\alpha_{i}c_{i}^{\sigma})^{\frac{1}{\sigma}}$, $\sigma<1$, $\sigma\neq0$, or $u(c)=\prod^{L}_{i=1}c_{i}^{\alpha_{i}}$, $\alpha\in\mathbb{R}^{L}_{++}$. Then $u(\cdot)$ is attractive.
\end{proposition}

One point that must be highlighted in Proposition \ref{propIsAttractive} is that it encompasses CES functions with $\sigma<0$, so we are not restricted to the case of gross substitutes (cf. the instability examples of \textcite{Scarf_1960}).

Finally, let me comment on how the attraction principle and Proposition \ref{propAttractionOfMarginalRates} can be related to the matter of the accessibility of Pareto optima (i.e., to the question of whether or not there is a non-tâtonnement process that starts in $\omega\in\mathbb{R}^{HL}_{++}$ and approaches, at least asymptotically, a given Pareto optimal allocation $\tilde{\omega}\in\mathcal{P}$). 

The results in the literature already provide sufficient conditions that ensure the accessibility of a given Pareto optimum. For instance, \textcite[p. 4]{Cornet_2024} proves accessibility under standard assumptions for the sort of non-tâtonnement processes described in this section\footnote{See Assumption PEE and Corollary 1 in \textcite[p. 4]{Cornet_2024}.}. Nonetheless, the attraction principle may be used to \textit{explicitly} build non-tâtonnement processes that approach a given Pareto optimal allocation $\tilde{\omega}\in\mathcal{P}$. 

To ease the argument, let all utility functions be attractive and consider a succession of linear joint trade paths. In each linear ``step,'' due to Proposition \ref{propAttractionOfMarginalRates}, the marginal rates of substitution come closer to the current normalized price governing trade. If one wants to reach $\tilde{\omega}\in\mathcal{P}$, one must ensure that in every step: (i) trade actually takes place at the ongoing prices and trade speeds; (ii) such trade does not cause $f^{1}_{L-1}(\tilde{\omega}^{1})\in\mathbb{R}^{L-1}_{++}$ to leave the region given by $\mathcal{F}(\cdot)$ on the floor of the flat domain. This point will not be further developed in this paper, but the idea of successive linear joint trade paths is the basis of the stochastic non-tâtonnement processes defined in the next section.

\section{Stochastic non-tâtonnement processes}\label{sec6}

This section builds a model of stochastic non-tâtonnement trade based on linear joint trade paths and the attraction principle. To do so, it is convenient to define the following sets. The \textit{set of relative trade speeds} at $\omega\in\mathbb{R}^{HL}_{++}$ under prices $p\in\mathbb{R}^{L}_{++}$, $\mathcal{S}(\omega,p)\subseteq[0,1]^{H}$, is given by $\mathcal{S}(\omega,p)=\{s\in[0,1]^{H}\setminus\{0\}\mid \exists\, c(\omega,p,s,\cdot)\textrm{ linear joint trade path}\}$, and the \textit{set of normalized trade-compatible prices} at $\omega\in\mathbb{R}^{HL}_{++}$, $\mathcal{T}(\omega)\subseteq\mathbb{R}^{L}_{++}$, $\mathcal{T}(\omega)=\{p\in\mathbb{R}^{L}_{++}\mid \mathcal{S}(\omega,p)\neq \emptyset, p_{L}=1\}$.

First, notice that if $s\in\mathcal{S}(\omega,p)$, then $\lambda s\in \mathcal{S}(\omega,p)$, $\lambda\in(0,1]$ (i.e., it is always possible to slow down a trade process), and there is $\tilde{\lambda}\geq1$ such that $\tilde{\lambda}s\in\mathcal{S}(\omega,p)$ and $\Vert\tilde{\lambda}s\Vert_{\infty}=1$ (i.e., it is always possible to trade at least until the desired bundle of some household is reached).

The next two results collect basic properties of $\mathcal{T}(\cdot)$ and $\mathcal{S}(\cdot,\cdot)$ that will be used throughout this section.

\begin{lemma}\label{lemmaTradeCompPrices}
Let $\omega\in\mathbb{R}^{HL}_{++}$ be an allocation. Then $\mathcal{T}(\omega)\neq\emptyset$ if, and only if, $\omega\notin\mathcal{P}$.
\end{lemma}

Lemma \ref{lemmaTradeCompPrices} follows from the classical result on equilibrium existence \parencite{ArrowDebreu_1954,McKenzie_1954}: given any allocation, the corresponding equilibrium price furnishes relative trade speeds $(1,\ldots,1)\in\mathcal{S}(\omega,p)$, so that $\mathcal{T}(\omega)\neq\emptyset$ whenever $\omega\notin\mathcal{P}$; conversely, at a Pareto optimal allocation the equilibrium price leaves every household already at its desired bundle, so that no relative trade speeds can produce a Pareto improvement. This simply means that there is a normalized trade-compatible price (i.e., $p\in\mathcal{T}(\omega)$) if, and only if, the allocation itself (i.e., $\omega\in\mathbb{R}^{HL}_{++}$) is trade-compatible.

\begin{proposition}\label{propTradeUnfold}
Let $\omega\in\mathbb{R}^{HL}_{++}$, $p=(q,1)\in\mathcal{T}(\omega)$, and $s\in\mathcal{S}(\omega,p)$. Then,
\begin{eqnarray}\label{eqBoxSet1}
    \mathcal{T}_{L-1}(\omega)\subseteq\mathcal{F}(\omega),
\end{eqnarray}
and $q\in\mathcal{T}_{L-1}(c(\omega,p,s,t))$, for $0\leq t<1$.
\end{proposition}

Proposition \ref{propTradeUnfold} shows two things. First, per (\ref{eqBoxSet1}), when looking for normalized trade-compatible prices, one can focus on the set $\mathcal{F}(\omega)$; this point will become clearer later on. Second, a normalized trade-compatible price $q$ remains trade-compatible throughout the linear joint trade path it supports, up to (but not necessarily including) its endpoint. Now we are able to define the stochastic non-tâtonnement processes.

\begin{definition}\label{defStochasticNTProcess}
Let $(X, \mathcal{B},\mathbb{P})$ be a probability space. Then, $\{\omega_{t}\}_{t\geq0}$ is a \textit{stochastic non-tâtonnement process} starting at $\omega\in\mathbb{R}^{HL}_{++}$ if, a.s., $\omega_{0}=\omega$ and\footnote{$\mathbf{I}_{A}:X\rightarrow\{0,1\}$ is the \textit{indicator function} of the set $A\subseteq X$, i.e., $\mathbf{I}_{A}(x)=1$ if, and only if, $x\in A$.}
\begin{eqnarray}\label{eqStochasticDef1}
    \omega_{t}= c(\omega_{t-1},p_{t},s_{t},1)\mathbf{I}_{\mathcal{T}(\omega_{t-1})\neq\emptyset}+\omega_{t-1}\mathbf{I}_{\mathcal{T}(\omega_{t-1})=\emptyset}
\end{eqnarray}
for $t\geq1$, with the \textit{normalized prices process} $\{p_{t}\}_{t\geq1}$ and the \textit{relative trade speeds process} $\{s_{t}\}_{t\geq1}$ satisfying 
\begin{eqnarray}\label{eqStochasticDef2}
    (p_{t},s_{t})\in\mathcal{T}(\omega_{t-1})\times \mathcal{S}(\omega_{t-1},p_{t}),
\end{eqnarray}
if $\mathcal{T}(\omega_{t-1})\neq\emptyset$, for $t\geq1$.
\end{definition}

Definition \ref{defStochasticNTProcess} is fairly intuitive if one imagines a random walk in the consumption domain. The initial allocation is given by $\omega\in\mathbb{R}^{HL}_{++}$ and, due to (\ref{eqStochasticDef1}), the random walk follows linear joint trade paths unless it reaches the manifold of Pareto optimal allocations $\mathcal{P}\subset\mathbb{R}^{HL}_{++}$, in which case it just remains there. Notice, in particular, that every realization of a stochastic non-tâtonnement process $\{\omega_{t}\}_{t\geq0}$ is an Edgeworth's barter process \parencite{Uzawa_1962}. Since the deterministic case can be seen as a particular case of the stochastic one and the stochastic or deterministic nature of the process is usually clear from context, I will often refer to these stochastic non-tâtonnement processes as Edgeworth processes to avoid unnecessary terminology variation.

Since the draws of an Edgeworth process are based on linear joint trade paths, there are two other processes underlying it: a normalized prices process $\{p_{t}\}_{t\geq1}$ and a relative trade speeds process $\{s_{t}\}_{t\geq1}$. Due to (\ref{eqStochasticDef2}), the normalized prices process must satisfy $p_{t}\in\mathcal{T}(\omega_{t-1})$, $t\geq 1$, and this can be seen as an Edgeworthian auctioneer (e.g., \textcite[p. 324]{ArrowHahn_1971}) calling successive trade-compatible prices. 

Unlike a Walrasian auctioneer, who does not allow trade to happen unless an equilibrium is reached (i.e., a maximal linear joint trade path), this Edgeworthian auctioneer allows trade to occur even out of equilibrium. The relative trade speeds process $\{s_{t}\}_{t\geq1}$ then defines how trade evolves under current prices, thus representing all sorts of interactions among households that are not explicitly carried out in the model (e.g., physical constraints and speculative behavior). Furthermore, since $s_{t}\in\mathcal{S}(\omega_{t-1},p_{t})$, $t\geq1$, this implies that if trade can happen, it happens\footnote{This characteristic of the Edgeworth process is also present, for example, in the exchange axiom in \textcite[p. 218]{Smale_1976a}. It is not hard to adapt the definitions to allow the process to stand still even when trade is possible (i.e., when the process is in the set of trade-compatible allocations), although the economic relevance of this is not clear.}. 

In view of Definition \ref{defStochasticNTProcess}, one may ask why we should work under this stochastic framework. I believe the answer to this question involves the following points. 

The stochastic framework of Definition \ref{defStochasticNTProcess} allows us to move beyond the ``too deterministic a setting'' \parencite[p. 211]{Smale_1976a} given by trade rules like (\ref{eqNonTatoPrice}). As noted by \textcite[p. 333]{ArrowHahn_1971}, in the deterministic setup, ``it is not easy to state a sensible [transaction] rule for the rate at which exchange is taking place between households at any one moment of time. There may be a considerable variety of exchanges possible at any [price] and [allocation], and it is not clear what simple rule of choice is appropriate.'' One way of dealing with this ``variety of exchanges possible at any [price] and [allocation]'' is precisely the stochastic setup.

Furthermore, the stochastic framework represents our lack of knowledge about many factors that influence price movements and trade, and all these ``unseen engines'' are always running behind the stochasticity of the model. Therefore, since we are unable to model each engine specifically, we bundle their effects under a probabilistic setup, switching certain/deterministic outcomes of the theory for stochastic ones. 

Clearly, this leads to the more profound and philosophical century-old question\footnote{This is the core of the famous debate between Niels Bohr and Albert Einstein \parencite{EinsteinPodolskyRosen_1935,Bohr_1935} on whether the universe is intrinsically stochastic, as proposed by the Copenhagen interpretation, or whether apparent randomness merely reflects the observer's incomplete knowledge of an underlying deterministic reality --- a position Einstein voiced in his famous remark, ``God does not play dice.''} of whether we  \textit{currently do not know all the economic engines}, in which case, with some effort, we could eventually be able to write the perfect deterministic model, or if we \textit{cannot know them} and, therefore, no effort will ever lead us to such a model. In any case, we are \textit{currently} not able to build models that account for all the complexity of an economy, so the stochastic setup of Definition \ref{defStochasticNTProcess} appears to be a reasonable workhorse.

The next example illustrates a stochastic non-tâtonnement process.

\begin{example}\label{ex5}
I build on Examples \ref{ex1}--\ref{ex3}. Let $H=2$, $u^{1}(c)=u^{2}(c)=c_{1}c_{2}$, $c\in\mathbb{R}^{2}_{+}$, and $\omega=(\omega^{1},\omega^{2})\in\mathbb{R}^{4}_{++}$ be an allocation, with $\omega^{1}+\omega^{2}=(3,3)$ and $\omega^{1}_{1}>\omega^{1}_{2}$ (i.e., the allocation lies below the diagonal of the Edgeworth box). Then, $ \mathcal{T}(\omega)=\{(q,1)\in\mathbb{R}^{2}_{++}\mid q\in (\omega^{1}_{2}/\omega^{1}_{1},\omega^{2}_{2}/\omega^{2}_{1})\}$. Furthermore, if $q\in(\omega^{1}_{2}/\omega^{1}_{1},1]$, then
\begin{eqnarray*}
    \mathcal{S}(\omega,(q,1))=\biggl\{\lambda \biggl(1,\frac{\omega^{2}_{2}-q\omega^{2}_{1}-3(1-q)}{\omega^{2}_{2}-q\omega^{2}_{1}}\biggr)\in\mathbb{R}^{2}_{++}\mid \lambda \in(0,1]\biggr\}.
\end{eqnarray*}
Given $(q,1)\in \mathcal{T}(\omega)$, $q\in(\omega^{1}_{2}/\omega^{1}_{1},1]$, let $\tilde{s}(\omega,q)\in \mathcal{S}(\omega,(q,1))$ be the largest possible vector of relative trade speeds (i.e., all trade possibilities are exhausted by the households under current prices), so that
\begin{eqnarray*}
    \tilde{s}(\omega,q)=\biggl(1,\dfrac{\omega^{2}_{2}-q\omega^{2}_{1}-3(1-q)}{\omega^{2}_{2}-q\omega^{2}_{1}}\biggr).
\end{eqnarray*}

Next, notice that the linear joint trade path starting at $\omega$ under prices $(q,1)\in\mathcal{T}(\omega)$ and relative trade speeds $\tilde{s}(\omega,q)$ has its endpoint given by
\begin{eqnarray*}
    c^{1}(\omega,(q,1),\tilde{s}(\omega,q),1)&=&\biggl(\frac{q\omega^{1}_{1}+\omega^{1}_{2}}{2q},\frac{q\omega^{1}_{1}+\omega^{1}_{2}}{2}\biggr)\\
    c^{2}(\omega,(q,1),\tilde{s}(\omega,q),1)&=&\biggl(\frac{2q\omega^{2}_{1}+q\omega^{1}_{1}-\omega^{1}_{2}}{2q},\frac{2\omega^{2}_{2}-q\omega^{1}_{1}+\omega^{1}_{2}}{2}\biggr),
\end{eqnarray*}
for $q\in(\omega^{1}_{2}/\omega^{1}_{1},1]$. 

Next, let $(X,\mathcal{B},\mathbb{P})$ be a probability space, $\{Y_{t}\}_{t\geq1}$ be a sequence of i.i.d. random variables with a Bernoulli distribution, and $\mathbb{P}(Y_{1}=1)=1/2$. Also, let $\omega^{1}=(2,1)$ and $\{q_{t}\}_{t\geq1}$ be given by $q_{t}=1-2^{-(t+1)}$, $t\geq1$. We can then define $\{\omega_{t}\}_{t\geq0}$ by $\omega_{0}=\omega$ and
\begin{eqnarray*}
    \omega_{t}=\begin{cases}
        c(\omega_{t-1},p_{t},s_{t},1)\textrm{ if } \prod^{t}_{i=1}Y_{i}=1\\
        (x^{1}_{n}(p_{t}/p_{t}\cdot\omega^{1}_{t-1}),x^{2}_{n}(p_{t}/p_{t}\cdot\omega^{2}_{t-1})), \textrm{ if }\prod^{t}_{i=1}Y_{i}=0,
    \end{cases}
\end{eqnarray*}
with the normalized prices process $\{p_{t}\}_{t\geq1}$ given by
\begin{eqnarray*}
    p_{t}=\begin{cases}
        (q_{t},1),\textrm{ if } \prod^{t}_{i=1}Y_{i}=1\\
        (1,1), \textrm{ if }\prod^{t}_{i=1}Y_{i}=0,
    \end{cases}
\end{eqnarray*}
for $t\geq1$, and the relative trade speeds process $\{s_{t}\}_{t\geq1}$ given by
\begin{eqnarray*}
    s_{t}=\begin{cases}
        \tilde{s}(\omega_{t-1},q_{t}),\textrm{ if } \prod^{t}_{i=1}Y_{i}=1\\
        (1,1), \textrm{ if }\prod^{t}_{i=1}Y_{i}=0.
    \end{cases}
\end{eqnarray*}
for $t\geq1$. It is straightforward that $\{\omega_{t}\}_{t\geq0}$ is a stochastic non-tâtonnement process. Furthermore, by the very construction of this process, the Borel-Cantelli lemma allows us to write
\begin{eqnarray*}
    \mathbb{P}(\lim_{t\rightarrow\infty}\omega_{t}\in\mathcal{P})\geq\mathbb{P}\biggl(\lim_{t\rightarrow+\infty}\prod^{t}_{i=1}Y_{i}=0\biggr)=1.
\end{eqnarray*}
\end{example}

The Edgeworth process from Example \ref{ex5} reaches, with full probability, the manifold of Pareto optimal allocations. This behavior of stochastic non-tâtonnement processes can be seen being derived from a stochastic version of the First Welfare Theorem.

Let me be more precise. The statement of the First Welfare Theorem implies that every Walrasian equilibrium (i.e., maximal linear joint trade path) is Pareto efficient (i.e., reaches the manifold of Pareto optimal allocations). Therefore, in a stochastic setup, this result could be generalized to state that, with full probability, every Edgeworth process reaches (at least asymptotically) the manifold of Pareto optimal allocations. The next result provides sufficient conditions for this to happen in the Edgeworth box.

\begin{theorem}[Stochastic First Welfare Theorem]\label{theoStochasticFirstWelfare}
Let $(X,\mathcal{B},\mathbb{P})$ be a probability space, $H=L=2$, and $\{\omega_{t}\}_{t\geq0}$ be a stochastic non-tâtonnement process starting at $\omega\in\mathbb{R}^{4}_{++}$. Suppose there is $\psi_{q}\in L^{1}_{++}(\mathbb{R}_{++})$\footnote{$L^{1}_{++}(\mathbb{R}_{++})$ denotes the set of integrable strictly positive functions defined on $\mathbb{R}_{++}$.} such that the normalized prices process $\{p_{t}\}_{t\geq1}$ satisfies
\begin{eqnarray}\label{eqTheoStochastic}
    \mathbb{P}(p_{t}\in A\times\{1\}\mid \omega_{t-1},\omega_{t-1}\notin\mathcal{P})=\frac{\int_{\mathcal{T}_{L-1}(\omega_{t-1})\cap A} \psi_{q}d\mu}{\int_{\mathcal{T}_{L-1}(\omega_{t-1})} \psi_{q}d\mu},
\end{eqnarray}
for $A\subseteq\mathbb{R}_{+}$ measurable set. Also, suppose $\mathcal{D}(\omega)\subset\mathbb{R}^{4}_{++}$ is compact, $\mathbb{P}(\Vert s_{t}\Vert_{\infty}=1)=1$, $t\geq1$, and that all utility functions are attractive. Then, $\mathbb{P}(\lim_{t\rightarrow\infty} \omega_{t}\in \mathcal{P})=1$.
\end{theorem}

Theorem \ref{theoStochasticFirstWelfare} states that a ``maximum speed'' Edgeworth process converges to the manifold of Pareto optimal allocations with full probability, and its core assumption is given by (\ref{eqTheoStochastic}). This assumption, in turn, is motivated by the attraction principle and can be more readily understood if one imagines an Edgeworthian auctioneer picking normalized prices that lie on the floor of the flat domain.

First, notice that Definition \ref{defStochasticNTProcess} implies that $p_{t}\in\mathcal{T}(\omega_{t-1})$, in each period $t\geq1$ with $\omega_{t-1}\notin\mathcal{P}$, so that the Edgeworthian auctioneer is picking a point inside $\mathcal{T}(\omega_{t-1})$ unless a Pareto optimal allocation has been reached. Then, the natural question is: how is the Edgeworthian auctioneer choosing normalized prices in $\mathcal{T}(\omega_{t-1})$?

The answer is given by (\ref{eqTheoStochastic}). The Edgeworthian auctioneer's behavior, given the current trade-compatible state of the stochastic non-tâtonnement process $\omega_{t-1}\notin\mathcal{P}$, is described by a probability distribution on the floor of the flat domain (i.e., $\mathbb{R}^{L-1}_{++}$) defined through the strictly positive integrable function $\psi_{q}\in L^{1}_{++}(\mathbb{R}^{L-1}_{++})$. 

Clearly, due to the attraction principle, as trade unfolds, the region on which the Edgeworthian auctioneer must focus his attention shrinks. Furthermore, (\ref{eqTheoStochastic}) implies that the probability distribution representing the auctioneer's choice is supported on $\mathcal{T}_{L-1}(\omega_{t-1})$, and Figure \ref{FigCommonDensity} depicts its evolution in the Edgeworth box.

\begin{figure}[h]
\centering
\includegraphics[width=\linewidth]{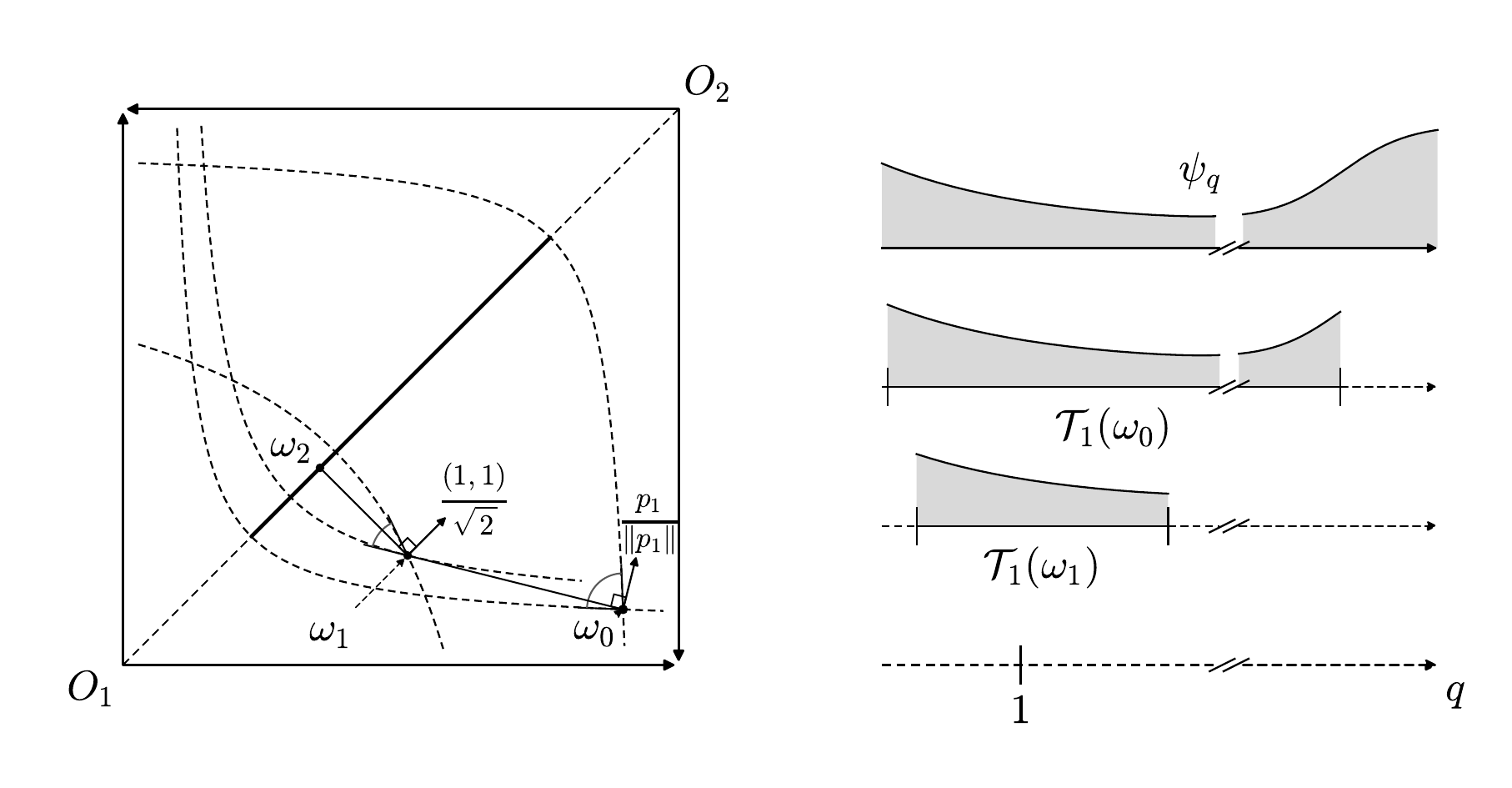}
\caption{Trade unfolding from $\omega_{0}$ toward the contract curve in the Edgeworth box (left), and the common density $\psi_{q}(\cdot)$ restricted to the shrinking trade-compatible-price sets $\mathcal{T}_{1}(\omega_{0})$ and $\mathcal{T}_{1}(\omega_{1})$ as trade progresses (right).}
\label{FigCommonDensity}
\end{figure}

A second remark about Example \ref{ex5} is in order. It reveals that the proper definition of the normalized prices process and the relative trade speeds process is the core of any stochastic non-tâtonnement process, but analytical definitions are challenging, even in an Edgeworth box. Numerical simulations through Monte Carlo methods are, therefore, a natural way to proceed. 

To do so, we start with the following definition based on (\ref{eqTheoStochastic}), where the sum of two sets is the Minkowski one\footnote{The Minkowski sum of two sets $A$ and $B$ is given by $A+B=\{a+b\mid a\in A,b\in B\}$.}, and $B(\varepsilon)$ is the open ball of radius $\varepsilon>0$ of the corresponding ambient space.

\begin{definition}\label{defCommonDensityProcess}
    Let $(X,\mathcal{B},\mathbb{P})$ be a probability space and $\{\omega_{t},p_{t+1},s_{t+1}\}_{t\geq0}$ an Edgeworth process. Then, $\{\omega_{t},p_{t+1},s_{t+1}\}_{t\geq0}$ is a \textit{common density process} if there are $\psi_{q}\in L^{1}_{++}(\mathbb{R}^{L-1}_{++})$ and $\psi_{s}\in L^{1}_{++}([0,1]^{H})$ such that
    \begin{eqnarray}
        \mathbb{P}(p_{t}\in A\times\{1\}\mid \omega_{t-1},\omega_{t-1}\notin\mathcal{P})&=&\lim_{\varepsilon\rightarrow0}\frac{\int_{(\mathcal{T}_{L-1}(\omega_{t-1})\cap A)+B(\varepsilon)} \psi_{q}d\mu}{\int_{\mathcal{T}_{L-1}(\omega_{t-1})+B(\varepsilon)} \psi_{q}d\mu}\label{eqEdgeAuctionPrices}
        \\
        \mathbb{P}(s_{t}\in A^{\prime}\mid \omega_{t-1},p_{t},\omega_{t-1}\notin\mathcal{P})&=&\lim_{\varepsilon\rightarrow0}\frac{\int_{(\mathcal{S}(\omega_{t-1},p_{t})\cap A^{\prime})+B(\varepsilon)} \psi_{s}d\mu}{\int_{\mathcal{S}(\omega_{t-1},p_{t})+B(\varepsilon)} \psi_{s}d\mu},\label{eqEdgeAuctionSpeeds}
    \end{eqnarray}
for $t\geq1$, and $A\subseteq\mathbb{R}^{L-1}_{++}$, $A^{\prime}\subseteq[0,1]^{H}$ measurable sets.    
\end{definition}

The economic interpretation of (\ref{eqEdgeAuctionPrices}) is analogous to that of (\ref{eqTheoStochastic}). Nonetheless, the limit in (\ref{eqEdgeAuctionPrices}) (and in (\ref{eqEdgeAuctionSpeeds})) requires explanation. If $\mathcal{T}_{L-1}(\omega_{t-1})$ has positive measure and a boundary of measure zero, the Dominated Convergence Theorem allows us to write
\begin{eqnarray*}
   \lim_{\varepsilon\rightarrow0}\frac{\int_{\mathcal{T}_{L-1}(\omega_{t-1})\cap A+B(\varepsilon)} \psi_{q}d\mu}{\int_{\mathcal{T}_{L-1}(\omega_{t-1})+B(\varepsilon)} \psi_{q}d\mu}=\frac{\int_{\mathcal{T}_{L-1}(\omega_{t-1})\cap A} \psi_{q}d\mu}{\int_{\mathcal{T}_{L-1}(\omega_{t-1})} \psi_{q}d\mu}.
\end{eqnarray*}
However, if $\mathcal{T}_{L-1}(\omega_{t-1})$ has zero measure, then the right-hand side is not well defined, and this is why I write (\ref{eqEdgeAuctionPrices}) as the limit on the left-hand side. Clearly, I am assuming (and thus not proving) that when $\mathcal{T}_{L-1}(\omega_{t-1})$ has measure zero, this limit will indeed exist. Nonetheless, this is, I believe, a weak assumption since, in most reasonable cases, $\mathcal{T}_{L-1}(\omega_{t-1})$ will have a non-pathological structure (e.g., a lower dimensional manifold as in Figure \ref{FigBoxSet}).

The choice of the relative trade speeds represented by (\ref{eqEdgeAuctionSpeeds}) is analogous to the trade prices, with the difference being that instead of the floor of the flat domain, the ambient space of the probability becomes the hypercube $[0,1]^{H}$.

Finally, the processes given by Definition \ref{defCommonDensityProcess} are called \textit{common density processes} because the functions $\psi_{q}(\cdot)$ and $\psi_{s}(\cdot)$ (which can be seen, without loss of generality, as densities) used to obtain the corresponding probability measures on the floor of the flat domain and in the hypercube do not depend on $t\geq1$ and therefore encompass all the information needed about the relative chances between possible normalized prices and trade speeds. 

This feature, therefore, makes the numerical simulation of a common density process considerably more tractable than that of a general Edgeworth process. The next example illustrates this.

\begin{example}\label{ex6}
    I build on Examples \ref{ex1}--\ref{ex3}. Let $H=2$ and $u^{1}(c)=u^{2}(c)=c_{1}c_{2}$, $c\in\mathbb{R}^{2}_{+}$. Consider that, initially, the economy is in equilibrium, with both households holding the consumption bundle $(2,2)\in\mathbb{R}^{2}_{++}$, and the equilibrium prices are given by $(1,1)\in\mathbb{R}^{2}_{++}$. However, after a shock, the allocation becomes $\omega=(\omega^{1},\omega^{2})$, $\omega^{1}=(3.6,0.4)$, and $\omega^{2}=(0.4,3.6)$. Notice, in particular, that the pre- and post-shock Walrasian equilibrium prices coincide, both being $(1,1)\in\mathbb{R}^{2}_{++}$; this follows from the symmetric choice of endowments and is meant only to ease the exposition in the square Edgeworth box.
    
    Trade under fluctuating market prices unfolds according to a common density process $\{\omega^{\delta}_{t}\}_{t\geq0}$ starting at $\omega\in\mathbb{R}^{4}_{++}$, with $\psi_{q}:\mathbb{R}_{+}\rightarrow\mathbb{R}_{+}$ given by
\begin{eqnarray}\label{eqEx6EdgeworthAuc}
    \psi_{q}(q)=\frac{1}{\arctan M-\arctan M^{-1}}\frac{\mathbf{I}_{[M^{-1},M]}(q)}{1+q^{2}}
\end{eqnarray}
for $q>0$, $M\gg1$, and $\psi^{\delta}_{s}:[0,1]^{2}\rightarrow\mathbb{R}_{++}$ given by 
\begin{eqnarray*}
    \psi^{\delta}_{s}(s)=\frac{\mathbf{I}_{[\delta,1]\times[0,1]}(s)+\mathbf{I}_{[0,1]\times[\delta,1]}(s)-\mathbf{I}_{[\delta,1]\times[0,1]}(s)\mathbf{I}_{[0,1]\times[\delta,1]}(s)}{1-\delta^{2}},
\end{eqnarray*}
for $s\in[0,1]^{2}$, $\delta\in[0,1)$\footnote{Notice that $\psi_{q}\notin L^{1}_{++}(\mathbb{R}_{++})$ and $\psi^{\delta}_{s}\notin L^{1}_{++}([0,1]^{2})$, $\delta\in[0,1)$. Nonetheless, since this is an Edgeworth box example, the common density processes are still well defined for a suitable choice of $M\gg1$.}. Notice that $\psi_{q}(\cdot)$ represents a uniform draw by the Edgeworthian auctioneer in the ``arc of trade-compatible angles (prices).'' Furthermore, $\psi^{\delta}_{s}(\cdot)$, $\delta\in[0,1)$, represents random draws in which, according to $\delta\in[0,1)$, we are excluding some ``low speed'' linear joint trade paths. Therefore, $\lim_{\delta\rightarrow1}\omega^{\delta}_{t}$ can be seen as the common density process that always moves at ``maximum speed'' (i.e., always exhausts all trade possibilities under current prices). Figure \ref{FigHistoCommonDensity} represents the distribution over the contract curve for the polar cases $\delta=0$ and $\delta\rightarrow1$. 

\begin{figure}[h]
\centering
\includegraphics[width=\linewidth]{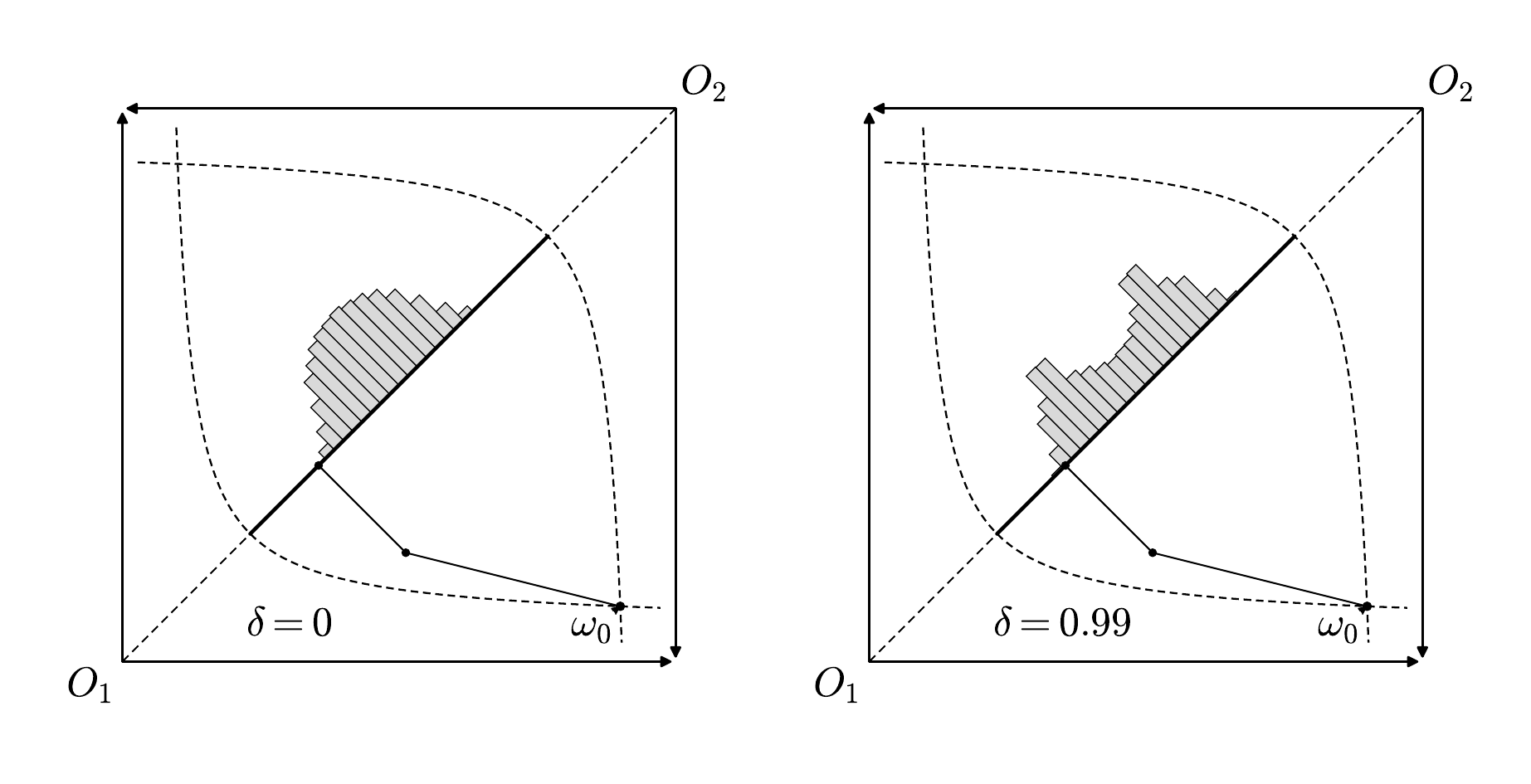}
\caption{Terminal distribution of the common density process over the contract curve for $\delta=0$ (left) and $\delta=0.99$, approximating the polar case $\delta\rightarrow1$ (right), unimodal and bimodal, respectively.}
\label{FigHistoCommonDensity}
\end{figure}
\end{example}

Example \ref{ex6} requires two remarks. First, although the equilibrium of this economy (i.e., the endpoint of the maximal linear joint trade path starting at $\omega_{0}\in\mathbb{R}^{HL}_{++}$) is the average outcome in both polar cases, in the second (bimodal) case, it is not the most probable outcome of non-tâtonnement trade. 

This is a reminder that the final distribution over the contract curve does not necessarily bear any specific relation to a Walrasian equilibrium of the economy (i.e., to a maximal linear joint trade path). However, this does not rule out Walrasian equilibria as a reasonable forecast for the outcome of trade under perfect competition. This is a fundamental remark, so let me elaborate on it.

Think of a market in which perishable consumption goods are traded, starting every day with the same allocation $\omega_{0}\in\mathbb{R}^{HL}_{++}$. The probability distribution on the contract curve in Figure \ref{FigHistoCommonDensity} can be seen as a histogram drawn after thousands of days of trade under \textit{the same common densities}. It is as if the unseen economic engines behind the trade process were \textit{stationary} (i.e., they do not change from one market opening to the next); therefore, the Edgeworthian auctioneer keeps following, day after day, the same instruction manual.

Is this stationarity plausible in this economic scenario? I believe not for three main reasons. First, since goods are perishable and the market functioning has time limitations, there is clearly a limit on the number of ``steps'' the Edgeworth process can have. Therefore, it seems desirable to trade under the smallest possible number of steps and with the largest possible speed in each of them to get as close as possible to a Pareto optimal allocation. 

Second, price fluctuations may be undesirable for several reasons. For instance, due to menu costs (e.g., \textcite{Mankiw_1985,GolosovLucas_2007}) or the fact that they may create strategic concerns among participants in the market (e.g., speculative behavior and unnecessary queuing).

Third, one can think along the lines of \textcite{DebreuScarf_1963} and \textcite{Aumann_1964} to see that in a large economy (i.e., with many households of each type), although coalitions cannot undo trades already completed within a market day, they can drive the unseen economic engines from day to day.

These three points indicate that a maximal linear joint trade path (i.e., a Walrasian equilibrium) is the natural candidate to represent the limiting point for the \textit{evolution} of the Edgeworthian auctioneer's manual (i.e., how the common densities change over time once the market keeps reopening under the same conditions). The precise way in which this evolution happens, however, is beyond the scope of this paper.

The other remark is as follows. Notice that in Figure \ref{FigHistoCommonDensity}, the common density processes lead to a probability distribution over the contract curve which, in turn, can be used to derive confidence intervals for comparative statics and predictions of the economic model (e.g., when evaluating welfare gains from a certain policy, one would take into account this entire distribution over the contract curve). However, although these distributions are drawn over the contract curve, one needs to be careful to ensure that no probability mass is forgotten outside this locus. In this case, the check was made numerically.  

The next example reveals how common density processes can model price stickiness.

\begin{example}[Price stickiness]\label{ex7}
   I build on Example \ref{ex6}. This time, however, let
\begin{eqnarray}\label{eqEx7EdgeworthAuc}
    \psi_{q}(q)=\frac{\exp(-(\arctan   q-\pi/4)^{2}/(2\sigma^{2}))}{(\int^{M}_{M^{-1}}\exp(-(\arctan   x-\pi/4)^{2}/(2\sigma^{2}))dx/(1+x^{2}))}\frac{\mathbf{I}_{[M^{-1},M]}(q)}{1+q^{2}}
\end{eqnarray}
for $q>0$, $\sigma>0$, $M\gg1$, and $\psi_{s}=1$, $s\in[0,1]^{2}$. Notice that (\ref{eqEx7EdgeworthAuc}) can be seen as a normal (i.e., Gaussian) draw in the ``arc of trade-compatible angles (prices),'' with the average of the distribution given by $\pi/4$ (i.e., equivalent to the equilibrium price $(1,1)$) and the standard deviation given by $\sigma>0$. In particular, by taking the limit $\sigma\rightarrow\infty$, (\ref{eqEx7EdgeworthAuc}) becomes (\ref{eqEx6EdgeworthAuc}).

Therefore, the value of $\sigma>0$ can be seen as an inverse measure of the degree of price stickiness of the economy to the original equilibrium price. Figure \ref{FigHistoSticky} shows how, as it increases, the terminal distribution over the contract curve becomes less concentrated around the allocation associated with the original equilibrium price $(1,1)$. Notice, in particular, that the graphics on the left of Figure \ref{FigHistoCommonDensity} and on the right of Figure \ref{FigHistoSticky} are fairly similar.

\begin{figure}[h]
\centering
\includegraphics[width=\linewidth]{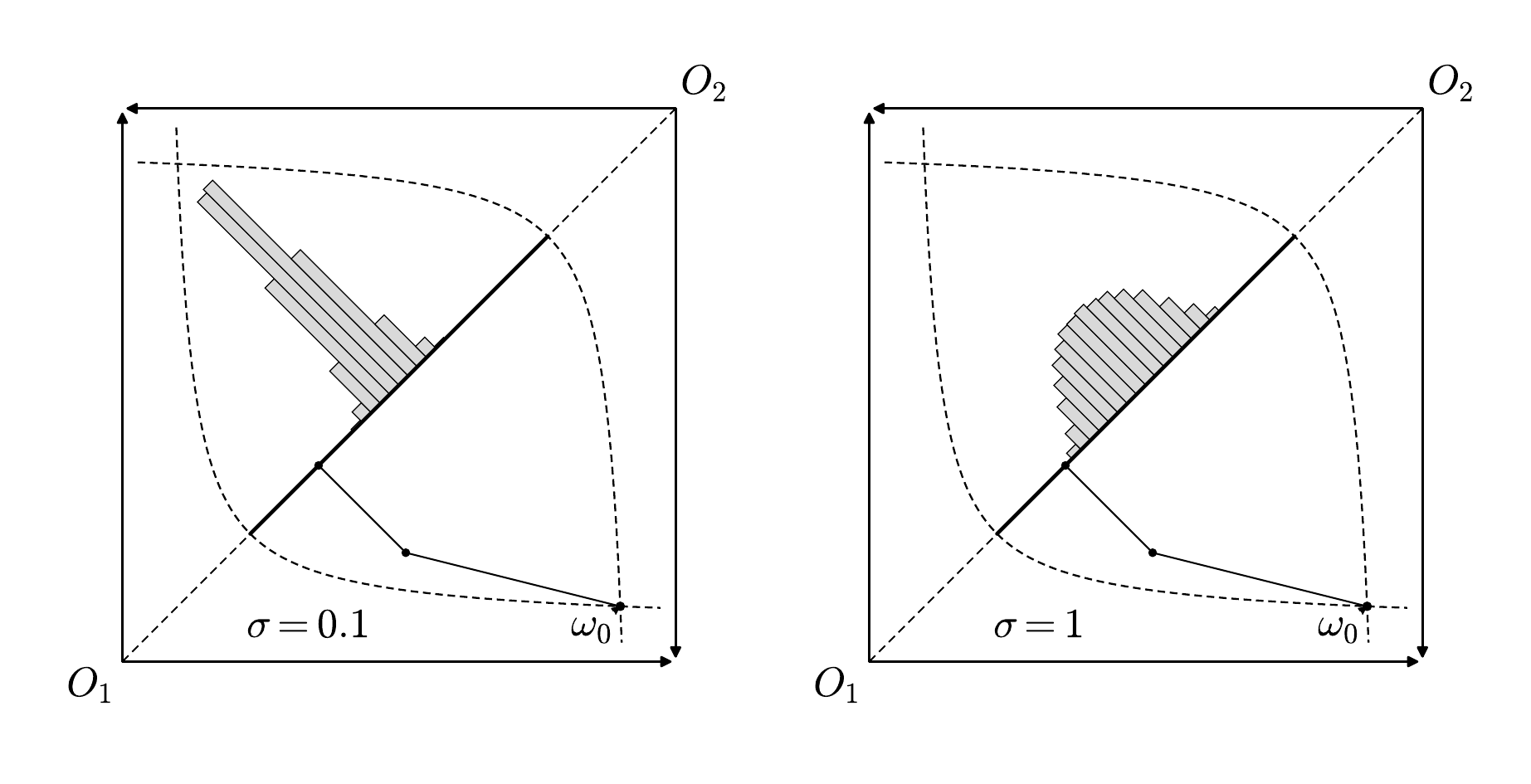}
\caption{Terminal distribution of the common density process over the contract curve for $\sigma=0.1$ (left) and $\sigma=1$ (right): as $\sigma$ increases, the distribution becomes less concentrated around the allocation associated with the original equilibrium price $(1,1)$.}
\label{FigHistoSticky}
\end{figure}
\end{example}

One needs to pay careful attention to Example \ref{ex7}, since its modeling of price stickiness is fundamentally different\footnote{Let me also point out that this form of price stickiness differs from those found in the literature on menu costs (e.g., \textcite{Mankiw_1985,GolosovLucas_2007}), search models (e.g., \textcite{BurdettJudd_1983,HeadLiuMenzioWright_2012}), and disequilibrium models (e.g., \textcite{BarroGrossman_1971,Dreze_1975,Benassy_1975}). Furthermore, although trade happens out of a Walrasian equilibrium both in Edgeworth processes and in the model of \textcite{Dreze_1975}, the fundamental difference is that the former involve no rationing: at each step, households trade voluntarily (i.e., along utility-enhancing linear trade paths) at the pace set by the relative trade speeds and aiming at their most desired bundle, without any sort of \textit{rationing scheme} on their budget sets \parencite[p. 303]{Dreze_1975}.} from the well-established New Keynesian model (e.g., \textcite{Calvo_1983,Clarida_1999,Woodford_2003,Christiano_2005,Gali_2015}). The New Keynesian model departs from the classical tradition in two ways \parencite[Chapter 3]{Gali_2015}: (i) imperfect competition in the goods market due to monopolistically competitive firms; and (ii) a \textcite{Calvo_1983} friction that allows only a fraction of firms to adjust their prices in each period.

In order to make a comparison between Example \ref{ex7} and this New Keynesian model, I need to introduce a third type of auctioneer. This New Keynesian auctioneer calls prices $\{p_{t}\}_{t\geq1}$\footnote{In order to keep the notation of the paper consistent, I decided to use $\{p_{t}\}_{t\geq1}$ to represent prices called by the New Keynesian auctioneer. One must bear in mind that this (imprecise) notation represents all goods prices and wages in all periods and states of nature of the economy encompassed by the New Keynesian model \parencite[Chapter 3]{Gali_2015}.} and will only allow trade to happen if all markets clear and the following two conditions are satisfied: (i) firms agree with the prices chosen for their goods (i.e., the auctioneer's choice perfectly represents all firms); and (ii) in each period, only a fraction of goods prices is adjusted. 

Notice that the description above does not mention how the New Keynesian auctioneer searches for prices (i.e., there is no equation analogous to (\ref{eqSamuelsonAuctioneer}) or (\ref{eqNonTatoPrice})), since this is not necessary for our comparison purposes. The first fundamental point to be highlighted is that the New Keynesian and the Walrasian auctioneers share a common characteristic: both allow trade to happen only under prices that clear all markets\footnote{In the New Keynesian model, not all prices fluctuate freely to equate supply and demand (output is demand-determined at posted prices), so when I say ``prices that clear all markets,'' I am bundling together prices driven only by perfect competition (e.g., wages in the New Keynesian model) and prices chosen by firms under the monopolistic competition and price stickiness setup (notice that this mechanism is represented by the agreement of the firms with the prices picked by the New Keynesian auctioneer in accordance with the price stickiness constraint).}. Notice, however, that market clearing implies, for the Walrasian auctioneer, that one hits the Pareto optimal manifold following the corresponding maximal linear joint trade path; whereas for the New Keynesian auctioneer, the frictions make the endpoint of the linear joint trade path, in general, an inefficient allocation\footnote{Since the New Keynesian model represents an economy with production, one must see a linear joint trade path as one in which each firm follows (under a convexity assumption) the line segment connecting the origin to its final production plan.}.

The second fundamental point to be highlighted is that the New Keynesian auctioneer's ``final'' price choice is not only dictated by market clearing, but it must also satisfy one condition related to the market power of firms and another related to price stickiness.

Therefore, up to this point, we can establish the following difference between price stickiness in Example \ref{ex7} and the New Keynesian model. While the former deals with trade out of equilibrium (i.e., without market clearing) and price stickiness is related to how the Edgeworthian auctioneer picks successive normalized prices, the latter deals with trade in equilibrium (i.e., all markets clear) and price stickiness is a constraint on the prices under which the New Keynesian auctioneer allows trade to happen (i.e., in each period, only a fraction of goods prices can change).

There is still a more subtle difference related to the interpretation of time in these two frameworks. In the New Keynesian model, prices are determined for the entire infinite time horizon of the economy\footnote{Due to uncertainty, these prices are actually stochastic processes. For the sake of the argument, I consider the perfect-foresight case here.}. If we consider trade happening sequentially (e.g., \textcite{Radner_1972}), then in every period, we can think of a linear joint trade path being followed. This path, however, represents only the exchange of goods available in the current period\footnote{One could also think of futures contracts being traded (or contingent commodities \parencite{Arrow_1964} in a stochastic framework), but this would just complicate the point I am trying to make. The relevant matter is that stochastic non-tâtonnement processes, in this paper, model the exchange of a \textit{finite number} of goods (i.e., $L<\infty$), and thus cannot encompass all goods that exist in an infinite-horizon economy.}, with (rational) expectations of future consumption streams kept fixed according to the desired consumption bundles derived from the equilibrium prices.

In Example \ref{ex7}, although the stochastic non-tâtonnement process $\{\omega_{t}\}_{t\geq0}$ is written with a $t\geq0$ index amounting to ``time,'' one must be cautious regarding the interpretation of this index when comparing it to the New Keynesian model. Notice that this $t\geq0$ index certainly represents the ``successive steps'' of the Edgeworth process, so there is a natural time flow interpretation here. But this flow does not need to coincide with the scale of the New Keynesian framework. 

Allow me to elaborate on this point by suitably redefining indices. Let $\{\omega_{ts}\}_{s\geq0}$, $t\geq1$, be a sequence of stochastic non-tâtonnement processes, with $\{p_{ts}\}_{s\geq1}$, $t\geq1$, the corresponding normalized prices processes, where $t\geq1$ stands for ``time'' and $s\geq 0$ stands for ``steps.'' We seek to interpret $\{\omega_{ts}\}_{s\geq0}$ as the Edgeworth process that represents currently traded goods in period $t\geq1$. If one thinks of this process simply as a linear joint trade path under prices $\{p_{t}\}_{t\geq1}$ after which trade stops for the period, then this can be seen as the New Keynesian setup described above, with price stickiness represented as a constraint on the values of the equilibrium price sequence $\{p_{t}\}_{t\geq1}$.

To ease the comparison, assume that $\{p_{t}\}_{t\geq1}$ is stationary, so that in each period $t\geq1$ the process follows the same linear joint trade path. What happens if a shock (i.e., an unforeseen change in the fundamentals) hits the economy at time $t\geq1$?

A first possibility is to say that the New Keynesian auctioneer would look, given the updated fundamentals, for a new equilibrium that satisfies the market power and price stickiness constraints described above. In this case, households only trade after the New Keynesian auctioneer actually finds the new equilibrium price.

Example \ref{ex7}, however, reveals a second possibility. Since the market has historically traded at prices $\tilde{p}=p_{t-1}$, there could be some stickiness at this price level driving trade after the shock hits the economy. Notice, in particular, that in Example \ref{ex7} the Walrasian equilibrium price of the \textit{post-shock economy} coincides with the historical $\tilde{p}=(1,1)$ reference, but this need not be the case. Therefore, in Example \ref{ex7}, stickiness concentrates the terminal distribution around the Walrasian allocation, whereas, when the two prices differ (i.e., the historical and the post-shock equilibrium), it pulls the terminal distribution away from it.

Furthermore, the shock does not prevent trade from occurring out of equilibrium after it (i.e., households will not wait for the New Keynesian auctioneer to find a new equilibrium price and will just trade according to the price and speed draws of the Edgeworthian auctioneer\footnote{With this description in mind, one could also think of imposing that the first draw of the Edgeworthian auctioneer would have to be $p_{t1}=\tilde{p}$ (i.e., the historical price level), with highest possible trade speed in order to exhaust trade under the historical price, provided $\tilde{p}\in\mathcal{T}(\omega_{t0})$. Then, the Edgeworth process would follow a variant of the common density setup of Definition \ref{defCommonDensityProcess} with the distribution of normalized prices concentrated around the historical value to represent the stickiness.}).

This \textit{price inertia} can be modeled by the type of probability distribution concentrated around the historical level illustrated in Example \ref{ex7}, from which the Edgeworthian auctioneer will draw the prices that conduct the steps of the process $\{\omega_{ts}\}_{s\geq1}$.

The attentive reader will realize that this last description of the price-sticky Edgeworth process unfolding in period $t$ did not describe how the shock affected households' expectations (i.e., the prospects of future prices and, therefore, future trade). Due to the dynamic nature of this economy, this matter is central and depends on several factors (e.g., whether the shock hits only the current period or all periods). The precise way in which stochastic non-tâtonnement processes and expectations are connected, however, is beyond the scope of this paper.

Next, in order to numerically simulate a common density process, Definition \ref{defCommonDensityProcess} requires the calculation, at every step of the way: first, of the current set of trade-compatible prices $\mathcal{T}(\cdot)$; then, of a draw from (\ref{eqEdgeAuctionPrices}); then, of the set of relative trade speeds $\mathcal{S}(\cdot)$; and finally, of a draw from (\ref{eqEdgeAuctionSpeeds}) to obtain the next step. Although in Examples \ref{ex6} and \ref{ex7} this calculation procedure has been simplified by taking $L=2$, for larger dimensions it can be computationally demanding. 

Therefore, the remainder of this section is directed towards results that can enhance the performance of numerical simulations of common density processes. I proceed with the following lemma, whose notation mimics the setting of Definition \ref{defCommonDensityProcess}.

\begin{lemma}\label{lemmaLargerDomainQ}
    Let $(X,\mathcal{B},\mathbb{P})$ be a probability space, $\mathcal{T}_{L-1},\mathcal{V}\subseteq\mathbb{R}^{L-1}_{++}$ measurable sets, $\mathcal{T}_{L-1}\subseteq \mathcal{V}$, $\mu(\mathcal{T}_{L-1})>0$, $\psi_{q}\in L^{1}_{++}(\mathbb{R}^{L-1}_{++})$, and $q:X\rightarrow \mathbb{R}^{L-1}_{++}$ a random variable with $\mathbb{P}(q\in A)=\int_{\mathcal{V}\cap A}\psi_{q} d\mu/\int_{\mathcal{V}}\psi_{q} d\mu$, for $A\subseteq\mathbb{R}^{L-1}_{++}$ measurable set. Then, $\mathbb{P}(q\in\mathcal{T}_{L-1})>0$ and $\mathbb{P}(q\in A \mid q\in \mathcal{T}_{L-1})=\int_{\mathcal{T}_{L-1}\cap A}\psi_{q} d\mu/\int_{\mathcal{T}_{L-1}}\psi_{q} d\mu$, for $A\subseteq\mathbb{R}^{L-1}_{++}$ measurable set.
\end{lemma}

Lemma \ref{lemmaLargerDomainQ} tells us that in order to simulate a random variable whose distribution is in the form of (\ref{eqEdgeAuctionPrices}), one can actually take draws by adopting larger \textit{box sets} (i.e., sets over which the integrals on the right-hand side of (\ref{eqEdgeAuctionPrices}) are taken) and only keep the first that actually ``falls'' within the baseline reference set. Therefore, instead of calculating the entire baseline reference set, one only needs to check if a given draw lies within it.

When simulating a random draw from (\ref{eqEdgeAuctionPrices}), (\ref{eqBoxSet1}) leads us to choose $\mathcal{V}(\omega_{t-1})=\mathcal{F}(\omega_{t-1})\supseteq \mathcal{T}_{L-1}(\omega_{t-1})$ (i.e., the set of dominant floor coordinates is taken as the box set). However, calculating $\mathcal{F}(\omega_{t-1})$ can be no less demanding than calculating $\mathcal{T}_{L-1}(\omega_{t-1})$. Therefore, if possible, we should aim for better behaved box sets to simulate a common density process by Monte Carlo, and the following classical definition moves in this direction (cf. \textcite[p. 89, Lemma 3]{ArrowHurwicz_1959} and \textcite[p. 653, Lemma 3]{Negishi_1962}).

\begin{definition}\label{defSharpUtility}
 Let $u:\mathbb{R}^{L}_{+}\rightarrow\mathbb{R}_{+}$ be a utility function satisfying all assumptions from Section \ref{sec2}. Then, $u(\cdot)$ is \textit{sharp} if
\begin{eqnarray*}
    p_{i}>\max_{j\neq i}\frac{p_{j}\nabla_{i}u(c)}{\nabla_{j}u(c)}\implies x_{ni}\biggl(\frac{p}{p\cdot c}\biggr)-c_{i}<0,
\end{eqnarray*}
and
\begin{eqnarray*}
    p_{i}<\min_{j\neq i}\frac{p_{j}\nabla _{i}u(c)}{\nabla_{j}u(c)}\implies x_{ni}\biggl(\frac{p}{p\cdot c}\biggr)-c_{i}>0,
\end{eqnarray*}
for $c,p\in\mathbb{R}^{L}_{++}$, $1\leq i\leq L$.
\end{definition}

Definition \ref{defSharpUtility} simply states that if the price of a good, given the households' marginal rates of substitution, is unambiguously low (or high), then this good must be demanded (or supplied). The next result is analogous to Proposition \ref{propIsAttractive}.

\begin{proposition}\label{propIsSharp}
Let $u:\mathbb{R}^{L}_{+}\rightarrow \mathbb{R}_{+}$ be given by $u(c)=(\sum^{L}_{i=1}\alpha_{i}c_{i}^{\sigma})^{\frac{1}{\sigma}}$, $\sigma<1$, $\sigma\neq0$, or $u(c)=\prod^{L}_{i=1}c_{i}^{\alpha_{i}}$, $\alpha\in\mathbb{R}^{L}_{++}$. Then $u(\cdot)$ is sharp.
\end{proposition}

Before proceeding, it is worth noting that being attractive or sharp are, in fact, properties of the underlying preference relation, and not of the particular utility function chosen to represent it.

\begin{lemma}\label{lemmaUtilitiesPreferencesAttractive}
Let $u:\mathbb{R}^{L}_{+}\rightarrow\mathbb{R}_{+}$ satisfy all assumptions of Section \ref{sec2}, and $g:\mathbb{R}_{+}\rightarrow\mathbb{R}_{+}$ be smooth with a strictly positive derivative. If $u(\cdot)$ is attractive, then $g\circ u$ is attractive. If $u(\cdot)$ is sharp, then $g\circ u$ is sharp.
\end{lemma}

Lemma \ref{lemmaUtilitiesPreferencesAttractive} follows from the fact that the normalized Walrasian demand $x_{n}(\cdot)$ is left unchanged by the composition $g\circ u(\cdot)$ (i.e., $u(\cdot)$ and $g\circ u(\cdot)$ represent the same preferences), together, for the attractive case, with $g^{\prime}(\cdot)>0$. Therefore, both Definition \ref{defAttractiveUtility} and Definition \ref{defSharpUtility} are, indeed, well posed as properties of preferences.

To proceed, a somewhat burdensome notation is needed. Let $m_{ij}:\mathbb{R}^{HL}_{++}\rightarrow \mathbb{R}_{++}$ and $M_{ij}:\mathbb{R}^{HL}_{++}\rightarrow \mathbb{R}_{++}$, $1\leq i,j\leq L$, $i\neq j$, be given by 
\begin{eqnarray*}
    m_{ij}(\omega)&=&\min_{1\leq h\leq H}\frac{\nabla_{i}u^{h}(\omega^{h})}{\nabla_{j}u^{h}(\omega^{h})}\\
    M_{ij}(\omega)&=&\max_{1\leq h\leq H}\frac{\nabla_{i}u^{h}(\omega^{h})}{\nabla_{j}u^{h}(\omega^{h})}
\end{eqnarray*}
for $\omega\in\mathbb{R}^{HL}_{++}$, so that $m_{ij}(\omega)=M_{ji}(\omega)^{-1}$. These functions furnish, for a given allocation, the minimum and maximum values attained by the marginal rates of substitution. Also, let $\Pi(\omega)\subset\mathbb{R}^{L-1}_{++}$, $\omega\in\mathbb{R}^{HL}_{++}$, be given by
\begin{eqnarray*}
    \Pi(\omega)&=&\{q\in\mathbb{R}^{L-1}_{++}\mid \min_{j\neq i}p_{j} m_{ij}(\omega)\leq p_{i}\leq \max_{j\neq i}p_{j}M_{ij}(\omega), 1\leq i\leq L,p=(q,1)\}\\
    &=&\bigcap^{L}_{i=1}\bigcup_{j,k\neq i}\biggl(\Pi^{-}_{ij}(\omega)\bigcap \Pi^{+}_{ik}(\omega)\biggr),
\end{eqnarray*}
with its ``strictly-positive-half-spaces'' $\Pi^{-}_{ij},\Pi^{+}_{ik}\subset\mathbb{R}^{L-1}_{++}$, given by 
\begin{eqnarray*}
    \Pi^{-}_{ij}(\omega)&=&\{q\in\mathbb{R}^{L-1}_{++}\mid p_{j}m_{ij}(\omega)\leq p_{i},p=(q,1)\}\\
    \Pi^{+}_{ik}(\omega)&=&\{q\in\mathbb{R}^{L-1}_{++}\mid p_{i}\leq p_{k}M_{ik}(\omega),p=(q,1)\},
\end{eqnarray*}
for $1\leq i,j,k\leq L$, $j,k\neq i$. The following lemma reveals that $\Pi(\cdot)$ can be used as a box set to perform Monte Carlo simulations on a common density process.

\begin{lemma}\label{lemmaTradePricesInBox}
Let $\omega\in\mathbb{R}^{HL}_{++}$ be an allocation. If $u^{h}(\cdot)$ is sharp, $1\leq h\leq H$, then $\mathcal{T}_{L-1}(\omega)\subseteq \Pi(\omega)$.
\end{lemma}

Next, let 
\begin{eqnarray*}
    \mathcal{H}^{-}_{ij}(\omega,p)&=&\{h\in\{1,\ldots,H\}\mid  \nabla_{i} u^{h}(\omega^{h})/\nabla_{j} u^{h}(\omega^{h})\leq p_{i}/p_{j}\}\\
    \mathcal{H}^{+}_{ij}(\omega,p)&=&\{h\in\{1,\ldots,H\}\mid p_{i}/p_{j}\leq \nabla_{i} u^{h}(\omega^{h})/\nabla_{j} u^{h}(\omega^{h})\}
\end{eqnarray*}
for $\omega\in\mathbb{R}^{HL}_{++}$, $p\in\mathbb{R}^{L}_{++}$, $1\leq i,j\leq L$, $i\neq j$, so that $\mathcal{H}^{-}_{ij}=\mathcal{H}^{+}_{ji}$. I proceed with the following lemma.

\begin{lemma}\label{lemmaMonotoneMm}
Let $\omega\in\mathbb{R}^{HL}_{++}$ be an allocation, $p\in\mathcal{T}(\omega)$, and $s\in\mathcal{S}(\omega,p)$. If $u^{h}(\cdot)$ is attractive, $1\leq h\leq H$, then, for $1\leq i,j \leq L$, $i\neq j$: (i) if $\mathcal{H}^{-}_{ij}(\omega,p)\neq\emptyset$, then $m_{ij}\circ c(\omega,p,s,\cdot)$ is non-decreasing; (ii) if $\mathcal{H}^{-}_{ij}(\omega,p)=\emptyset$, then $m_{ij}\circ c(\omega,p,s,\cdot)$ is non-increasing; (iii) if $\mathcal{H}^{+}_{ij}(\omega,p)\neq\emptyset$, then $M_{ij}\circ c(\omega,p,s,\cdot)$ is non-increasing; and (iv) if $\mathcal{H}^{+}_{ij}(\omega,p)=\emptyset$, then $M_{ij}\circ c(\omega,p,s,\cdot)$ is non-decreasing.
\end{lemma}

Lemma \ref{lemmaMonotoneMm} is a necessary step towards our last result.

\begin{theorem}\label{theoNested}
Let $\omega\in\mathbb{R}^{HL}_{++}$ be an allocation, $p=(q,1)\in\mathcal{T}(\omega)$, and $s\in\mathcal{S}(\omega,p)$. If $u^{h}(\cdot)$ is attractive and sharp, $1\leq h\leq H$, then $q\in\mathcal{T}_{L-1}(c(\omega,p,s,t))\subseteq \Pi(c(\omega,p,s,t))\cap \mathcal{F}(c(\omega,p,s,t))$, for $0\leq t<1$, and $q\in \Pi(c(\omega,p,s,1))\cap \mathcal{F}(c(\omega,p,s,1))$. Furthermore,  
\begin{eqnarray*}
    \Pi(c(\omega,p,s,t))=\bigcap^{L}_{i=1}\bigcup_{j,k\neq i}\biggl(\Pi^{-}_{ij}(c(\omega,p,s,t))\bigcap \Pi^{+}_{ik}(c(\omega,p,s,t))\biggr)
\end{eqnarray*}
for $0\leq t\leq1$, with all strictly-positive-half-spaces either non-increasing or non-decreasing. If $\mathcal{H}^{-}_{ij}(\omega,p)\neq\emptyset$, $1\leq i,j\leq L$, $i\neq j$, then $\{\Pi(c(\omega,p,s,t))\}_{0\leq t\leq1}$ is a non-increasing net.
\end{theorem}

Theorem \ref{theoNested} furnishes a more precise description of the dynamics of the set of trade-compatible prices seen on the floor of the flat domain. In particular, it reveals that $\Pi(\cdot)$ is a well-behaved box set for Monte Carlo simulations of a common density process, since it is considerably easier to compute than $\mathcal{F}(\cdot)$. Furthermore, if a reasonable condition is met along the path (i.e., $\mathcal{H}^{-}_{ij}(\cdot)\neq\emptyset$, $1\leq i,j\leq L$, $i\neq j$), it also becomes nested as depicted below.

\begin{figure}[H]
\centering
\includegraphics[width=\linewidth]{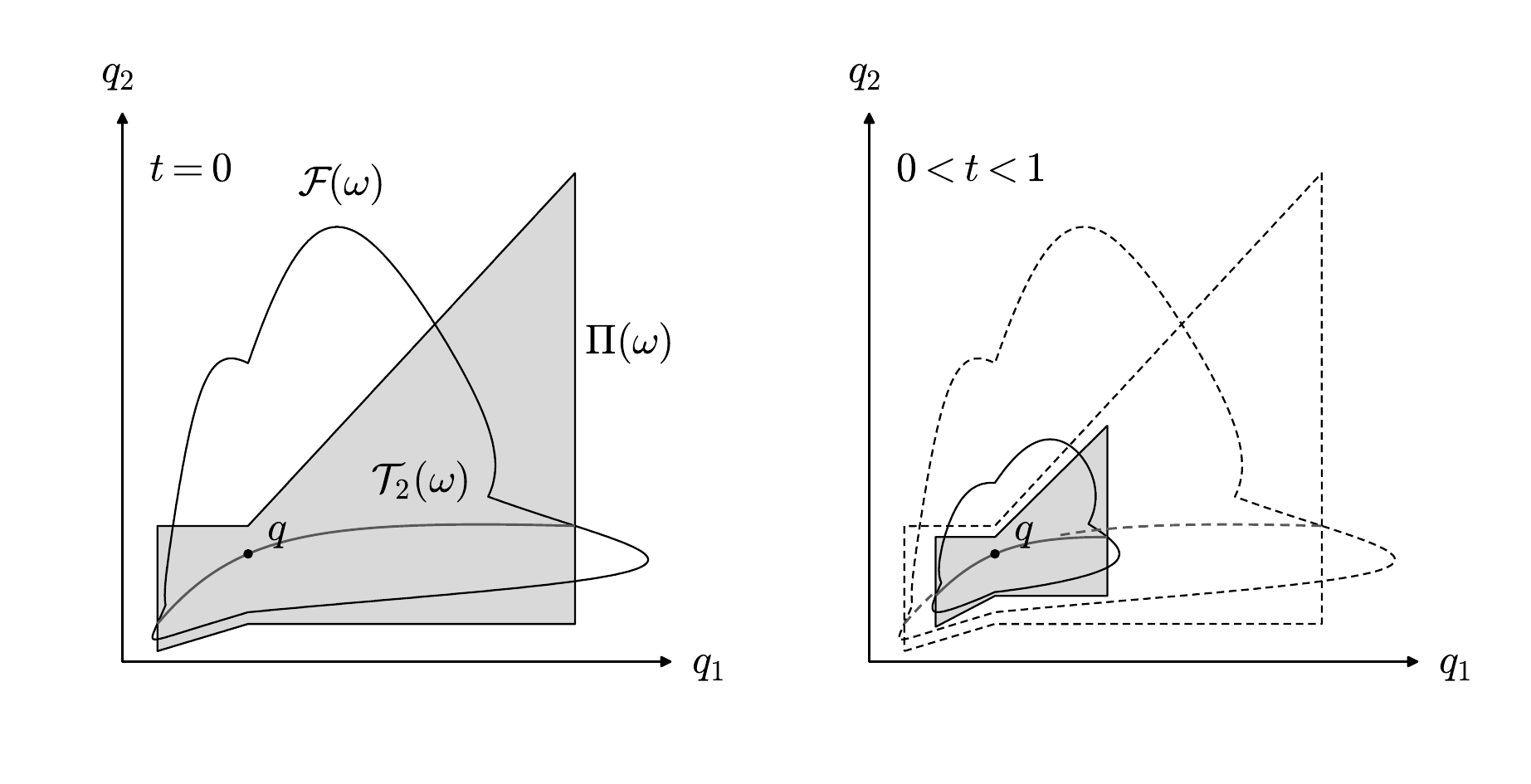}
\caption{Graphical depiction of the box set $\Pi(\cdot)$ at the beginning (left) and at an intermediate time (right) of a linear joint trade path starting at $\omega\in\mathbb{R}^{6}_{++}$ (i.e., $H=2$, $L=3$) when the final condition of Theorem \ref{theoNested} is met.}
\label{FigBoxSet}
\end{figure}

Figure \ref{FigBoxSet} depicts the dynamics of $\Pi(\cdot)$ throughout a linear joint trade path when the utility functions are attractive and sharp, and the initial allocation satisfies $\mathcal{H}^{-}_{ij}(\omega,p)\neq\emptyset$, $1\leq i,j\leq L$, $i\neq j$. When running a Monte Carlo simulation of a common density process starting at $\omega\in\mathbb{R}^{6}_{++}$ (i.e., $H=2$, $L=3$), the left graphic indicates the region over which Lemma \ref{lemmaLargerDomainQ} can be applied to find $q\in\mathcal{T}_{L-1}(\omega)$. Clearly, for every draw, one would need to further test if $\mathcal{S}(\omega,p)\neq\emptyset$. Once the ongoing trade price is defined along with the relative trade speeds $s\in\mathcal{S}(\omega,p)$, the graphic on the right depicts the box set that will be used for the random draw of the next normalized prices. 

A closer look at Figure \ref{FigBoxSet} reveals another matter. The line passing through $q\in\mathbb{R}^{2}_{++}$ is the set $\mathcal{T}_{2}(\omega)$ and, as one can see, it is a manifold of dimension one. When making a numerical simulation and taking draws in $\Pi(\omega)$, it will be impossible to hit it. This happens, in this case, because I have taken $L=3$ and $H=2$. By adding one more household, we can then get a $\mathcal{T}_{2}(\omega)$ with non-empty interior.

\section{A natural outcome of general equilibrium theory}\label{sec7}

As shown in the \hyperref[sec1]{Introduction}, the inquiries into the dynamic aspects of general equilibrium theory have been present since its earliest steps (e.g., \textcite{Walras_1874} and \textcite{Edgeworth_1881}). The results in this paper seek to tackle this matter based on three fundamental premises: (i) there exists a common prevailing market price at every moment; (ii) households can trade even out of equilibrium (i.e., under prices that do not clear all markets); and (iii) households trade aiming for their preferred affordable bundle, and they do so on the shortest possible path (i.e., on a linear trade path). 

The first two premises are well established in the non-tâtonnement literature (e.g., \textcite{Negishi_1961,Hahn_1962,Morishima_1962,Uzawa_1962,HahnNegishi_1962, Smale_1976a,Cornet_1977}). Premise (i) reflects the fact that in competitive markets, trade terms do not deviate significantly from current prices; therefore, households are ``price takers.'' Premise (ii) is derived from the observation that trade happens in ``steps'' and that households apparently do not refrain from trading at current prices because of a possible future imbalance of aggregate supply and demand. Stated differently, \textit{first we buy, then the market shelves are empty}, and this premise leads us to the realm of non-tâtonnement processes.

Apart from its simplicity, the last premise stems from the fact that trade happens in ``steps'' that enhance utility (see Lemma \ref{lemmaLinearIncreasingUtility}) and must aim at the most desired consumption bundle through, presumably, the shortest possible path. This premise leads us directly to the linear trade path paradigm given by Definition \ref{defJointLinearTradePath}, in which two fundamental sets emerge: (i) the set of normalized trade-compatible prices; and (ii) the set of relative trade speeds. 

The first represents all prices that lead to actual trade through a linear joint trade path. The second can be seen as reflecting implicit market conditions that determine trade (e.g., if one buyer is, in a given period, closer to the seller than a second buyer, this likely leads to a larger speed for the first; if a household refrains from trade for speculative reasons, this would then reduce the trade speed).

The choice of the linear trade path paradigm is also reinforced by the attraction principle (see Proposition \ref{propAttractionPrinciple}), particularly when utility functions are attractive. In turn, the attraction principle is well grounded in the concepts of flat domain and flattening diffeomorphism, and these can be naturally derived from classical demand theory after the proper definitions are made (see Definitions \ref{defNormalizedWalrasian} and \ref{defFlatteningDiffeo}). Furthermore, the symmetries between the canonical manifolds in the consumption domain and their counterparts in the normalized and flat domains (see Proposition \ref{propManifoldsDomains} and Theorems \ref{theoIndifferenceNormDom} and \ref{theoOfferFlatDom}), and the result regarding the manifold of Pareto optimal allocations (see Theorem \ref{theoParetoOptimalSet}), testify to the relevance of such definitions.

With these premises and concepts at hand, we can look at the \textit{static} results of general equilibrium theory and give them a \textit{dynamic} perspective. From this perspective, the classical equilibrium existence result \parencite{ArrowDebreu_1954,McKenzie_1954} furnishes an element of the set of normalized trade-compatible prices and, therefore, a linear joint trade path\footnote{Except, of course, if we are dealing with a non-trade or autarkic equilibrium.}. Also, the First Welfare Theorem \parencite{Arrow_1951,Debreu_1951} can be seen as a statement that this linear joint trade path is maximal. 

Therefore, from this perspective, these static results of the theory can be related to maximal linear joint trade paths, which arise from a dynamic view of non-tâtonnement trade. However, the theory does not need to be constrained to \textit{maximal} linear joint trade paths, and a natural outcome, given our three premises, is to imagine a sequence of successive linear trade paths.

One possible way to pursue this direction is to directly state the adjustment procedure that will define the successive prices and trade paths as in (\ref{eqNonTatoPrice}). Another possibility is to state that prices and trade paths are stochastic, thus leading to stochastic non-tâtonnement processes (see Definition \ref{defStochasticNTProcess}).

This stochastic setting can be seen as representing the \textit{economic engines} that are not explicitly incorporated into the model (e.g., the connection between aggregate excess demand and prices, and the physical constraints or speculative behavior that reverberate on trade speeds), and it does not prevent us from obtaining well-behaved theorems that parallel the static results of the theory (see Theorem \ref{theoStochasticFirstWelfare}). 

Furthermore, this stochastic setting brings about a fundamental change in the predictive nature of economic models. Instead of focusing on deterministic outcomes, stochastic non-tâtonnement models furnish ``confidence regions'' over the contract curve and thus can be used to evaluate, with varying degrees of certainty, where the economy will be and, perhaps more importantly, \textit{where it will not be} (see Figures \ref{FigHistoCommonDensity} and \ref{FigHistoSticky}). In particular, price stickiness is a phenomenon that can be addressed within this framework (see Example \ref{ex7} and the discussion presented after it).

Clearly, stochastic non-tâtonnement theory, as far as this paper goes, has its shortcomings. For instance, the analytical treatment of these models is not an easy one (see Example \ref{ex5}), and adopting numerical techniques (e.g., Monte Carlo methods) with well-behaved common density processes is a possible alternative (see Definition \ref{defCommonDensityProcess} and Theorem \ref{theoNested}).

However, which common density best fits specific economic scenarios (i.e., which Edgeworthian auctioneer best represents a given market) is a topic that still deserves great attention. One could try to cope with this difficulty by adopting some sort of robustness technique, focusing on a large set of ``reasonable'' densities instead, but this also brings its own challenges.  

Even with its shortcomings (every theory has its own), stochastic non-tâtonnement theory is grounded in the fundamental definitions and results of general equilibrium theory and helps to build its dynamic aspect. In this sense, it is a natural outcome of general equilibrium theory.

\section{Concluding remarks}\label{sec8}

The main contribution of this paper is to lay out a theory of stochastic non-tâtonnement processes motivated by the attraction principle, which is itself built upon the definition of the flattening diffeomorphism (and, indirectly, of the normalized Walrasian demand).

I believe that this theory can lead to a fruitful branch of research; therefore, I present six necessary steps for its enhancement. First, it is desirable to have a more general version of Theorem \ref{theoStochasticFirstWelfare}. Second, one needs to carefully analyze the relationship between stochastic non-tâtonnement processes, the Second Welfare Theorem, and accessibility of Pareto optima. Third, one needs to tackle the matter of how stochastic non-tâtonnement processes and expectations interact in the setting of a dynamic economy. 

Fourth, the theory must move beyond pure exchange economies to incorporate production and incomplete markets. Fifth, it is necessary to tackle the shortcoming described at the end of Section \ref{sec7}, related to the selection of densities that best fit specific economic scenarios. Finally, the theory must address the evolutionary process of the Edgeworthian auctioneer described after Example \ref{ex6}, in order to precisely state the conditions under which it will converge to a Walrasian equilibrium.

\appendix
\section*{Appendix}\label{appx}

All omitted proofs are given in this \hyperref[appx]{Appendix}.

\begin{proof}[Proof of Lemma~{\upshape\ref{lemmaBasicIdentities}}]
Let $p\in\mathbb{R}^{L}_{++}$. Walras' law implies $p\cdot x_{n}(p)=1$, for $p\in\mathbb{R}^{L}_{++}$, and differentiating it yields $p\mathbf{J}x_{n}(p)+x_{n}(p)=0$. The definition of $\lambda_{n}(\cdot)$ implies $\nabla u(x_{n}(p))\cdot x_{n}(p)=\lambda_{n}(p)p\cdot x_{n}(p)=\lambda_{n}(p)$. The definition of $v_{n}(\cdot)$ implies $\nabla v_{n}(p)=\nabla u(x_{n}(p))\mathbf{J}x_{n}(p)=\lambda_{n}(p)p\mathbf{J}x_{n}(p)=-\lambda_{n}(p)x_{n}(p)$. Also, $p\cdot\nabla v_{n}(p)=-\lambda_{n}(p)p\cdot x_{n}(p) =-\lambda_{n}(p)$. 

Next, let $(p,u)\in\mathbb{R}^{L+1}_{++}$. Since $u(\cdot)$ is strictly increasing and $u(0)=0$, then $x_{n}(p)=h(p,v_{n}(p))$ and $h(p,u)=x_{n}(p/e(p,u))$ \parencite[p. 58, Proposition 3.E.1]{Mas-ColellWhinstonGreen_1995}. Finally, $e(p,v_{n}(p))=p\cdot h(p,v_{n}(p))=p\cdot x_{n}(p)=1$.
\end{proof}

\begin{proof}[Proof of Proposition~{\upshape\ref{propInverseDemand}}]
Let $c\in\mathbb{R}^{L}_{++}$. Since $\nabla u(c)\in\mathbb{R}^{L}_{++}$ by assumption, we can define $p=\nabla u(c)/(c\cdot\nabla u(c))\in\mathbb{R}^{L}_{++}$. Then, $p\in\mathbb{R}^{L}_{++}$, $p\cdot c=1$, and $\nabla u(c)=(\nabla u(c)\cdot c)p$, so that $c\in B(p,1)$ satisfies the necessary and sufficient first-order condition of (\ref{UMP}) and $c=x_{n}(p)$. Therefore, we have
\begin{eqnarray}\label{eqNormalizedIsDiffeo}
    c=x_{n}\biggl(\frac{\nabla u(c)}{\nabla u(c)\cdot c}\biggr)
\end{eqnarray}
for $c\in\mathbb{R}^{L}_{++}$, and, therefore, $x_{n}(\cdot)$ is surjective. 

Next, let $c=x_{n}(p_{1})=x_{n}(p_{2})$, $p_{1},p_{2}\in\mathbb{R}^{L}_{++}$. The first order conditions imply $\nabla u(c)=\lambda_{n}(p_{1})p_{1}=\lambda_{n}(p_{2})p_{2}$, and, by Lemma \ref{lemmaBasicIdentities}, we have $\lambda_{n}(p_{1})=c\cdot\nabla u(c)=\lambda_{n}(p_{2})$, so that $p_{1}=p_{2}$. We conclude that $x_{n}(\cdot)$ is injective and, thus, invertible.
By applying $x^{-1}_{n}(\cdot)$ in (\ref{eqNormalizedIsDiffeo}), we obtain $x^{-1}_{n}(c)=\nabla u(c)/(\nabla u(c)\cdot c)$, $c\in\mathbb{R}^{L}_{++}$. Since both $x_{n}(\cdot)$ and $x^{-1}_{n}(\cdot)$ are smooth functions on $\mathbb{R}^{L}_{++}$, $x_{n}(\cdot)$ is a diffeomorphism.
\end{proof}

\begin{proof}[Proof of Proposition~{\upshape\ref{propFixedPoint}}]
Let $\alpha>0$. Since $u(\cdot)$ is continuous, strictly increasing and strictly quasi-concave, and $B(\alpha)=\{y\in\mathbb{R}^{L}_{+}\mid \Vert y\Vert\leq \alpha\}$ is convex and compact, there is a unique $c(\alpha)\in\mathbb{R}^{L}_{+}$ such that
\begin{eqnarray}\label{eqFixedPoint}
    c(\alpha)=\argmax_{y\in\mathbb{R}^{L}_{+}} u(y) \textrm{ s.t. } \Vert y\Vert \leq \alpha,
\end{eqnarray}
with $\Vert c(\alpha)\Vert=\alpha$. In particular, $u(c(\alpha))>u(0)=0$. 

Suppose $c(\alpha)\in\partial\mathbb{R}^{L}_{+}$ and let $U(\alpha)=\{y\in\mathbb{R}^{L}_{++}\mid u(y)>u(c(\alpha))\}$. Since $u(\cdot)$ is continuous and strictly quasi-concave, $U(\alpha)$ is open and strictly convex, with $c(\alpha)\in\partial U(\alpha)$. Furthermore, (\ref{eqFixedPoint}) implies $U(\alpha)\cap \{y\in\mathbb{R}^{L}\mid \Vert y\Vert\leq\alpha\}=\emptyset$.

The Separating Hyperplane Theorem implies the existence of $p\in\mathbb{R}^{L}\setminus\{0\}$, $\beta\in\mathbb{R}$, such that $p\cdot y\leq \beta\leq p\cdot z$, for all $z\in U(\alpha)$, $y\in\mathbb{R}^{L}$, $\Vert y\Vert\leq\alpha$. Since $u(\cdot)$ is strictly increasing, we have $p\in\mathbb{R}^{L}_{+}$. Also, $\Vert c(\alpha)\Vert=\alpha$ and $c(\alpha)\in\partial U(\alpha)$ imply $\beta=p\cdot c(\alpha)$. Since the ball of radius $\alpha>0$ has a unique supporting hyperplane passing through $c(\alpha)$, we conclude that $p\parallel c(\alpha)$.

Then, $\Vert c(\alpha)\Vert^{2}\leq c(\alpha)\cdot z$, for all $z\in U(\alpha)$. Since $c(\alpha)\in\partial\mathbb{R}^{L}_{+}\setminus\{0\}$, there is $1\leq k\leq L$ such that $c_{k}(\alpha)>0$. Since $u(\cdot)$ is strictly increasing, there are $\delta,\delta_{k}>0$ such that $\tilde{z}=c(\alpha)+\delta\sum_{c_{j}(\alpha)=0}e_{j}-\delta_{k}e_{k}\in U(\alpha)$. Therefore,
\begin{eqnarray*}
    \Vert c(\alpha)\Vert^{2}\leq c(\alpha)\cdot\tilde{z}=  \Vert c(\alpha)\Vert^{2}-\delta_{k}c(\alpha)\cdot e_{k}\implies 0\leq -\delta_{k}c_{k}(\alpha),
\end{eqnarray*}
absurd. We conclude that $c(\alpha)\in\mathbb{R}^{L}_{++}$. 

The first order conditions for a strictly positive solution of (\ref{eqFixedPoint}) allow us to write $\nabla u(c(\alpha))=\lambda(\alpha) c(\alpha)$, for some $\lambda(\alpha)>0$. Let $p(\alpha)=x^{-1}_{n}(c(\alpha))$. Then,
\begin{eqnarray*}
     \frac{x_{n}(p(\alpha))}{\Vert x_{n}(p(\alpha))\Vert}=\frac{c(\alpha)}{\Vert c(\alpha)\Vert}=\frac{\nabla u(c(\alpha))}{\Vert \nabla u(c(\alpha))\Vert}=\frac{x^{-1}_{n}(c(\alpha))}{\Vert x^{-1}_{n}(c(\alpha))\Vert}=\frac{p(\alpha)}{\Vert p(\alpha)\Vert}.
\end{eqnarray*}
By Lemma \ref{lemmaBasicIdentities},
\begin{eqnarray*}
    1=\biggl\Vert \frac{p(\alpha)}{\Vert p(\alpha)\Vert}\biggr\Vert^{2}=\frac{p(\alpha)\cdot c(\alpha)}{\Vert p(\alpha)\Vert\Vert c(\alpha)\Vert}=\frac{1}{\Vert p(\alpha)\Vert\Vert c(\alpha)\Vert}\implies \Vert p(\alpha)\Vert\Vert c(\alpha)\Vert=1.
\end{eqnarray*}
Therefore, $\Vert p(\alpha)\Vert=\alpha^{-1}$. If $\alpha=1$, then $p(1)=c(1)=x_{n}(p(1))$ and, therefore, $p(1)$ is a fixed point.

Next, let $y\in\mathbb{R}^{L}_{++}$, $\Vert y\Vert= \alpha$ and $\nabla u(y)=\lambda y$, $\lambda>0$. Then, (\ref{eqFixedPoint}) implies $u(c(\alpha))\geq u(y)$ and quasi-concavity of $u(\cdot)$ allows us to write $\nabla u(y)\cdot (c(\alpha)-y)\geq0$. Therefore, $y\cdot c(\alpha)\geq\Vert y\Vert^{2}$ and Cauchy--Schwarz implies $y=c(\alpha)$.

Therefore, if $p\in\mathbb{R}^{L}_{++}$ satisfies $\Vert x_{n}(p)\Vert=\alpha$ and $p/\Vert p\Vert=x_{n}(p)/\Vert x_{n}(p)\Vert$, then $y=x_{n}(p)\in\mathbb{R}^{L}_{++}$ satisfies $\Vert y\Vert=\alpha$ and $\nabla u(y)=(y\cdot\nabla u(y))p$, with $y\cdot\nabla u(y)>0$. We conclude that $y=c(\alpha)$ and, therefore, $p=x^{-1}_{n}(c(\alpha))=p(\alpha)$, thus establishing the uniqueness of $p(\alpha)$, $\alpha>0$.

Finally, suppose $\tilde{p}\in\mathbb{R}^{L}_{++}$ is a fixed point of $x_{n}(\cdot)$. Lemma \ref{lemmaBasicIdentities} implies $\Vert \tilde{p}\Vert^{2}=\tilde{p}\cdot x_{n}(\tilde{p})=1$, so that $\Vert x_{n}(\tilde{p})\Vert=1$. Since $\nabla u(x_{n}(\tilde{p}))\parallel x_{n}(\tilde{p})$, we conclude by the previous reasoning that $x_{n}(\cdot)$ has a unique fixed point given by $p(1)$.
\end{proof}

\begin{proof}[Proof of Proposition~{\upshape\ref{propDDiffeo}}]
Let $p\in\mathbb{R}^{L}_{++}$. Since $u(\cdot)$ is strictly increasing and $x_{n}(p)\in\mathbb{R}^{L}_{++}$, $v_{n}(p)=u(x_{n}(p))>u(0)=0$. The homogeneity of $e(\cdot)$ and Lemma \ref{lemmaBasicIdentities} allow us to write
\begin{eqnarray}\label{eqDIsDiffeo}
    d\biggl(\frac{p_{1}}{p_{L}},\ldots,\frac{p_{L-1}}{p_{L}},v_{n}(p)\biggr)=e(p,v_{n}(p))^{-1}p=p,
\end{eqnarray}
and, therefore, $d(\cdot)$ is surjective.

For $(q_{1},u_{1}),(q_{2},u_{2})\in\mathbb{R}^{L}_{++}$, suppose $d(q_{1},u_{1})=d(q_{2},u_{2})$. Then, $e((q_{1},1),u_{1})=e((q_{2},1),u_{2})$. This implies $q_{1}=q_{2}$, and, therefore, $u_{1}=u_{2}$ since the expenditure function is strictly increasing in its utility argument. We conclude that $d(\cdot)$ is injective and, thus, invertible. By applying $d^{-1}(\cdot)$ in (\ref{eqDIsDiffeo}), we obtain $d^{-1}(p)=(p_{1}/p_{L},\ldots,p_{L-1}/p_{L},v_{n}(p))$, $p\in\mathbb{R}^{L}_{++}$. 

Notice that $d^{-1}(\cdot)$ is smooth since $v_{n}=u\circ x_{n}$ and its Jacobian is given by\footnote{There is abuse of notation when I write $\nabla_{L-1}v_{n}(p)\in\mathbb{R}^{L-1}$ (i.e., all the first $L-1$ coordinates of the gradient) and $\nabla_{L}v_{n}(p)\in\mathbb{R}$ (i.e., only the last coordinate of the gradient).}
\begin{eqnarray*}
    \mathbf{J}d^{-1}(p)=\begin{bmatrix}
        Id/p_{L} & -p_{L}^{-2}(p_{1},\ldots,p_{L-1})^{T}\\
        \nabla_{L-1} v_{n}(p) & \nabla_{L} v_{n}(p)
    \end{bmatrix},
\end{eqnarray*}
for $p\in\mathbb{R}^{L}_{++}$. Lemma \ref{lemmaBasicIdentities} implies $\nabla v_{n}(p)=-\lambda_{n}(p)x_{n}(p)$, so that
\begin{eqnarray*}
\textrm{det\,}\mathbf{J}d^{-1}(p)= -\frac{\lambda_{n}(p)}{p_{L}^{L}} p\cdot x_{n}(p)= -\frac{\lambda_{n}(p)}{p_{L}^{L}}<0,
\end{eqnarray*}
for all $p\in\mathbb{R}^{L}_{++}$. Therefore, the Inverse Function Theorem implies that $d(\cdot)$ is smooth. Finally, since both $d(\cdot)$ and $d^{-1}(\cdot)$ are smooth functions on $\mathbb{R}^{L}_{++}$, $d(\cdot)$ is a diffeomorphism.
\end{proof}

\begin{proof}[Proof of Proposition~{\upshape\ref{propAreManifolds}}]
Let $c,y\in\mathbb{R}^{L}_{++}$. Lemma \ref{lemmaBasicIdentities} implies that $\nabla u(y)\cdot(y-c)=0$ if, and only if, $x^{-1}_{n}(y)\cdot c=x^{-1}_{n}(y)\cdot y=1$, and the offer hypersurface can be written as $\mathcal{O}(c)=\{y\in\mathbb{R}^{L}_{++}\mid x^{-1}_{n}(y)\cdot c=1\}$. Then, for $y\in\mathcal{O}(c)$, 
\begin{eqnarray*}
    y=x_{n}(x_{n}^{-1}(y))=x_{n}\biggl(\frac{x^{-1}_{n}(y)/x^{-1}_{nL}(y)}{x^{-1}_{n}(y)\cdot c/x^{-1}_{nL}(y)}\biggr)\implies y\in \biggl\{ x_{n}\biggl(\frac{(q,1)}{(q,1)\cdot c}\biggr)\mid q\in\mathbb{R}^{L-1}_{++}\biggr\}.
\end{eqnarray*}
Also, for $q\in\mathbb{R}^{L-1}_{++}$,
\begin{eqnarray*}
    x^{-1}_{n}\biggl(x_{n}\biggl(\frac{(q,1)}{(q,1)\cdot c}\biggr)\biggr)\cdot c=\frac{(q,1)\cdot c}{(q,1)\cdot c}=1\implies x_{n}\biggl(\frac{(q,1)}{(q,1)\cdot c}\biggr)\in\mathcal{O}(c).
\end{eqnarray*}
Therefore, $\mathcal{O}(c)=\{x_{n}((q,1)/(q,1)\cdot c)\mid q\in\mathbb{R}^{L-1}_{++}\}$. Then, by analyzing the map
\begin{eqnarray}\label{eqAuxMap}
    y\in\mathcal{O}(c)\mapsto \frac{\nabla_{L-1}u(y)}{\partial u(y)/\partial c_{L}}\in\mathbb{R}^{L-1}_{++}, 
\end{eqnarray}
it is straightforward to conclude that $\mathcal{O}(c)$ is globally diffeomorphic to $\mathbb{R}^{L-1}_{++}$ and, hence, a connected manifold without boundary.

Next, for $y\in\mathcal{I}(c)$,
\begin{eqnarray*}
    y=x_{n}\biggl(\frac{x^{-1}_{n}(y)/x^{-1}_{nL}(y)}{e(x^{-1}_{n}(y)/x^{-1}_{nL}(y),v_{n}(x^{-1}_{n}(y)))}\biggr)=h\biggl(\frac{x^{-1}_{n}(y)}{x^{-1}_{nL}(y)},u(c)\biggr),
\end{eqnarray*}
and, therefore, $y\in\{h((q,1),u(c))\mid q\in\mathbb{R}^{L-1}_{++}\}$. Also, for $q\in\mathbb{R}^{L-1}_{++}$, $u(h((q,1),u(c)))=u(c)$, and so $\mathcal{I}(c)=\{h((q,1),u(c))\in\mathbb{R}^{L}_{++}\mid q\in\mathbb{R}^{L-1}_{++}\}$. Then, by analyzing the same map in (\ref{eqAuxMap}), it is straightforward to conclude that $\mathcal{I}(c)$ is globally diffeomorphic to $\mathbb{R}^{L-1}_{++}$ and, hence, a connected manifold without boundary. 

Finally, for $y,c\in\mathbb{R}^{L}_{++}$, $\nabla u(c)\cdot(y-c)=0$ if, and only if, $x^{-1}_{n}(c)\cdot y=x^{-1}_{n}(c)\cdot c=1$. We conclude that the trade hyperplane can be written as $\mathcal{H}(c)=\{y\in\mathbb{R}^{L}_{++}\mid x^{-1}_{n}(c)\cdot y=1\}$, so that it is also a connected manifold without boundary of dimension $L-1$.
\end{proof}

\begin{proof}[Proof of Proposition~{\upshape\ref{propManifoldsDomains}}]
For $p\in\mathbb{R}^{L}_{++}$,
\begin{eqnarray*}
   p\in x^{-1}_{n}(\mathcal{I}(c)) \iff x_{n}(p)\in \mathcal{I}(c)\iff u(x_{n}(p))=u(c)\iff v_{n}(p)=u(c),
\end{eqnarray*}
and
\begin{eqnarray*}
    p\in x^{-1}_{n}(\mathcal{O}(c)) \iff x_{n}(p)\in \mathcal{O}(c)\iff x^{-1}_{n}(x_{n}(p))\cdot c=1\iff p\cdot c=1,
\end{eqnarray*}
so that the first two identities are valid. Also, Proposition \ref{propAreManifolds} implies
\begin{eqnarray*}
   p\in x^{-1}_{n}(\mathcal{H}(c)) \iff x_{n}(p)\in \mathcal{H}(c)\iff x_{n}^{-1}(c)\cdot x_{n}(p)=1,
\end{eqnarray*}
and the third identity is valid. For $(q,u)\in\mathbb{R}^{L}_{++}$,
\begin{eqnarray*}
    (q,u)\in f(\mathcal{I}(c))\iff h((q,1),u)\in \mathcal{I}(c)
    \iff u=u(c),
\end{eqnarray*}
and the fourth identity is valid. Also,
\begin{eqnarray*}
    (q,u)\in f(\mathcal{O}(c))\iff h((q,1),u)\in\mathcal{O}(c)
    \iff x^{-1}_{n}(h((q,1),u))\cdot c=\frac{(q,1)\cdot c}{e((q,1),u)}=1,
\end{eqnarray*}
and the fifth identity is also valid. Finally,
\begin{eqnarray*}
    (q,u)\in f(\mathcal{H}(c)) \iff h((q,1),u)\in \mathcal{H}(c)\iff x_{n}^{-1}(c)\cdot h((q,1),u)=1,
\end{eqnarray*}
and the last identity is valid.
\end{proof}

\begin{proof}[Proof of Lemma~{\upshape\ref{lemmaIndirectQuasiconvex}}]
Let $p_{1},p_{2}\in\mathbb{R}^{L}_{++}$ and $\alpha\in[0,1]$. Then, Proposition 3.D.3 from \textcite[p. 56]{Mas-ColellWhinstonGreen_1995} implies
\begin{eqnarray*}
    v_{n}(\alpha p_{1}+(1-\alpha)p_{2})
    &=&v(\alpha p_{1}+(1-\alpha)p_{2},\alpha+(1-\alpha))\\
    &\leq& \max\{v(p_{1},1),v(p_{2},1)\}\\
    &=&\max\{v_{n}(p_{1}),v_{n}(p_{2})\},
\end{eqnarray*}
and, therefore, $v_{n}(\cdot)$ is quasi-convex.
\end{proof}

\begin{proof}[Proof of Theorem~{\upshape\ref{theoIndifferenceNormDom}}]
Let $c\in\mathbb{R}^{L}_{++}$. The continuity of $v_{n}(\cdot)$ implies that $\Omega(c)$ is closed in $\mathbb{R}^{L}_{++}$. Continuity and strict monotonicity of $v_{n}(\cdot)$ and Proposition \ref{propManifoldsDomains} imply $\partial \Omega(c)=\{p\in\mathbb{R}^{L}_{++}\mid v_{n}(p)=u(c)\}=x^{-1}_{n}(\mathcal{I}(c))$, with the boundary taken in $\mathbb{R}^{L}_{++}$. Next, let $p_{1},p_{2}\in \Omega(c)$ and $\alpha \in [0,1]$. Then, by Lemma \ref{lemmaIndirectQuasiconvex}, 
\begin{eqnarray}\label{eqTheoIndifference}
    v_{n}(\alpha p_{1}+(1-\alpha)p_{2})\leq\max\{v_{n}(p_{1}),v_{n}(p_{2})\}\leq u(c)\implies \alpha p_{1}+(1-\alpha)p_{2}\in\Omega(c),
\end{eqnarray}
and $\Omega(c)\subset\mathbb{R}^{L}_{++}$ is a convex set. 

Next, notice that Proposition \ref{propManifoldsDomains} implies that
$x^{-1}_{n}(\mathcal{I}(c))\cap x^{-1}_{n}(\mathcal{O}(c))=\{p\in\mathbb{R}^{L}_{++}\mid v_{n}(p)=u(c),p\cdot c=1\}$. Let $p\in x^{-1}_{n}(\mathcal{I}(c))\cap x^{-1}_{n}(\mathcal{O}(c))$. Then, $p\cdot c=1$ implies $c\in B(p,1)$, so that $u(c)=v_{n}(p)$ and the strict quasi-concavity of $u(\cdot)$ imply $c=x_{n}(p)$. We conclude that $x^{-1}_{n}(\mathcal{I}(c))\cap x^{-1}_{n}(\mathcal{O}(c))=\{x^{-1}_{n}(c)\}$.

Let $p\in\Omega(c)$ and suppose $p\cdot c<1$. Then, $c\in B(p,1)\setminus\partial B(p,1)$ and so $u(c)<v_{n}(p)$, absurd. We conclude that $p\cdot c\geq1=x^{-1}_{n}(c)\cdot c$, for all $p\in\Omega(c)$. This conclusion (i.e., all points in $\Omega(c)$ lie on the same side of the hyperplane), $\partial\Omega(c)=x_{n}^{-1}(\mathcal{I}(c))$ with $x_{n}^{-1}(\cdot)$ smooth,
and $x^{-1}_{n}(\mathcal{I}(c))\cap x^{-1}_{n}(\mathcal{O}(c))=\{x^{-1}_{n}(c)\}$ imply that $x^{-1}_{n}(\mathcal{O}(c))$ is the supporting hyperplane of $\Omega(c)$.  

Finally, if $v_{n}(\cdot)$ is strictly quasi-convex, then (\ref{eqTheoIndifference}) becomes a strict inequality, for $p_{1}\neq p_{2}$, and, therefore, $\Omega(c)\subset\mathbb{R}^{L}_{++}$ is strictly convex.
\end{proof}

\begin{proof}[Proof of Theorem~{\upshape\ref{theoOfferFlatDom}}]
Let $c\in\mathbb{R}^{L}_{++}$. The continuity of $e(\cdot)$ implies $\Gamma(c)$ and $\Gamma(c,s)$, $s\geq u(c)$, are closed in $\mathbb{R}^{L}_{++}$. Also, continuity and monotonicity (for the last coordinate) of $e(\cdot)$ and Proposition \ref{propManifoldsDomains} imply $ \partial \Gamma (c)=\{(q,u)\in\mathbb{R}^{L}_{++}\mid (q,1)\cdot c=e((q,1),u)\}=f(\mathcal{O}(c))$, with the boundary taken in $\mathbb{R}^{L}_{++}$. 

Let $(q_{1},u_{1}),(q_{2},u_{2})\in \Gamma(c,s)$, $s\geq u(c)$, $\alpha \in [0,1]$, and $(q_{\alpha},u_{\alpha})=\alpha(q_{1},u_{1})+(1-\alpha)(q_{2},u_{2})\in\mathbb{R}^{L}_{++}$. Then, $u_{1}=u_{2}=u_{\alpha}=s$ and 
\begin{eqnarray}\label{eqAuxGamma1}
    (q_{i},1)\cdot c\leq e((q_{i},1),s),
\end{eqnarray}
for $i=1,2$. By Proposition 3.E.2 from \textcite[p. 59]{Mas-ColellWhinstonGreen_1995}, $e(\cdot,s)$ is concave. Then, (\ref{eqAuxGamma1}) implies
\begin{eqnarray*}
    e((q_{\alpha},1),s)\geq\alpha e((q_{1},1),s)+(1-\alpha)e((q_{2},1),s)\geq
    (q_{\alpha},1)\cdot c,
\end{eqnarray*}
and, therefore, $(q_{\alpha},u_{\alpha})\in \Gamma(c,s)$. We conclude that $\Gamma(c,s)$, $s\geq u(c)$, is convex. 

Next, for all $q\in\mathbb{R}^{L-1}_{++}$, the homogeneity of degree zero of the Hicksian demand and Lemma \ref{lemmaBasicIdentities} imply
\begin{eqnarray}\label{eqAuxGamma2}
    e\biggl((q,1),v_{n}\biggl(\frac{(q,1)}{(q,1)\cdot c}\biggr)\biggr)&=&(q,1)\cdot h\biggl(\frac{(q,1)}{(q,1)\cdot c},v_{n}\biggl(\frac{(q,1)}{(q,1)\cdot c}\biggr)\biggr)\nonumber\\
    &=&(q,1)\cdot x_{n}\biggl(\frac{(q,1)}{(q,1)\cdot c}\biggr)\nonumber\\
    &=&(q,1)\cdot c.
\end{eqnarray}
Since $e((q,1),\cdot)$ is strictly increasing \parencite[p. 59, Proposition 3.E.2]{Mas-ColellWhinstonGreen_1995}, (\ref{eqAuxGamma2}) implies that $(q,1)\cdot c\leq e((q,1),u)$, $u\geq v_{n}((q,1)/(q,1)\cdot c)=\gamma_{c}(q)$. Since $u(\cdot)$ is unbounded by assumption, we have
\begin{eqnarray}\label{eqAuxGamma3}
    \{u>0\mid (q,u)\in \Gamma(c)\}=[\gamma_{c}(q),+\infty),
\end{eqnarray}
for $q\in\mathbb{R}^{L-1}_{++}$.  

Then, let $s_{2}>s_{1}\geq u(c)$. If $q_{1}\in \Gamma_{L-1}(c,s_{1})$, then (\ref{eqAuxGamma3}) implies $\gamma_{c}(q_{1})\leq s_{1}<s_{2}$, so that $q_{1}\in \Gamma_{L-1}(c,s_{2})$. Therefore, $\Gamma_{L-1}(c,s_{1})\subseteq \Gamma_{L-1}(c,s_{2})$ and we conclude that $\{\Gamma_{L-1}(c,s)\}_{s\geq u(c)}$ is a non-decreasing net of closed nonempty convex sets.

Next, suppose $\gamma_{c}(\cdot)$ is convex and let $(q_{1},u_{1}),(q_{2},u_{2})\in\Gamma(c)$, $u_{1}\leq u_{2}$, $\alpha\in [0,1]$ and $(q_{\alpha},u_{\alpha})=\alpha(q_{1},u_{1})+(1-\alpha)(q_{2},u_{2})\in\mathbb{R}^{L}_{++}$. Then, (\ref{eqAuxGamma3}) implies $u_{1}\geq \gamma_{c}(q_{1})$ and $u_{2}\geq \gamma_{c}(q_{2})$, and 
\begin{eqnarray}\label{eqtheoOfferFlatDom}
    \gamma_{c}(q_{\alpha})\leq\alpha \gamma_{c}(q_{1})+(1-\alpha)\gamma_{c}(q_{2})\leq \alpha u_{1} +(1-\alpha)u_{2}=u_{\alpha}. 
\end{eqnarray}
Therefore, $(q_{\alpha},u_{\alpha})\in \Gamma(c)$ and $\Gamma(c)$ is convex. 

Next, notice that Proposition \ref{propManifoldsDomains} implies that $f(\mathcal{O}(c))\cap f(\mathcal{I}(c))=\{(q,u)\in\mathbb{R}^{L}_{++}\mid (q,1)\cdot c=e((q,1),u),u=u(c)\}$. Let $(q,u)\in f(\mathcal{O}(c))\cap f(\mathcal{I}(c))$. Then, $(q,1)\cdot c=e((q,1),u(c))$, so that $c\in B((q,1),e((q,1),u(c)))$. Also, $h((q,1),u(c))\in B((q,1),e((q,1),u(c)))$, with $h((q,1),u(c))=x_{n}((q,1)/e((q,1),u(c)))$ by Lemma \ref{lemmaBasicIdentities}. Since $v_{n}((q,1)/e((q,1),u(c)))=u(c)$, we conclude that $c=h((q,1),u(c))$ and, therefore, $(q,u)=f(c)$.

Let $(q,u)\in \Gamma(c)$ and notice that $u\geq u(c)$. This conclusion (i.e., all points in $\Gamma(c)$ lie on or above the horizontal hyperplane) and $f(\mathcal{O}(c))\cap f(\mathcal{I}(c))=\{f(c)\}$ imply that $f(\mathcal{I}(c))$ is the supporting hyperplane of $\Gamma(c)$. 

Finally, if $\gamma_{c}(\cdot)$ is strictly convex, then (\ref{eqtheoOfferFlatDom}) becomes a strict inequality, for $q_{1}\neq q_{2}$, and $\Gamma(c)\subset\mathbb{R}^{L}_{++}$ is strictly convex. 
\end{proof}

\begin{proof}[Proof of Theorem~{\upshape\ref{theoParetoOptimalSet}}]
First, notice that for any $q\in\mathbb{R}^{L-1}_{++}$, $w_{h}>0$, $1\leq h\leq H$, it is possible to define $\omega^{h}=x^{h}_{n}((q,1)/w_{h})$, $1\leq h\leq H$, so that $(q,1)\in\mathbb{R}^{L}_{++}$ becomes a no-trade equilibrium of an economy with these endowments. By the First Welfare Theorem (e.g., \textcite[p. 549]{Mas-ColellWhinstonGreen_1995}), $\omega\in \mathcal{P}$, and, therefore, $\{(x^{1}_{n}((q,1)/w_{1}),\ldots,x^{H}_{n}((q,1)/w_{H}))\in\mathbb{R}^{HL}_{++}\mid q\in \mathbb{R}^{L-1}_{++}, w_{h}>0,1\leq h\leq H\}\subseteq \mathcal{P}$.

Next, let $\omega\in\mathcal{P}$. Then, by the Second Welfare Theorem (e.g., \textcite[p. 552]{Mas-ColellWhinstonGreen_1995}), there is $p\in \mathbb{R}^{L}_{++}$, $p_{L}=1$, and $w_{h}>0$, $1\leq h\leq H$, such that $\sum^{H}_{h=1}w_{h}=\sum^{H}_{h=1} p\cdot \omega^{h}$ and $\omega^{h}=x^{h}_{n}(p/w_{h})$, $1\leq h\leq H$. Then, Lemma \ref{lemmaBasicIdentities} implies
\begin{eqnarray*}
    p\cdot\omega^{h}=p\cdot x^{h}_{n}\biggl(\frac{p}{w_{h}}\biggr)=w_{h},
\end{eqnarray*}
for $1\leq h\leq H$, and we conclude that
\begin{eqnarray}\label{eqParetoOptimalSet}
    \mathcal{P}=\{(x^{1}_{n}((q,1)/w_{1}),\ldots,x^{H}_{n}((q,1)/w_{H}))\in\mathbb{R}^{HL}_{++}\mid q\in \mathbb{R}^{L-1}_{++}, w_{h}>0,1\leq h\leq H\}.
\end{eqnarray}

Next, let $q\in\mathbb{R}^{L-1}_{++}$, $1\leq h\leq H$, and notice that for all $w_{h}>0$ (resp., $u_{h}>0$), there is a unique $u_{h}>0$ (resp., $w_{h}>0$) such that $u_{h}=v_{n}((q,1)/w_{h})$ or, equivalently, $w_{h}=e((q,1),u_{h})$. Therefore, it is possible to rewrite (\ref{eqParetoOptimalSet}) as
\begin{eqnarray*}
    \mathcal{P}=\{(h^{1}((q,1),u_{1}),\ldots,h^{H}((q,1),u_{H}))\in\mathbb{R}^{HL}_{++}\mid q\in \mathbb{R}^{L-1}_{++}, u_{h}>0,1\leq h\leq H\}.
\end{eqnarray*}

Finally, let $F:\mathbb{R}^{HL}_{++}\rightarrow\mathbb{R}^{HL}_{++}$, $F=(f^{1},\ldots,f^{H})$, which is a diffeomorphism since all $f^{i}(\cdot)$, $1\leq i\leq H$, are. Then, since $\mathcal{M}=\{((q,u_{1}),\ldots,(q,u_{H}))\in\mathbb{R}^{HL}_{++}\mid q\in \mathbb{R}^{L-1}_{++}, u_{h}>0,1\leq h\leq H\}\subset\mathbb{R}^{HL}_{++}$ is a connected manifold without boundary of dimension $H+L-1\geq3$ and $\mathcal{M}=F(\mathcal{P})$, we conclude that $\mathcal{P}\subset\mathbb{R}^{HL}_{++}$ is a connected manifold without boundary of dimension $H+L-1\geq3$.
\end{proof}

\begin{proof}[Proof of Lemma~{\upshape\ref{lemmaLinearIncreasingUtility}}]
Let $1\leq h\leq H$ and $c^{h}(\cdot)$ be given by (\ref{eqLinearPath}). If $p\parallel \nabla u^{h}(\omega^{h})$, then Proposition \ref{propInverseDemand} implies $x^{h}_{n}(p/p\cdot\omega^{h})=\omega^{h}$. Therefore, $c^{h}(t)=\omega^{h}$ and $u^{h}\circ c^{h}$ is constant.

If $p\nparallel \nabla u^{h}(\omega^{h})$, then, again by Proposition \ref{propInverseDemand}, $c^{h}(1)=x^{h}_{n}(p/p\cdot \omega^{h})\neq \omega^{h}=c^{h}(0)$. Notice that $c^{h}(t)=tc^{h}(1)+(1-t)c^{h}(0)$ and $p\cdot c^{h}(t)=p\cdot c^{h}(0)=p\cdot \omega^{h}$, $t\in[0,1]$. Therefore, $c^{h}(1)\neq c^{h}(0)$ implies $c^{h}(t)\neq c^{h}(1)$, for $t\in[0,1)$.

Next, notice that 
\begin{eqnarray*}
    \frac{d u^{h}(c^{h}(t))}{dt}&=&\nabla u^{h}(c^{h}(t))\cdot \biggl(x^{h}_{n}\biggl(\frac{p}{p\cdot \omega^{h}}\biggr)-\omega^{h}\biggr)=\frac{\nabla u^{h}(c^{h}(t))\cdot (c^{h}(1)-c^{h}(t))}{1-t},
\end{eqnarray*}
for $t\in[0,1)$, and, therefore,
\begin{eqnarray}\label{eqAuxLemmaIncreasing2}
    \frac{d u^{h}(c^{h}(t))}{dt}>0\iff \nabla u^{h}(c^{h}(t))\cdot (c^{h}(1)-c^{h}(t))>0,
\end{eqnarray}
for $t\in[0,1)$. However, $c^{h}(t)\neq c^{h}(1)$, for $t\in[0,1)$, implies $c^{h}(t)\in B(p,p\cdot c^{h}(t))\setminus\{c^{h}(1)\}=B(p,p\cdot c^{h}(1))\setminus\{c^{h}(1)\}$, so that 
\begin{eqnarray}\label{eqAuxLemmaIncreasing3}
 u^{h}(c^{h}(1))>u^{h}(c^{h}(t)),
\end{eqnarray}
for $t\in[0,1)$. Since $c^{h}(t)\in B(\nabla u^{h}(c^{h}(t)),c^{h}(t)\cdot\nabla u^{h}(c^{h}(t)))$ is optimal (i.e., $c^{h}(t)=x^{h}_{n}(\nabla u^{h}(c^{h}(t))/c^{h}(t)\cdot\nabla u^{h}(c^{h}(t)))$), (\ref{eqAuxLemmaIncreasing3}) implies $c^{h}(1)\notin B(\nabla u^{h}(c^{h}(t)),c^{h}(t)\cdot\nabla u^{h}(c^{h}(t)))$, so that 
\begin{eqnarray*}
    \nabla u^{h}(c^{h}(t))\cdot(c^{h}(1)-c^{h}(t))>0,
\end{eqnarray*}
for $t\in[0,1)$. We conclude by (\ref{eqAuxLemmaIncreasing2}) that $u^{h}\circ c^{h}$ is strictly increasing.
\end{proof}

\begin{proof}[Proof of Lemma~{\upshape\ref{lemmaLinearNet}}]
Let $\omega\in\mathbb{R}^{HL}_{++}$, $p=(q,1)\in\mathbb{R}^{L}_{++}$, $s\in[0,1]^{H}$, and $c(\omega,p,s,\cdot)$ be a linear joint trade path. To ease notation, I write $c(t)=c(\omega,p,s,t)$ and $\Delta^{h}=x^{h}_{n}(p/p\cdot\omega^{h})-\omega^{h}$, so that $c^{h}(t)=\omega^{h}+ts_{h}\Delta^{h}$, $1\leq h\leq H$.

First, $c(t)\in \mathcal{D}(c(t))$, $0\leq t\leq 1$, so that $f^{h}_{L-1}(c^{h}(t))\in f^{h}_{L-1}(\mathcal{D}^{h}(c(t)))\subseteq \mathcal{F}(c(t))$, $1\leq h\leq H$.

Next, let $0\leq t_{1}\leq t_{2}\leq 1$. Since $c(\cdot)$ is a joint trade path, $u^{h}\circ c^{h}$ is non-decreasing, $1\leq h\leq H$, and $\sum^{H}_{h=1}c^{h}(t_{1})=\sum^{H}_{h=1}c^{h}(t_{2})$. Therefore, $y\in \mathcal{D}(c(t_{2}))$ implies $u^{h}(y^{h})\geq u^{h}(c^{h}(t_{2}))\geq u^{h}(c^{h}(t_{1}))$, $1\leq h\leq H$, and $\sum^{H}_{h=1}y^{h}=\sum^{H}_{h=1}c^{h}(t_{1})$, so that $\mathcal{D}(c(t_{2}))\subseteq \mathcal{D}(c(t_{1}))$. We conclude that $\mathcal{D}^{h}(c(t_{2}))\subseteq \mathcal{D}^{h}(c(t_{1}))$, $1\leq h\leq H$, and $\mathcal{F}(c(t_{2}))\subseteq \mathcal{F}(c(t_{1}))$, so that $\{\mathcal{F}(c(t))\}_{0\leq t\leq 1}$ is a non-increasing net. Furthermore, each of its elements is arc-connected due to the argument preceding Proposition \ref{propAttractionPrinciple}, which holds for $\mathcal{F}(y)$, $y\in\mathbb{R}^{HL}_{++}$.

Finally, Definition \ref{defJointLinearTradePath} implies $c(1)\succ \omega$, so that $s\neq 0$ and $\mu=\Vert s\Vert^{-1}_{\infty}\geq 1$ is well defined. Let $\tilde{s}=\mu s$, so that $\tilde{s}\in[0,1]^{H}$ and $\Vert \tilde{s}\Vert_{\infty}=1$, and let $\tilde{y}^{h}=\omega^{h}+\tilde{s}_{h}\Delta^{h}$, $1\leq h\leq H$. Definition \ref{defJointLinearTradePath} implies $\sum^{H}_{h=1}s_{h}\Delta^{h}=0$ and, therefore,
\begin{eqnarray*}
\sum^{H}_{h=1}\tilde{y}^{h}=\sum^{H}_{h=1}\omega^{h}+\mu\sum^{H}_{h=1}s_{h}\Delta^{h}=\sum^{H}_{h=1}\omega^{h}=\sum^{H}_{h=1}c^{h}(1).
\end{eqnarray*}
Also, Lemma \ref{lemmaLinearIncreasingUtility} implies that $t\mapsto u^{h}(\omega^{h}+t\Delta^{h})$ is non-decreasing on $[0,1]$, and $\tilde{s}_{h}\geq s_{h}$, so that $u^{h}(\tilde{y}^{h})\geq u^{h}(c^{h}(1))$, $1\leq h\leq H$. We conclude that $\tilde{y}\in \mathcal{D}(c(1))$.

Let $1\leq \tilde{h}\leq H$ be such that $\tilde{s}^{\tilde{h}}=1$. Then, $\tilde{y}^{\tilde{h}}=x^{\tilde{h}}_{n}(p/p\cdot\omega^{\tilde{h}})$ and Proposition \ref{propInverseDemand} implies
\begin{eqnarray*}
    \frac{\nabla u^{\tilde{h}}(\tilde{y}^{\tilde{h}})}{\tilde{y}^{\tilde{h}}\cdot\nabla u^{\tilde{h}}(\tilde{y}^{\tilde{h}})}=\frac{p}{p\cdot \omega^{\tilde{h}}},
\end{eqnarray*}
so that $\nabla u^{\tilde{h}}(\tilde{y}^{\tilde{h}})$ is a strictly positive multiple of $p=(q,1)$ and $f^{\tilde{h}}_{L-1}(\tilde{y}^{\tilde{h}})=q$. Therefore,
\begin{eqnarray*}
    q=f^{\tilde{h}}_{L-1}(\tilde{y}^{\tilde{h}})\in f^{\tilde{h}}_{L-1}(\mathcal{D}^{\tilde{h}}(c(1)))\subseteq \mathcal{F}(c(1)).
\end{eqnarray*}
\end{proof}

\begin{proof}[Proof of Proposition~{\upshape\ref{propAttractionOfMarginalRates}}]
To ease notation, let $\tilde{c}(t)=c^{h}(\omega^{h},p,s_{h},t)$ and $\Delta^{h}=x^{h}_{n}(p/p\cdot\omega^{h})-\omega^{h}$, $t\in[0,1]$. Lemma \ref{lemmaBasicIdentities} implies $p\cdot \Delta^{h}=0$, so that $p\cdot\tilde{c}(t)=p\cdot \omega^{h}$ and
\begin{eqnarray*}
    x^{h}_{n}\biggl(\frac{p}{p\cdot \tilde{c}(t)}\biggr)-\tilde{c}(t)=(1-ts_{h})\Delta^{h},
\end{eqnarray*}
for $t\in[0,1]$. Therefore, $\tilde{c}^{\prime}(t)=s_{h}\Delta^{h}$ and direct differentiation of $\delta_{ij}(\cdot)$ leads us to
\begin{eqnarray*}
    \delta^{\prime}_{ij}(t)=\frac{2s_{h}}{1-ts_{h}}\biggl(\frac{\nabla_{i}u^{h}(\tilde{c}(t))}{\nabla_{j}u^{h}(\tilde{c}(t))}-\frac{p_{i}}{p_{j}}\biggr)\frac{\nabla u^{h}(\tilde{c}(t))(e_{j}^{T}e_{i}-e_{i}^{T}e_{j})\mathbf{H}u^{h}(\tilde{c}(t))}{\nabla_{j}u^{h}(\tilde{c}(t))^{2}}\biggl(x^{h}_{n}\biggl(\frac{p}{p\cdot \tilde{c}(t)}\biggr)-\tilde{c}(t)\biggr)^{T},
\end{eqnarray*}
for $t\in[0,1]$ with $ts_{h}<1$. Since $u^{h}(\cdot)$ is attractive and $2s_{h}(1-ts_{h})^{-1}\nabla_{j}u^{h}(\tilde{c}(t))^{-2}\geq0$, we have $\delta^{\prime}_{ij}(t)\leq 0$ whenever $ts_{h}<1$. If $s_{h}=1$, then $\delta^{\prime}_{ij}(1)\leq0$ follows by continuity, and we conclude that $\delta_{ij}(\cdot)$ is non-increasing on $[0,1]$.

Finally, if $s_{h}=1$, then $\tilde{c}(1)=x^{h}_{n}(p/p\cdot \tilde{c}(1))$ and Proposition \ref{propInverseDemand} implies that $\nabla u^{h}(\tilde{c}(1))$ is a strictly positive multiple of $p$, so that $\delta_{ij}(1)=0$.
\end{proof}

\begin{proof}[Proof of Proposition~{\upshape\ref{propIsAttractive}}]
    Let $\eta=1/(1-\sigma)>0$ (the Cobb--Douglas case is given by $\eta=1$), $\theta(c)=(\alpha_{1}c_{1}^{\sigma-1},\ldots,\alpha_{L}c_{L}^{\sigma-1})\in\mathbb{R}^{L}_{++}$ and $\gamma(p)=(\alpha_{1}^{\eta}/p_{1}^{\eta},\ldots,\alpha_{L}^{\eta}/p_{L}^{\eta})$, for $c,p\in\mathbb{R}^{L}_{++}$. Since $\sigma-1=-\eta^{-1}$, we have $\nabla u(c)=u(c)^{\frac{1}{\eta}}\theta(c)$, $x_{n}(p/p\cdot c)=\gamma(p)p\cdot c/\sum^{L}_{m=1}\alpha_{m}^{\eta}p_{m}^{1-\eta}$ and
\begin{eqnarray}\label{eqAuxIsAttractive}
    \mathbf{H}u(c)=\frac{u(c)^{\frac{2-\eta}{\eta}}}{\eta}\theta(c)^{T}\theta(c)-\frac{u(c)^{\frac{1}{\eta}}}{\eta}\textrm{Diag}(\alpha_{1}c^{\sigma-2}_{1},\ldots,\alpha_{L}c^{\sigma-2}_{L}).
\end{eqnarray}
To ease notation, let $\Delta_{ij}(c,p):\mathbb{R}^{2L}_{++}\rightarrow\mathbb{R}$, $1\leq i,j\leq L$, be given by
\begin{eqnarray*}
    \Delta_{ij}(c,p)=\biggl(\frac{\nabla_{i} u(c)}{\nabla_{j}u(c)}-\frac{p_{i}}{p_{j}}\biggr)\nabla u(c)(e_{j}^{T}e_{i}-e_{i}^{T}e_{j})\mathbf{H}u(c)\biggl(x_{n}\biggl(\frac{p}{p\cdot c}\biggr)-c\biggr)^{T},
\end{eqnarray*}
for $c,p\in\mathbb{R}^{L}_{++}$. Notice that $\nabla u(c)(e_{j}^{T}e_{i}-e_{i}^{T}e_{j})\theta(c)^{T}=0$, so that the first term of (\ref{eqAuxIsAttractive}) does not contribute to $\Delta_{ij}(\cdot)$ and
\begin{eqnarray*}
    \Delta_{ij}(c,p) =-\frac{\alpha_{i}\alpha_{j}c_{i}^{\sigma-1}c_{j}^{\sigma-1}u(c)^{\frac{2}{\eta}}p\cdot c}{\eta\sum^{L}_{m=1}\alpha_{m}^{\eta}p_{m}^{1-\eta}}\biggl(\frac{\alpha_{i}c_{i}^{\sigma-1}}{\alpha_{j}c_{j}^{\sigma-1}}-\frac{p_{i}}{p_{j}}\biggr)\biggl(\frac{\alpha_{i}^{\eta}}{c_{i}p_{i}^{\eta}}-\frac{\alpha_{j}^{\eta}}{c_{j}p_{j}^{\eta}}\biggr),
\end{eqnarray*}
for $c,p\in\mathbb{R}^{L}_{++}$. Furthermore, since $\eta>0$ and $\eta(\sigma-1)=-1$,
\begin{eqnarray*}
    \frac{\alpha_{i}c_{i}^{\sigma-1}}{\alpha_{j}c_{j}^{\sigma-1}}\geq \frac{p_{i}}{p_{j}}\iff \biggl(\frac{p_{j}\alpha_{i}c_{i}^{\sigma-1}}{p_{i}\alpha_{j}c_{j}^{\sigma-1}}\biggr)^\eta\geq1 \iff \frac{\alpha_{i}^{\eta}}{c_{i}p_{i}^{\eta}}\geq\frac{\alpha_{j}^{\eta}}{c_{j}p_{j}^{\eta}},
    \end{eqnarray*}
for $c,p\in\mathbb{R}^{L}_{++}$. Therefore, $\Delta_{ij}(c,p)\leq0$, $c,p\in\mathbb{R}^{L}_{++}$, and $u(\cdot)$ is attractive.
\end{proof}

\begin{proof}[Proof of Lemma~{\upshape\ref{lemmaTradeCompPrices}}]
Let $\omega\in\mathbb{R}^{HL}_{++}$. If $\mathcal{T}(\omega)\neq\emptyset$, then there are $p\in\mathcal{T}(\omega)$, $s\in\mathcal{S}(\omega,p)$, and a linear joint trade path $c(\omega,p,s,\cdot)$. Definition \ref{defJointLinearTradePath} implies $c(\omega,p,s,1)\succ\omega$, so that $\omega\in\mathbb{R}^{HL}_{++}$ is Pareto dominated and we conclude that $\omega\notin\mathcal{P}$.

Next, suppose $\mathcal{T}(\omega)=\emptyset$. Let $\tilde{p}=(\tilde{q},1)\in\mathbb{R}^{L}_{++}$ be an equilibrium price for the allocation $\omega\in\mathbb{R}^{HL}_{++}$ \parencite{ArrowDebreu_1954,McKenzie_1954} and $\tilde{s}=(1,\ldots,1)\in[0,1]^{H}$. Market clearing implies
\begin{eqnarray*}
    \sum^{H}_{h=1}\tilde{s}_{h}\biggl(x^{h}_{n}\biggl(\frac{\tilde{p}}{\tilde{p}\cdot\omega^{h}}\biggr)-\omega^{h}\biggr)=\sum^{H}_{h=1}x^{h}_{n}\biggl(\frac{\tilde{p}}{\tilde{p}\cdot\omega^{h}}\biggr)-\sum^{H}_{h=1}\omega^{h}=0,
\end{eqnarray*}
so that $\sum^{H}_{h=1}c^{h}(\omega^{h},\tilde{p},\tilde{s}_{h},t)=\sum^{H}_{h=1}\omega^{h}$, $0\leq t\leq1$. Since $\tilde{p}_{L}=1$ and $\mathcal{T}(\omega)=\emptyset$, we have $\mathcal{S}(\omega,\tilde{p})=\emptyset$ and, as $\tilde{s}\in[0,1]^{H}\setminus\{0\}$, $c(\omega,\tilde{p},\tilde{s},\cdot)$ is not a linear joint trade path. Therefore, the only requirement of Definition \ref{defJointLinearTradePath} that can fail is $c(\omega,\tilde{p},\tilde{s},1)\succ\omega$. However, Lemma \ref{lemmaLinearIncreasingUtility} implies $u^{h}(c^{h}(\omega^{h},\tilde{p},\tilde{s}_{h},1))\geq u^{h}(\omega^{h})$, $1\leq h\leq H$, with strict inequality whenever $\tilde{p}\nparallel \nabla u^{h}(\omega^{h})$. Then, $\tilde{p}\parallel \nabla u^{h}(\omega^{h})$, $1\leq h\leq H$, and Proposition \ref{propInverseDemand} implies $x^{h}_{n}(\tilde{p}/\tilde{p}\cdot\omega^{h})=\omega^{h}$, $1\leq h\leq H$. Due to Theorem \ref{theoParetoOptimalSet}, with $w_{h}=\tilde{p}\cdot\omega^{h}$, $1\leq h\leq H$, we conclude that $\omega\in\mathcal{P}$.
\end{proof}

\begin{proof}[Proof of Proposition~{\upshape\ref{propTradeUnfold}}]
Let $\omega\in\mathbb{R}^{HL}_{++}$ and $q\in\mathcal{T}_{L-1}(\omega)$. Then, for $p=(q,1)\in\mathcal{T}(\omega)$, there is $s\in\mathcal{S}(\omega,p)$ such that $c(\omega,p,s,\cdot)$ is a linear joint trade path and, due to the paragraph preceding Lemma \ref{lemmaTradeCompPrices}, we can assume without loss of generality that $\Vert s\Vert_{\infty}=1$. Therefore, there is $1\leq \tilde{h}\leq H$ with $s_{\tilde{h}}=1$, so that $c^{\tilde{h}}(\omega^{\tilde{h}},p,s_{\tilde{h}},1)=x^{\tilde{h}}_{n}(p/p\cdot\omega^{\tilde{h}})$ and Proposition \ref{propInverseDemand} implies that $\nabla u^{\tilde{h}}(c^{\tilde{h}}(\omega^{\tilde{h}},p,s_{\tilde{h}},1))$ is a strictly positive multiple of $p\in\mathbb{R}^{L}_{++}$. Hence,
\begin{eqnarray*}
    q=f^{\tilde{h}}_{L-1}(c^{\tilde{h}}(\omega^{\tilde{h}},p,s_{\tilde{h}},1))\in f^{\tilde{h}}_{L-1}(\mathcal{D}^{\tilde{h}}(\omega))\subseteq\bigcup^{H}_{h=1}f^{h}_{L-1}(\mathcal{D}^{h}(\omega))=\mathcal{F}(\omega),
\end{eqnarray*}
where the first inclusion is due to $c(\omega,p,s,1)\in\mathcal{D}(\omega)$, since $c(\omega,p,s,1)\succ \omega$. We conclude that $\mathcal{T}_{L-1}(\omega)\subseteq \mathcal{F}(\omega)$.

Next, let $\omega\in\mathbb{R}^{HL}_{++}$, $p=(q,1)\in\mathcal{T}(\omega)$, $s\in\mathcal{S}(\omega,p)$, and $\Delta^{h}=x^{h}_{n}(p/p\cdot\omega^{h})-\omega^{h}$, $1\leq h\leq H$, with $c(\omega,p,s,\cdot)$ the corresponding linear joint trade path. Then, (\ref{eqLinearPathSpeed}) allows us to write
\begin{eqnarray}\label{eqAuxTradeUnfold}
    c^{h}(\omega^{h},p,s_{h},r)&=&c^{h}(\omega^{h},p,s_{h},t)+\frac{(r-t)s_{h}}{1-ts_{h}}\biggl(x^{h}_{n}\biggl(\frac{p}{p\cdot \omega^{h}}\biggr)-c^{h}(\omega^{h},p,s_{h},t)\biggr)
\end{eqnarray}
for $0\leq t\leq r\leq 1$, $ts_{h}<1$, $1\leq h\leq H$. In particular, notice that $t<1$ is necessary since, if $\Vert s\Vert_{\infty}=1$, then for some household the denominator in (\ref{eqAuxTradeUnfold}) would not be well defined.

Let $0\leq t<1$, $\tilde{\omega}=c(\omega,p,s,t)$ and
\begin{eqnarray*}
    \tilde{s}_{h}=\frac{(1-t)s_{h}}{1-ts_{h}},\qquad 1\leq h\leq H.
\end{eqnarray*}
Since $s_{h}\leq1$ implies $(1-t)s_{h}\leq1-ts_{h}$, and $s\neq0$, we have $\tilde{s}\in[0,1]^{H}\setminus\{0\}$. Furthermore, Lemma \ref{lemmaBasicIdentities} implies $p\cdot \Delta^{h}=0$, so that $p\cdot\tilde{\omega}^{h}=p\cdot\omega^{h}$, $x^{h}_{n}(p/p\cdot\tilde{\omega}^{h})-\tilde{\omega}^{h}=(1-ts_{h})\Delta^{h}$ and
\begin{eqnarray*}
    \sum^{H}_{h=1}c^{h}(\tilde{\omega}^{h},p,\tilde{s}_{h},\tau)=\sum^{H}_{h=1}\tilde{\omega}^{h}+\tau(1-t)\sum^{H}_{h=1}s_{h}\Delta^{h}=\sum^{H}_{h=1}\tilde{\omega}^{h},
\end{eqnarray*}
for $0\leq\tau\leq1$, where the last equality is due to Definition \ref{defJointLinearTradePath}. Also, (\ref{eqAuxTradeUnfold}) with $r=1$ gives $c^{h}(\tilde{\omega}^{h},p,\tilde{s}_{h},1)=c^{h}(\omega^{h},p,s_{h},1)$, $1\leq h\leq H$. Finally, Lemma \ref{lemmaLinearIncreasingUtility} implies that $u^{h}\circ c^{h}(\omega^{h},p,s_{h},\cdot)$ is non-decreasing on $[0,1]$, $1\leq h\leq H$, and, since $c(\omega,p,s,1)\succ\omega$, strictly increasing for some $1\leq \tilde{h}\leq H$. As $t<1$, we conclude that $c(\tilde{\omega},p,\tilde{s},1)=c(\omega,p,s,1)\succ\tilde{\omega}$, so that $c(\tilde{\omega},p,\tilde{s},\cdot)$ is a linear joint trade path, $\tilde{s}\in\mathcal{S}(\tilde{\omega},p)$ and $q\in\mathcal{T}_{L-1}(c(\omega,p,s,t))$, $0\leq t<1$.
\end{proof}

\begin{proof}[Proof of Theorem~{\upshape\ref{theoStochasticFirstWelfare}}]
Given the probability space $(X,\mathcal{B},\mathbb{P})$, I denote an event by $x\in X$. Let $X_{1}=\{x\in X\mid \omega_{0}(x)=\omega, (\ref{eqStochasticDef1}),(\ref{eqStochasticDef2})\textrm{ hold}, \textrm{ and }\Vert s_{t}\Vert_{\infty}=1, t\geq1\}$, so that Definition \ref{defStochasticNTProcess} and the fact that $\mathbb{P}(\Vert s_{t}\Vert_{\infty}=1)=1$, $t\geq1$, imply $\mathbb{P}(X_{1})=1$. 

Without loss of generality, let us assume
\begin{eqnarray}\label{eqStochasticWelAux1}
    f^{1}_{1}(\omega^{1})=\frac{\nabla_{1} u^{1}(\omega^{1})}{\nabla_{2} u^{1}(\omega^{1})}\leq\frac{\nabla_{1} u^{2}(\omega^{2})}{\nabla_{2} u^{2}(\omega^{2})}=f^{2}_{1}(\omega^{2}).
\end{eqnarray}

If (\ref{eqStochasticWelAux1}) holds with an equality, then $\mathcal{T}(\omega)=\emptyset$ and (\ref{eqStochasticDef1}) imply $\omega_{t}(x)=\omega$, $t\geq0$, $x\in X_{1}$. Then, Lemma \ref{lemmaTradeCompPrices} allows us to write $\mathbb{P}(\lim_{t\rightarrow\infty}\omega_{t}\in\mathcal{P})\geq \mathbb{P}(X_{1})=1$, and we conclude that $\mathbb{P}(\lim_{t\rightarrow\infty}\omega_{t}\in\mathcal{P})=1$.

If (\ref{eqStochasticWelAux1}) is strict, then $\mathcal{T}_{1}(\omega)=\mathcal{T}_{1}(\omega_{0}(x))=(f^{1}_{1}(\omega^{1}_{0}(x)),f^{2}_{1}(\omega^{2}_{0}(x)))\neq\emptyset$, $x\in X_{1}$. Then, due to (\ref{eqStochasticDef2}), we have $p_{1}(x)=(q_{1}(x),1)\in\mathcal{T}(\omega_{0}(x))$ and $s_{1}(x)\in\mathcal{S}(\omega_{0}(x),p_{1}(x))$, with $\Vert s_{1}(x)\Vert_{\infty}=1$, $x\in X_{1}$. Since we have a maximum relative trade speed, the utilities are attractive and $\omega_{1}(x)=c(\omega,p_{1}(x),s_{1}(x),1)$, $x\in X_{1}$, Proposition \ref{propAttractionOfMarginalRates} implies that there are only three possible outcomes for this initial step: $f^{1}_{1}(\omega^{1}_{1}(x))=f^{2}_{1}(\omega^{2}_{1}(x))=q_{1}(x)$; $f^{1}_{1}(\omega^{1})\leq f^{1}_{1}(\omega^{1}_{1}(x))<f^{2}_{1}(\omega^{2}_{1}(x))=q_{1}(x)$; or $f^{1}_{1}(\omega^{1}_{1}(x))=q_{1}(x)<f^{2}_{1}(\omega^{2}_{1}(x))\leq f^{2}_{1}(\omega^{2})$.

The first case corresponds to the process reaching the manifold of Pareto optimal allocations in its first step (i.e., $\omega_{1}(x)\in\mathcal{P}$). The other two cases correspond to steps that still allow for further trade, since either $\mathcal{T}_{1}(\omega_{1}(x))=(f^{1}_{1}(\omega^{1}_{1}(x)),q_{1}(x))\neq\emptyset$ or $\mathcal{T}_{1}(\omega_{1}(x))=(q_{1}(x),f^{2}_{1}(\omega^{2}_{1}(x)))\neq\emptyset$.

By induction, for $x\in X_{1}$: (i) $\{f^{1}_{1}(\omega^{1}_{t}(x))\}_{t\geq0}$ is non-decreasing and $\{f^{2}_{1}(\omega^{2}_{t}(x))\}_{t\geq0}$ is non-increasing, with $f^{1}_{1}(\omega^{1}_{t}(x))\leq f^{2}_{1}(\omega^{2}_{t}(x))$, $t\geq0$; and (ii) $f^{1}_{1}(\omega^{1}_{t}(x))=q_{t}(x)$ or $f^{2}_{1}(\omega^{2}_{t}(x))=q_{t}(x)$, $t\geq1$.

Then, (i), and the fact that utility is increasing (due to Lemma \ref{lemmaLinearIncreasingUtility}) and bounded along every linear joint trade path imply, since $f^{1}(\cdot)$ and $f^{2}(\cdot)$ are diffeomorphisms, that $\lim_{t\rightarrow\infty}\omega_{t}(x)$ exists for $x\in X_{1}$. Furthermore, for $x\in X_{1}$, $\lim_{t\rightarrow\infty}\omega_{t}(x)\notin\mathcal{P}$ if, and only if, there is $(a,b)\in\mathbb{Q}^{2}_{++}$, $a<b$, such that $(a,b)\subseteq(f^{1}_{1}(\omega^{1}_{t}(x)),f^{2}_{1}(\omega^{2}_{t}(x)))$, $t\geq0$.

This last remark leads us to define $\mathcal{A}(\omega,a,b)=\{y\in\mathcal{D}(\omega)\mid (a,b)\subseteq \mathcal{T}_{1}(y)\}\subseteq\mathbb{R}^{4}_{++}$, for $(a,b)\in\mathbb{Q}^{2}_{++}$, $a<b$, so that
\begin{eqnarray}\label{eqStochasticWelAux2}
    \{x\in X_{1}\mid \lim_{t\rightarrow\infty}\omega_{t}(x)\notin\mathcal{P}\}&=&\bigcup_{a<b\in\mathbb{Q}_{++}}\{x\in X_{1}\mid \lim_{t\rightarrow\infty}\omega_{t}(x)\in\mathcal{A}(\omega,a,b)\}.
\end{eqnarray}

Next, notice that for $x\in X_{1}$, $\mathcal{T}_{1}(\omega_{t}(x))=(f^{1}_{1}(\omega^{1}_{t}(x)),f^{2}_{1}(\omega^{2}_{t}(x)))$, $t\geq0$, and (i) implies that $\mathcal{T}_{1}(\omega_{t+1}(x))\subseteq\mathcal{T}_{1}(\omega_{t}(x))$ (i.e., the sets of normalized trade-compatible prices form a non-increasing sequence). Then, for $(a,b)\in\mathbb{Q}^{2}_{++}$, $a<b$,
\begin{eqnarray*}
    \{x\in X_{1}\mid \lim_{t\rightarrow\infty}\omega_{t}(x)\in\mathcal{A}(\omega,a,b)\}=\bigcap_{t\geq0}\{x\in X_{1}\mid \omega_{t}(x)\in \mathcal{A}(\omega,a,b)\},
\end{eqnarray*}
so that
\begin{eqnarray}\label{eqStochasticWelAux3}
    \mathbb{P}(\lim_{t\rightarrow\infty}\omega_{t}\in\mathcal{A}(\omega,a,b))=\lim_{t\rightarrow\infty}\mathbb{P}(\omega_{t}\in \mathcal{A}(\omega,a,b)).
\end{eqnarray}

Notice that if there is $T\geq0$ such that $\mathbb{P}(\omega_{T}\in \mathcal{A}(\omega,a,b))=0$, then $\lim_{t\rightarrow\infty}\mathbb{P}(\omega_{t}\in \mathcal{A}(\omega,a,b))=0$. If not, then the following conditional probabilities are well defined
\begin{eqnarray}\label{eqStochasticWelAux4}
    \mathbb{P}(\omega_{t}\in \mathcal{A}(\omega,a,b))=\mathbb{P}(\omega_{t}\in \mathcal{A}(\omega,a,b)\mid \omega_{t-1}\in \mathcal{A}(\omega,a,b))\mathbb{P}(\omega_{t-1}\in \mathcal{A}(\omega,a,b))
\end{eqnarray}
for $t\geq1$. Then, (ii) above allows us to write
\begin{eqnarray}\label{eqStochasticWelAux5}
    \mathbb{P}(\omega_{t}\in \mathcal{A}(\omega,a,b)\mid \omega_{t-1}\in \mathcal{A}(\omega,a,b))&=&1-\mathbb{P}(\omega_{t}\notin \mathcal{A}(\omega,a,b)\mid \omega_{t-1}\in \mathcal{A}(\omega,a,b))\nonumber\\
    &\leq& 1-\mathbb{P}(q_{t}\in (a,b) \mid \omega_{t-1}\in \mathcal{A}(\omega,a,b)).
\end{eqnarray}
Since, by assumption, $\mathcal{D}(\omega)\subseteq\mathbb{R}^{4}_{++}$ is compact, there is $\varepsilon\in(0,1)$ such that
\begin{eqnarray}\label{eqStochasticWelAux6}
    \mathcal{T}_{1}(y)\subseteq(\varepsilon,\varepsilon^{-1}),
\end{eqnarray}
for $y\in\mathcal{D}(\omega)$. Then, for all $y\in\mathcal{A}(\omega,a,b)\subseteq\mathcal{D}(\omega)$,
\begin{eqnarray*}
    \mathbb{P}(q_{t}\in (a,b) \mid \omega_{t-1}=y)=\frac{\int_{\mathcal{T}_{1}(y)\cap (a,b)} \psi_{q}d\mu}{\int_{\mathcal{T}_{1}(y)} \psi_{q}d\mu}=\frac{\int_{(a,b)} \psi_{q}d\mu}{\int_{\mathcal{T}_{1}(y)} \psi_{q}d\mu}\geq\frac{\int_{(a,b)} \psi_{q}d\mu}{\int_{(\varepsilon,\varepsilon^{-1})} \psi_{q}d\mu}=\alpha_{a,b}>0
\end{eqnarray*}
where the first equality is due to (\ref{eqTheoStochastic}), the second to $y\in\mathcal{A}(\omega,a,b)$ (i.e., $(a,b)\subseteq\mathcal{T}_{1}(y)$), the third (inequality) to (\ref{eqStochasticWelAux6}), and the last inequality to $\psi_{q}\in L^{1}_{++}(\mathbb{R}_{++})$. This implies that $\mathbb{P}(q_{t}\in (a,b) \mid \omega_{t-1}\in\mathcal{A}(\omega,a,b))\geq \alpha_{a,b}>0$, so that (\ref{eqStochasticWelAux4}) and (\ref{eqStochasticWelAux5}) allow us to write
\begin{eqnarray*}
    \mathbb{P}(\omega_{t}\in \mathcal{A}(\omega,a,b))\leq (1-\alpha_{a,b})\mathbb{P}(\omega_{t-1}\in \mathcal{A}(\omega,a,b)),
\end{eqnarray*}
for $t\geq1$, and therefore $\lim_{t\rightarrow\infty}\mathbb{P}(\omega_{t}\in \mathcal{A}(\omega,a,b))=0$. Then, due to (\ref{eqStochasticWelAux3}), we once again have $\mathbb{P}(\lim_{t\rightarrow\infty}\omega_{t}\in \mathcal{A}(\omega,a,b))=0$. 

Since $\mathbb{Q}^{2}_{++}$ is a countable set, (\ref{eqStochasticWelAux2}) implies that $\mathbb{P}(\lim_{t\rightarrow\infty}\omega_{t}\notin\mathcal{P})=0$, so that $\mathbb{P}(\lim_{t\rightarrow\infty}\omega_{t}\in\mathcal{P})=1$. 
\end{proof}

\begin{proof}[Proof of Lemma~{\upshape\ref{lemmaLargerDomainQ}}]
Notice that $\mathcal{T}_{L-1}\subseteq \mathcal{V}$ implies $\int_{\mathcal{V}}\psi_{q}d\mu\geq \int_{\mathcal{T}_{L-1}}\psi_{q}d\mu>0$ and
\begin{eqnarray*}
        \mathbb{P}(q\in \mathcal{T}_{L-1})=\frac{\int_{\mathcal{V}\cap \mathcal{T}_{L-1}}\psi_{q}d\mu}{\int_{\mathcal{V}}\psi_{q}d\mu}=\frac{\int_{ \mathcal{T}_{L-1}}\psi_{q}d\mu}{\int_{\mathcal{V}}\psi_{q}d\mu}>0,
\end{eqnarray*}
where the strict inequality is due to $\psi_{q}\in L^{1}_{++}(\mathbb{R}^{L-1}_{++})$ and $\mu(\mathcal{T}_{L-1})>0$.
Furthermore,
\begin{eqnarray*}
    \mathbb{P}(q\in A \mid q\in \mathcal{T}_{L-1})=\frac{\mathbb{P}(q\in \mathcal{T}_{L-1}\cap A)}{\mathbb{P}(q\in\mathcal{T}_{L-1})}=\frac{\int_{\mathcal{V}\cap \mathcal{T}_{L-1}\cap A}\psi_{q}d\mu}{\int_{\mathcal{V}\cap\mathcal{T}_{L-1}}\psi_{q}d\mu}=\frac{\int_{\mathcal{T}_{L-1}\cap A}\psi_{q}d\mu}{\int_{\mathcal{T}_{L-1}}\psi_{q}d\mu},
\end{eqnarray*}
for $A\subseteq\mathbb{R}^{L-1}_{++}$ measurable set.
\end{proof}

\begin{proof}[Proof of Proposition~{\upshape\ref{propIsSharp}}]
Let $\eta=1/(1-\sigma)>0$ (the Cobb--Douglas case is given by $\eta=1$) and $c,p\in\mathbb{R}^{L}_{++}$. Given $1\leq i\leq L$, suppose
\begin{eqnarray}\label{eqIsSharpAux}
    \frac{p_{i}}{p_{j}}>\frac{\nabla_{i}u(c)}{\nabla_{j}u(c)}=\frac{\alpha_{i}c_{i}^{\sigma-1}}{\alpha_{j}c_{j}^{\sigma-1}},
\end{eqnarray}
for $1\leq j\leq L$, $j\neq i$. Since $\eta>0$ and $\eta(1-\sigma)=1$, (\ref{eqIsSharpAux}) is equivalent to
\begin{eqnarray*}
    \biggl(\frac{\alpha_{j}p_{i}}{\alpha_{i}p_{j}}\biggr)^{\eta}>\frac{c_{j}}{c_{i}},
\end{eqnarray*}
for $1\leq j\leq L$, $j\neq i$, while both sides equal $1$ for $j=i$. Therefore,
\begin{eqnarray*}
    \sum^{L}_{j=1}\biggl(\frac{\alpha_{j}p_{i}}{\alpha_{i}p_{j}}\biggr)^{\eta}\frac{p_{j}}{p_{i}}>\sum^{L}_{j=1}\frac{c_{j}p_{j}}{c_{i}p_{i}}>0,
\end{eqnarray*}
and
\begin{eqnarray*}
    x_{ni}\biggl(\frac{p}{p\cdot c}\biggr)-c_{i}
    =\frac{p\cdot c/p_{i}}{\sum^{L}_{j=1}(\alpha_{j}p_{i}/\alpha_{i}p_{j})^{\eta}p_{j}/p_{i}}-c_{i}
    <\frac{p\cdot c/p_{i}}{\sum^{L}_{j=1}c_{j}p_{j}/c_{i}p_{i}}-c_{i}
    =0.
\end{eqnarray*}
The case with reverse inequality in (\ref{eqIsSharpAux}) is analogous, and we conclude that $u(\cdot)$ is sharp.
\end{proof}

\begin{proof}[Proof of Lemma~{\upshape\ref{lemmaUtilitiesPreferencesAttractive}}]
Let $u(\cdot)$ be sharp. Since Definition \ref{defSharpUtility} involves only marginal rates of substitution and the normalized Walrasian demand, we conclude that $g\circ u$ is sharp.

Next, let $u(\cdot)$ be attractive. Notice that
\begin{eqnarray*}
\mathbf{H}(g\circ u)(c)=g^{\prime}(u(c))\mathbf{H}u(c)+g^{\prime\prime}(u(c))\nabla u(c)^{T}\nabla u(c),
\end{eqnarray*}
for $c\in\mathbb{R}^{L}_{++}$. Since $\nabla u(c)(e_{j}^{T}e_{i}-e_{i}^{T}e_{j})\nabla u(c)^{T}=0$, $1\leq i,j\leq L$, we have
\begin{eqnarray*}
    \nabla (g\circ u)(c)(e_{j}^{T}e_{i}-e_{i}^{T}e_{j})\mathbf{H}(g\circ u)(c)=g^{\prime}(u(c))^{2}\nabla u(c)(e_{j}^{T}e_{i}-e_{i}^{T}e_{j})\mathbf{H}u(c),
\end{eqnarray*}
for $c\in\mathbb{R}^{L}_{++}$, $1\leq i,j\leq L$. Therefore, the left-hand side of the inequality in Definition \ref{defAttractiveUtility} is multiplied by $g^{\prime}(u(c))^{2}>0$ after the composition, and we conclude that $g\circ u$ is attractive.
\end{proof}

\begin{proof}[Proof of Lemma~{\upshape\ref{lemmaTradePricesInBox}}]
Suppose $q\in\mathcal{T}_{L-1}(\omega)\setminus\Pi(\omega)$, $p=(q,1)$. Then, for some $1\leq i\leq L$, $p_{i}>\max_{j\neq i}p_{j}M_{ij}(\omega)$ or
$p_{i}<\min_{j\neq i}p_{j} m_{ij}(\omega)$. In the first case,
\begin{eqnarray*}
  p_{i}>\max_{j\neq i}p_{j}M_{ij}(\omega)=\max_{j\neq i}p_{j}\max_{1\leq \tilde{h}\leq H}\frac{\nabla_{i}u^{\tilde{h}}(\omega^{\tilde{h}})}{\nabla_{j}u^{\tilde{h}}(\omega^{\tilde{h}})}\geq \max_{j\neq i}p_{j}\frac{\nabla_{i}u^{h}(\omega^{h})}{\nabla_{j}u^{h}(\omega^{h})},
\end{eqnarray*}
for $1\leq h\leq H$, so that Definition \ref{defSharpUtility} implies
\begin{eqnarray}\label{eq1TradePricesInBox}
    x^{h}_{ni}\biggl(\frac{p}{p\cdot\omega^{h}}\biggr)-\omega^{h}_{i}<0,
\end{eqnarray}
for $1\leq h\leq H$. Also, $q\in\mathcal{T}_{L-1}(\omega)$ implies that there is $s\in\mathcal{S}(\omega,p)\subseteq[0,1]^{H}\setminus\{0\}$ and, due to Definition \ref{defJointLinearTradePath},
\begin{eqnarray}\label{eq2TradePricesInBox}
    \sum^{H}_{h=1}s_{h}\biggl(x^{h}_{n}\biggl(\frac{p}{p\cdot\omega^{h}}\biggr)-\omega^{h}\biggr)=0.
\end{eqnarray}
The $i$-th coordinate of (\ref{eq2TradePricesInBox}) is a sum of the terms $s_{h}(x^{h}_{ni}(p/p\cdot\omega^{h})-\omega^{h}_{i})$, $1\leq h\leq H$, which are non-positive due to $s_{h}\geq0$ and (\ref{eq1TradePricesInBox}). Therefore, $s_{h}=0$, $1\leq h\leq H$, which contradicts $s\neq0$. The same reasoning applies if $p_{i}<\min_{j\neq i}p_{j} m_{ij}(\omega)$, with the inequality in (\ref{eq1TradePricesInBox}) reversed. We conclude that $\mathcal{T}_{L-1}(\omega)\setminus\Pi(\omega)=\emptyset$ and, therefore, $\mathcal{T}_{L-1}(\omega)\subseteq \Pi(\omega)$.
\end{proof}

\begin{proof}[Proof of Lemma~{\upshape\ref{lemmaMonotoneMm}}]
To ease notation, let
\begin{eqnarray*}
    \rho^{h}_{ij}(t)=\frac{\nabla_{i}u^{h}(c^{h}(\omega^{h},p,s_{h},t))}{\nabla_{j}u^{h}(c^{h}(\omega^{h},p,s_{h},t))},
\end{eqnarray*}
for $1\leq h\leq H$, $1\leq i,j\leq L$, $i\neq j$, $t\in[0,1]$, so that $m_{ij}(c(\omega,p,s,t))=\min_{1\leq h\leq H}\rho^{h}_{ij}(t)$ and $M_{ij}(c(\omega,p,s,t))=\max_{1\leq h\leq H}\rho^{h}_{ij}(t)$, $t\in[0,1]$.

Fix $1\leq i,j\leq L$, $i\neq j$, and $1\leq h\leq H$. Since $u^{h}(\cdot)$ is attractive, Proposition \ref{propAttractionOfMarginalRates} implies that $\delta^{h}_{ij}(\cdot)=(\rho^{h}_{ij}(\cdot)-p_{i}/p_{j})^{2}$ is non-negative and non-increasing on $[0,1]$. I claim that $\rho^{h}_{ij}(\cdot)-p_{i}/p_{j}$ does not change sign on $[0,1]$. Suppose $\rho^{h}_{ij}(0)-p_{i}/p_{j}<0$ and $\rho^{h}_{ij}(t_{0})-p_{i}/p_{j}>0$, for some $t_{0}\in(0,1]$. Since $\rho^{h}_{ij}(\cdot)$ is continuous, there is $t_{1}\in(0,t_{0})$ with $\rho^{h}_{ij}(t_{1})=p_{i}/p_{j}$, so that $\delta^{h}_{ij}(t_{1})=0$ and, therefore, $\delta^{h}_{ij}(t_{0})=0$, absurd. The case $\rho^{h}_{ij}(0)-p_{i}/p_{j}>0$ is analogous.

Since $\vert \rho^{h}_{ij}(\cdot)-p_{i}/p_{j}\vert=\delta^{h}_{ij}(\cdot)^{1/2}$ is non-increasing on $[0,1]$, we conclude that if $h\in\mathcal{H}^{-}_{ij}(\omega,p)$, then $\rho^{h}_{ij}(\cdot)-p_{i}/p_{j}$ is non-positive and non-decreasing on $[0,1]$, while if $h\notin\mathcal{H}^{-}_{ij}(\omega,p)$, then $\rho^{h}_{ij}(\cdot)-p_{i}/p_{j}$ is non-negative and non-increasing on $[0,1]$.

If $\mathcal{H}^{-}_{ij}(\omega,p)\neq \emptyset$, the above implies
\begin{eqnarray}\label{eqMonotoneMmAux1}
    \min_{h\in \mathcal{H}^{-}_{ij}(\omega,p)}\rho^{h}_{ij}(t)\leq \frac{p_{i}}{p_{j}}\leq  \min_{h\in \{1,\ldots,H\}\setminus\mathcal{H}^{-}_{ij}(\omega,p)}\rho^{h}_{ij}(t),
\end{eqnarray}
for $t\in[0,1]$. Therefore,
\begin{eqnarray*}
     m_{ij}(c(\omega,p,s,t))=\min_{1\leq h\leq H}\rho^{h}_{ij}(t)=\min_{h\in\mathcal{H}^{-}_{ij}(\omega,p)}\rho^{h}_{ij}(t),
\end{eqnarray*}
for $t\in[0,1]$, and, since the minimum of finitely many non-decreasing functions is itself non-decreasing, $m_{ij}\circ c(\omega,p,s,\cdot)$ is non-decreasing.

If $\mathcal{H}^{-}_{ij}(\omega,p)=\emptyset$, then $m_{ij}\circ c(\omega,p,s,\cdot)$ is the minimum of finitely many non-increasing functions and, therefore, is itself non-increasing.

Finally, $\mathcal{H}^{+}_{ij}(\omega,p)=\mathcal{H}^{-}_{ji}(\omega,p)$ and $M_{ij}(c(\omega,p,s,t))=m_{ji}(c(\omega,p,s,t))^{-1}$, $1\leq i,j\leq L$, $i\neq j$, $t\in[0,1]$. Since $x\mapsto x^{-1}$ is strictly decreasing on $\mathbb{R}_{++}$, items (iii) and (iv) follow from applying items (i) and (ii), respectively, to the pair $(j,i)$.
\end{proof}

\begin{proof}[Proof of Theorem~{\upshape\ref{theoNested}}]
Since $u^{h}(\cdot)$ is sharp, $1\leq h\leq H$, Proposition \ref{propTradeUnfold} and Lemma \ref{lemmaTradePricesInBox} imply
\begin{eqnarray*}
    q\in\mathcal{T}_{L-1}(c(\omega,p,s,t))\subseteq \Pi(c(\omega,p,s,t))\cap \mathcal{F}(c(\omega,p,s,t)),
\end{eqnarray*}
for $0\leq t<1$. At $t=1$, Lemma \ref{lemmaLinearNet} implies $q\in\mathcal{F}(c(\omega,p,s,1))$. Furthermore, $m_{ij}\circ c(\omega,p,s,\cdot)$ and $M_{ij}\circ c(\omega,p,s,\cdot)$, $1\leq i,j\leq L$, $i\neq j$, are continuous on $[0,1]$, as minima and maxima of finitely many continuous functions, and $\Pi(\cdot)$ is defined through weak inequalities. Therefore, $\min_{j\neq i}p_{j}m_{ij}(c(\omega,p,s,t))\leq p_{i}\leq \max_{j\neq i}p_{j}M_{ij}(c(\omega,p,s,t))$, $1\leq i\leq L$, $0\leq t<1$, extends to $t=1$ by continuity and $q\in \Pi(c(\omega,p,s,1))$.

The identity
\begin{eqnarray*}
    \Pi(c(\omega,p,s,t))=\bigcap^{L}_{i=1}\bigcup_{j,k\neq i}\biggl(\Pi^{-}_{ij}(c(\omega,p,s,t))\bigcap \Pi^{+}_{ik}(c(\omega,p,s,t))\biggr),
\end{eqnarray*}
$0\leq t\leq1$, is the one preceding Lemma \ref{lemmaTradePricesInBox}, applied at the allocation $c(\omega,p,s,t)$. Notice that $\Pi^{-}_{ij}(\cdot)$ shrinks when $m_{ij}(\cdot)$ increases and grows when it decreases, while $\Pi^{+}_{ik}(\cdot)$ shrinks when $M_{ik}(\cdot)$ decreases and grows when it increases. Therefore, since $u^{h}(\cdot)$ is attractive, $1\leq h\leq H$, Lemma \ref{lemmaMonotoneMm} implies that $\{\Pi^{-}_{ij}(c(\omega,p,s,t))\}_{0\leq t\leq1}$ and $\{\Pi^{+}_{ik}(c(\omega,p,s,t))\}_{0\leq t\leq1}$ are either non-increasing or non-decreasing, $1\leq i,j,k\leq L$, $j,k\neq i$.

Finally, suppose $\mathcal{H}^{-}_{ij}(\omega,p)\neq\emptyset$, $1\leq i,j\leq L$, $i\neq j$. Since this holds for every ordered pair, $\mathcal{H}^{+}_{ij}(\omega,p)=\mathcal{H}^{-}_{ji}(\omega,p)\neq\emptyset$, $1\leq i,j\leq L$, $i\neq j$, as well. Therefore, Lemma \ref{lemmaMonotoneMm} implies that $m_{ij}\circ c(\omega,p,s,\cdot)$ is non-decreasing and $M_{ij}\circ c(\omega,p,s,\cdot)$ is non-increasing, $1\leq i,j\leq L$, $i\neq j$, so that $\{\Pi^{-}_{ij}(c(\omega,p,s,t))\}_{0\leq t\leq1}$ and $\{\Pi^{+}_{ik}(c(\omega,p,s,t))\}_{0\leq t\leq1}$ are non-increasing, $1\leq i,j,k\leq L$, $j,k\neq i$. Since both unions and intersections preserve inclusion, we conclude that $\{\Pi(c(\omega,p,s,t))\}_{0\leq t\leq1}$ is a non-increasing net.
\end{proof}


\printbibliography
\end{document}